\documentclass[prx,twocolumn,superscriptaddress,amssymb,floatfix]{revtex4-2}

\usepackage{amsmath}
\usepackage{amssymb}
\usepackage{amsfonts}
\usepackage{amsthm}
\usepackage{mathtools}
\usepackage{bm}
\usepackage{physics}

\usepackage{array}
\newcolumntype{L}[1]{>{\raggedright\arraybackslash}p{#1}}

\usepackage{booktabs}
\usepackage[table]{xcolor}
\definecolor{tablehighlightblue}{RGB}{225,239,255}
\usepackage{graphicx}
\usepackage{tabularx}

\makeatletter
\let\switch@array\relax
\makeatother

\newcolumntype{Y}{>{\raggedright\arraybackslash}X}

\usepackage[linesnumbered,ruled,vlined]{algorithm2e}
\SetKwInput{KwIn}{Input}
\SetKwInput{KwOut}{Output}
\SetKw{KwRet}{Return}
\DontPrintSemicolon
\SetAlgoCaptionSeparator{.}

\usepackage[colorlinks, citecolor=blue, linkcolor=blue]{hyperref}

\usepackage{stmaryrd}

\newcommand{\cH}{ {\cal H} }
\newcommand{\cL}{ {\cal L} }

\newcommand{\rd}{\partial}

\newcommand{\secref}[1]{Sec.\,\ref{#1}}
\newcommand{\appref}[1]{Appendix~\ref{#1}}

\newcommand{\eqnref}[1]{Eq.\,\eqref{#1}}

\theoremstyle{plain}
\newtheorem{theorem}{Theorem}
\newtheorem{corollary}{Corollary}[theorem]
\newtheorem{lemma}[theorem]{Lemma}
\newtheorem{proposition}[theorem]{Proposition}

\theoremstyle{definition}

\newcommand{\Z}{\mathbb{Z}}
\newcommand{\F}{\mathbb{F}}
\newcommand{\RG}{R_G}
\newcommand{\wt}{\operatorname{wt}}
\newcommand{\rankop}{\operatorname{rank}}
\newcommand{\coker}{\operatorname{coker}}

\newcommand{\transpose}{\mathsf{T}}
\newcommand{\LP}{\mathsf{LP}}
\newcommand{\HGP}{\mathsf{HGP}}
\newcommand{\row}{\mathsf{row}}
\newcommand{\Wcert}{W_{\rm full}^{\rm cert}}
\newcommand{\Wgen}{W_{\rm gen}^{\rm conj}}
\newcommand{\Dbasis}{D_{\rm basis}^{\rm ub}}
\newcommand{\Ddiag}{D_{\rm diag}^{\rm ub}}
\newcommand{\Dseed}{D_{\rm seed}^{\rm ub}}
\newcommand{\Dminor}{D_{\rm minor}^{\rm ub}}
\newcommand{\Dquot}{D_{\rm quot}^{\rm ub}}
\newcommand{\Dsyz}{D_{\rm syz}^{\rm ub}}
\newcommand{\Dwit}{D_{\rm wit}^{\rm ub}}
\newcommand{\Dbest}{D_{\rm best}^{\rm ub}}
\newcommand{\Drref}{D_{\rm RREF,min}^{\rm ub}}
\newcommand{\Wrref}{W_{\rm RREF,pair}^{\rm max}}
\newcommand{\Wrrefred}{W_{\rm RREF,pair}^{\rm max,red}}
\newcommand{\seedlabel}[4]{\ensuremath{\mathsf{#1}^{#2}_{#3,#4}}}

\graphicspath{{figures/}}

\begin{document}

\title{Logical Spectroscopy: Lifted-Product Codes with Addressable Bases}

\author{Jong Yeon Lee}
\email{jongyeon@illinois.edu}
\affiliation{Physics Department, University of Illinois at Urbana-Champaign, Urbana, Illinois 61801, USA}
\affiliation{Korea Institute for Advanced Study, Seoul 02455, South Korea}

\begin{abstract}
Quantum low-density parity-check (LDPC) memories can encode many logical qubits, but that alone does not make them usable: applications need to know where the logical operators are supported, how they are labeled, and how conjugate $X/Z$ partners pair. For hypergraph-product (HGP) codes this information follows from row reduction over $\mathbb{F}_2$. For Abelian lifted-product codes, which include prominent high-rate constructions, it does not: the entries of the defining seed matrices live in a group algebra rather than a field, so pivots need not be invertible and row reduction can fail.

To address this problem, we introduce \emph{logical spectroscopy}. For an odd-order Abelian lift group, the Chinese remainder theorem splits the group algebra into finite fields, one for each Frobenius orbit of characters; we call these orbits packets. Packet by packet, we solve ordinary finite-field linear algebra, lift the answers back to the physical code with the associated packet projectors, and pair $X$ and $Z$ logicals between reciprocal packets by trace duality. The result is a complete conjugate logical basis for finite Abelian lifted products $\mathsf{LP}(A,B)$ that is addressable: every logical coordinate carries canonical packet and K\"unneth-summand labels, deterministic within-block indices, a conjugate partner, and an explicit binary representative. The code's own translation symmetry then acts on these coordinates in closed form. The same computation supplies design diagnostics: the packet-resolved logical dimension, certified weight bounds for the constructed representatives, and an exact packet-by-packet account of which whole-orbit erasure patterns destroy logical information. We apply the construction to quasi-cyclic examples with up to roughly $5000$ physical qubits, including high-rate examples whose basis-width and distance-witness are reported explicitly. We further extend the packet decomposition and dimension formulas to even-order lifts, where nilpotent Smith profiles replace packet ranks; in fact, every even-lift complex is shown to be a square-zero thickening of its half-length complex. Logical spectroscopy thus equips Abelian lifted products with an explicit logical coordinate system and an algebraic toolkit for high-rate code search.
\end{abstract}

\maketitle

\section{Introduction}

Quantum error-correcting (QEC) codes are the basic substrate for fault-tolerant quantum computation. A QEC code is often characterized by parameters $[[n,k,d]]$, where $n$ is the number of physical qubits, $k$ is the number of encoded logical qubits, and $d$ is the minimum weight of a nontrivial logical Pauli operator~\cite{Gottesman1997}. Much of QEC code design seeks large $k$ and large $d$ at fixed or growing $n$~\cite{BreuckmannEberhardt2021}. For applications, the parameter $k$ must be supplemented by concrete logical data: where the logical operators are, how they are labeled, and how to choose conjugate $X/Z$ pairs whose supports can be measured or coupled to ancillary systems, for example in lattice-surgery, code-deformation, or qLDPC logical-measurement schemes~\cite{HorsmanFowlerDevittVanMeter2012,Litinski2019,VuillotLaoCrigerAlmudeverBertelsTerhal2018,CohenKimBartlettBrown2021,Xu2024,Williamson_2026,Swaroop2026,Cowtan2026,zheng2025highratesurgeryconstantoverheadlogical,BaspinBerentCohen2025,QGPU2026}. We preview the convention formalized below by calling such a conjugate basis \emph{addressable}: its packet-level labels are canonical, and its remaining coordinates are fixed by deterministic conventions rather than by arbitrary pivot or coordinate-order choices.

This addressability requirement already appears at the level of ordinary CSS codes~\cite{CalderbankShor1996,Steane1996}.  Such a code is specified by two binary check matrices $H_X$ and $H_Z$, whose rows generate the $X$- and $Z$-type stabilizer checks and which obey the commutation condition $H_X H_Z^\transpose = 0$ over $\mathbb F_2$. With this convention, the logical Pauli spaces are the quotient spaces $\cL_X=\ker H_Z/ \Im H_X$ and $\cL_Z=\ker H_X/ \Im H_Z$.
In practice, extracting a useful paired logical basis from a large sparse pair $(H_X,H_Z)$ can be difficult: generic binary Gaussian elimination gives some quotient basis, but the resulting labels are arbitrary, presentation-dependent, and often incompatible with the algebraic or geometric structure of the code. For hypergraph-product codes~\cite{TillichZemor2014}, the situation is comparatively transparent. The input classical codes $(A,B)$ are ordinary matrices over the field $\mathbb F_2$; Gaussian elimination exposes pivot coordinates, free coordinates, kernels, and quotient coordinates;
and the product construction turns these data into the familiar row/column grid of HGP logical operators~\footnote{This is evident from the fact that, for two classical input codes with parameters $[n_1,k_1,d_1]$ and $[n_2,k_2,d_2]$, the hypergraph product in the standard convention has $k = k_1 k_2 + k_1^\transpose k_2^\transpose$, where $k_i^\transpose=k_i+r_i-n_i$ is the number of logical bits of the transpose code defined by the transposed $r_i$-row check matrix; see \appref{app:reader-guide} for the convention used in this paper.}.
This structured addressability is one ingredient in recent parallel logical-computation schemes for homological-product codes~\cite{Xu2024,QGPU2026}. However, the distance of hypergraph-product codes scales only as the square root of the blocklength, short of asymptotically good quantum low-density parity-check (qLDPC) families.

\begin{figure*}[!t]
\includegraphics[width=0.99\textwidth]{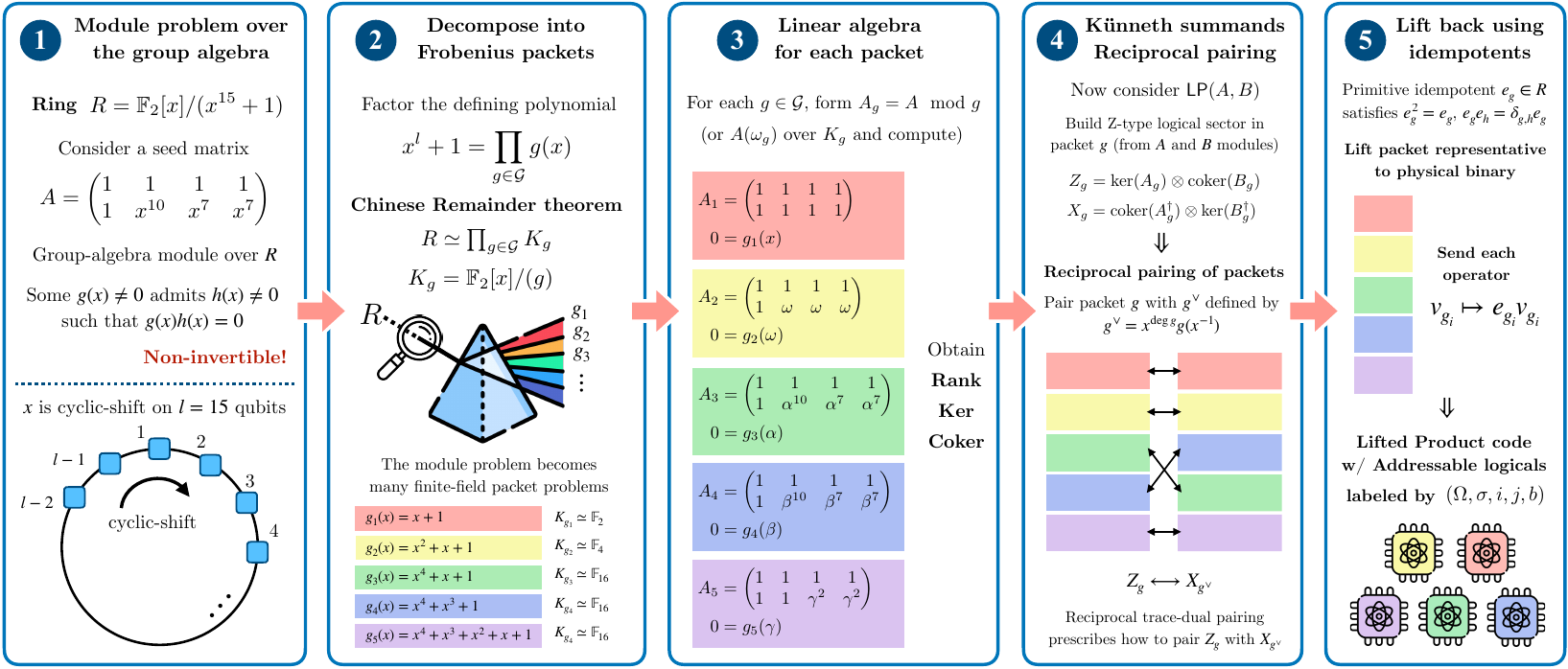}
\caption{ {\bf Logical spectroscopy.} (1) The toy example $\seedlabel{T}{15}{2}{4}$ from \appref{app:worked-example} is a $2\times4$ monomial matrix over $R=\F_2[\mathbb{Z}_{15}] = \F_2[x]/(x^{15}+1)$. For this example, HGP-style row reduction at the level of $A$ over $R$ is not available before binary expansion.
(2) By the Chinese remainder theorem, factoring $x^{15}+1$ decomposes $R$ into finite-field packets $K_g=\F_2[x]/(g)$.
(3) In each packet, the seed reduces to an ordinary finite-field matrix, so packet ranks, kernel and quotient coordinates, and K\"unneth representatives can be computed independently.
(4) Same-packet tensor data give the $Z$-logical packet contributions, while conjugate $X/Z$ labels are paired through reciprocal packets and trace-dual coefficient bases because the binary commutation pairing sends $x$ to $x^{-1}$.
(5) The primitive idempotents $e_g$ then lift packet representatives to physical binary logical representatives, preserving packet labels and giving an addressable logical basis.
The color coding follows a fixed packet throughout the diagram, emphasizing that the construction factors, solves, pairs, and lifts packet by packet. }
\label{fig:microscope}
\end{figure*}

By contrast, lifted-product (LP) codes can achieve high-rate and high-distance qLDPC families~\cite{PanteleevKalachev2022,PanteleevKalachev2022STOC}. At first sight the construction is similar, again using two classical seed checks. The linear algebra, however, takes place over a different object. In the single-variable quasi-cyclic setting, the seed entries lie in
\begin{equation}
        R_\ell=\F_2[x]/(x^\ell+1).
\end{equation}
The symbol $x$ represents a one-step cyclic shift; after expansion, a monomial $x^a$ becomes an $\ell\times \ell$ circulant permutation block.  A sparse polynomial such as $x^a+x^b$ becomes the sum of two such blocks.  A compact polynomial seed therefore describes a large binary block-circulant check matrix.  The resulting binary code is quasi-cyclic rather than cyclic: its coordinates split into many cyclic block-orbits, and the simultaneous shift inside all orbits is a code automorphism.
The difficulty is that $R_\ell$ is not a field: a nonzero determinant need not be a unit. Even if a candidate square pivot block has a nonzero determinant as a polynomial class, it may still fail to be invertible over $R_\ell$ and cannot necessarily be used as a Gaussian-elimination pivot.
In principle, one can still obtain logical representatives by applying binary Gaussian elimination and symplectic Gram--Schmidt~\cite{Wilde2009} to the expanded CSS code.  When used in this direct binary form, however, the labels are tied to pivot choices in the expanded coordinate system rather than to the quasi-cyclic or finite Abelian structure.
The resulting representatives can depend strongly on coordinate order and need not be organized by the lift symmetry, making it difficult to design reusable measurement or addressing patterns.
This structural mismatch means that the efficient logical-operation protocols of Ref.~\cite{Xu2024}, which rely on HGP coordinates, do not carry over directly to LP codes.

In this paper, we develop a spectral replacement for global row reduction.  The Chinese remainder theorem (CRT) says that, when the lift group has odd order, the group algebra splits into a product of finite fields.  For $R_\ell$, this amounts to factoring $x^\ell+1$ over $\F_2$.  Each irreducible factor determines a \emph{Frobenius packet}, or simply a packet, of cyclic-shift modes; as orbits of characters under the Frobenius squaring map, packets are the Frobenius orbits, or $2$-cyclotomic cosets, of classical coding theory~\cite{LingSole2001,LidlNiederreiter1997}.  These packets are the finite-field analogue of Fourier or character modes, but we use the single term \emph{packet} after this point.  We reduce the seed matrix in each packet, solve ordinary finite-field kernel and quotient problems there, and then lift the answers back to the physical quasi-cyclic code using the primitive idempotents attached to the packets.  These idempotents act as packet projectors supplied by the Chinese remainder decomposition: after substituting $x$ by the cyclic-shift matrix, they become matrix projectors onto the corresponding shift-invariant spectral subspaces.

\begin{table*}[!t]
\scriptsize
\setlength{\tabcolsep}{2.5pt}
\begin{tabular}{lcccccccccccc}
\toprule
Seed & shape & $\ell$ & $\quad \llbracket n,k,d\rrbracket \quad$ & $k/n$ & $w_{\rm check}$  & $\Ddiag$ & $\Dwit$ & $\Dbasis$ & $\Wcert$ & $\Drref$ & $\Wrref$ & $\Wrrefred$ \\
\midrule
\seedlabel{P}{15}{2}{4} & $2\times4$ & 15 & $\llbracket 300,60,9\rrbracket$ & 0.200 & 11 & 17 & 9 & 15 & 24 & 44 & 102 & 81\\
\seedlabel{E}{14}{3}{4} & $3\times4$ & 14 & $\llbracket 350,26,\leq 14\rrbracket$ & 0.074 & 7 & 14 & 14 & 24 & 168 & 35 & 101 & 75\\
\midrule
\seedlabel{M}{45}{4}{10} & $4\times10$ & 45 & $\llbracket 5220,1674,\leq 16\rrbracket$ & 0.321 & 14 &  18 & 16 & 60 & 126 & 72 & 1672 & 1422 \\
\seedlabel{R}{75}{3}{7} & $3\times7$ & 75 & $\llbracket 4350,1224,\leq 16\rrbracket$ & 0.281 & 10 & 22 & 16 & 90 & 160 & 96 & 1462 & 1242\\
\seedlabel{R}{81}{3}{7} & $3\times7$ & 81 & $\llbracket 4698,1320,\leq 16\rrbracket$ & 0.281 & 10 & 18 & 16 & 60 & 162 & 102 & 1602 & 1314\\
\rowcolor{tablehighlightblue}
\seedlabel{M}{45}{5}{9} & $5\times9$ & 45 & $\llbracket 4770,784,\leq 50\rrbracket$  & 0.164 & 14 & 50 & 86 & 90 & 150 & 110 & 1476 & 1296\\
\rowcolor{tablehighlightblue}
\seedlabel{P}{75}{3}{7} & $3\times7$ & 75 & $\llbracket 4350,1200,\leq 25\rrbracket$ & 0.276 & 18/19 & 25 & 25 & 50 & 180 & 89 & 1109 & 972\\
\rowcolor{tablehighlightblue}
\seedlabel{P}{91}{3}{7} & $3\times7$ & 91 & $\llbracket 5278,1456,\leq 70\rrbracket$ & 0.276 & 18/19  & 70 & -- & 91 & 208 & 123 & 1515 & 1346\\
\midrule
\seedlabel{H}{41}{3}{8} & $3\times8$ & 41 & $\llbracket 2993,1025,\leq 26\rrbracket$ & 0.342 & 21 & 80 & 26 & 41 & 84 & 45 & 884 & 703\\
\seedlabel{H}{33}{3}{9} & $3\times9$ & 33 & $\llbracket 2970,1188,\leq 18\rrbracket$ & 0.400 & 23 & 62 & 18 & 33 & 74 & 39 & 656 & 563\\
\midrule
\seedlabel{R}{21}{4}{5} & $4\times5$ & 21 & $\llbracket 861,45,\leq 26\rrbracket$ & 0.052 & 9 & 58 & 26 & 42 & 60 & 174 & 248 & 202\\
\seedlabel{P}{35}{3}{4} & $3\times4$ & 35 & $\llbracket 875,35,\leq 35\rrbracket$ & 0.040 & 13 & 75 & 35 & 35 & 98 & 188 & 258 & 237\\
\seedlabel{P}{31}{2}{5} & $2\times5$ & 31 & $\llbracket 899,279,\leq 8\rrbracket$ & 0.310 & 13 & 10 & 8 & 31 & 48 & 23 & 258 & 216\\
\midrule
\seedlabel{A}{45}{3}{7} & $3\times7$ & 45 & $\llbracket 2610,744,\leq 16\rrbracket$ & 0.285 & 10 & 18 & 16 & 36 & 120 & 58 & 794 & 670\\
\seedlabel{A}{75}{3}{7} & $3\times7$ & 75 & $\llbracket 4350,1224,\leq 20\rrbracket$ & 0.281 & 10 & 24 & 20 & 90 & 200  & 92 & 1474 & 1234\\
\seedlabel{A}{91}{3}{7} & $3\times7$ & 91 & $\llbracket 5278,1480,\leq 22\rrbracket$ & 0.280 & 10 & 24 & 22 & 112 & 224 & 122 & 1582 & 1316 \\
\bottomrule
\end{tabular}
\caption{{\bf Self-adjoint construction examples.} Seed matrices are given in \appref{app:seed-data}; seed-label conventions and distance-witness definitions are given in \secref{sec:applications}, and the certificate notation is summarized in Table~\ref{tab:certificate-notation}. The code-parameter column is written as $\llbracket n,k,d\rrbracket$ using the exact distance when available and otherwise the tightest reported verified upper bound $\Dbest$. Entries such as $18/19$ in $w_{\rm check}$ report the two CSS check sides when they differ; $\Dwit$ entries marked ``--'' mean that no completed independent distance witness is reported. The even-order row \seedlabel{E}{14}{3}{4} is included as a small exponent-lift example: its dimension and $\Ddiag=14$ primary-packet witness use the nonsemisimple tools of Theorem~\ref{thm:beyond-odd} and Appendix~\ref{app:even-lifts}; the matching $\Dwit=14$ entry is a completed randomized QDist witness.  Its $\Dbasis$ and $\Wcert$ entries come from a primary-packet-projected binary quotient basis paired by finite binary Gram inversion, while its RREF columns are computed from the direct binary expansion as for the odd rows.
The highlighted rows are \seedlabel{M}{45}{5}{9}, \seedlabel{P}{75}{3}{7}, and \seedlabel{P}{91}{3}{7}; their Cramer/minor and quotient-lift diagnostics are analyzed in Appendices~\ref{app:quotient-lift-witnesses} and \ref{app:cramer-spectra}.
The $\mathsf A$ rows are from Ref.~\cite{Cain2026}. Several $\Dwit$ entries use the fast distance-sampling method from Ref.~\cite{ZhouMaskara2026}; the \seedlabel{P}{75}{3}{7} entry is a verified mixed-packet packet-sum witness.
}
\label{tab:main}
\end{table*}

The spectral method also explains how conjugate logicals pair.  In the tensor product over $R_G$, only diagonal packet components survive: an elementary tensor with packets $\Omega$ and $\Lambda$ vanishes unless $\Omega=\Lambda$.  A general product representative is therefore a sum of same-packet contributions.  By contrast, the binary $X/Z$ commutation pairing transposes cyclic shifts, which sends $x$ to $x^{-1}$.  A $Z$ representative in one packet pairs with an $X$ representative in the reciprocal packet, so same-packet tensoring and reciprocal-packet pairing play different roles.

Beyond the basis itself, the packet computation yields rank-defect contributions to $k$, group-basis supports of lifted idempotent representatives, and exact whole-orbit erasure ranks with packet attribution.  Figure~\ref{fig:microscope} summarizes this packetwise pipeline.

This addressability viewpoint is especially relevant for high-rate LP memories, where the number of encoded qubits can be large enough that arbitrary quotient-basis labels are not operationally useful for the structured logical operations envisioned in recent work~\cite{CohenKimBartlettBrown2021,Williamson_2026,Swaroop2026,Xu2024,Cowtan2026,zheng2025highratesurgeryconstantoverheadlogical,BaspinBerentCohen2025,QGPU2026}.

\vspace{5pt}

\textbf{Relation to prior work.}
CRT and constituent decompositions are classical tools in the theory of quasi-cyclic codes~\cite{LingSole2001,GuneriLingOzkaya2020}; repeated-root and chain-ring constituents cover the even case~\cite{vanLint1991,DinhLopez2004,LingSole2003}.  Circulant-permutation, exponent-matrix, and permanent/minor methods are standard in QC-LDPC code design~\cite{Fossorier2004,SmarandacheVontobel2004,SmarandacheVontobel2012}.  On the quantum side, K\"unneth-type decompositions underlie hypergraph-product, homological-product, lifted-product, fiber-bundle, balanced-product, non-Abelian lifted-product, and quantum Tanner codes~\cite{BravyiHastings2013,PanteleevKalachev2021,HastingsHaahODonnell2021,BreuckmannEberhardtBalanced2021,PanteleevKalachev2022,LeverrierZemor2022}.

The present contribution is the logical-coordinate structure obtained by keeping these ingredients together for rectangular Abelian $\LP(A,B)$ codes.  The three pieces used throughout the paper are: a reciprocal-packet trace-dual construction of conjugate $X/Z$ bases; a block-diagonal translation action on the stored representatives (Corollary~\ref{cor:equivariance}); and exact packet attribution for whole-orbit erased-logical dimensions.  These outputs supplement the packetwise dimension count with labeled representatives, conjugate pairing data, basis-width data for the constructed representatives, and erasure attribution.

Structured cyclic and group-algebra quantum LDPC families include hyperbicycle, two-block group-algebra, generalized-bicycle, bivariate-bicycle, multivariate-bicycle, univariate-bicycle, and radial codes~\cite{KovalevPryadko2013,LinPryadko2024,BravyiCross2024,ScrubyHillmannRoffe2026,Voss2024,EberhardtSteffan2024,RabetiMahdavifar2026,DavenportBlueChuang2026}.  These works analyze important subclasses, including explicit logical operators, automorphisms, distance upper bounds, or transversal gates.  Radial codes~\cite{ScrubyHillmannRoffe2026} give a canonical basis with geometric code-copy and ring labels; the univariate-bicycle analysis of Ref.~\cite{RabetiMahdavifar2026} describes logical cosets and distance upper bounds for a structured generalized-bicycle subclass; and Ref.~\cite{EberhardtSteffan2024} constructs explicit logical operators and fold-transversal gates for bivariate-bicycle codes.  The framework here covers rectangular seeds over arbitrary odd-order finite Abelian lift groups, with the even-order extensions summarized in Appendix~\ref{app:even-lifts}.  Generalized surgery, code-deformation, and qLDPC logical-measurement schemes motivate the addressability requirement~\cite{HorsmanFowlerDevittVanMeter2012,Litinski2019,VuillotLaoCrigerAlmudeverBertelsTerhal2018,CohenKimBartlettBrown2021,Cowtan2026,zheng2025highratesurgeryconstantoverheadlogical,BaspinBerentCohen2025,QGPU2026}; the construction below supplies the packet-labeled coordinates on which such protocol layers can build.  This viewpoint is also complementary to recent work on cohomological invariants and logical gates in qLDPC codes~\cite{LiShao2026}.

\vspace{5pt} \noindent The paper is organized as follows. \secref{sec:pivots-to-packets} identifies the obstruction to global HGP-style row reduction and introduces Frobenius packets, primitive idempotents, and the reciprocal pairing rule. \secref{sec:finite-abelian-basis} proves the finite Abelian addressable-basis theorem. \secref{sec:representatives-certificates} turns the theorem into explicit representatives, basis-width certificates, and verified distance witnesses. \secref{sec:applications} applies the construction to quasi-cyclic construction examples, including high-rate codes, and to a noncyclic two-variable example, and proves the whole-orbit erasure decomposition. \secref{sec:discussion} concludes with rank-profile design, certificate-search diagnostics, neutral-atom motivation, and open directions.

\begin{table*}[t]
\renewcommand{\arraystretch}{1.12}
\begin{tabularx}{\textwidth}{L{0.2\textwidth}L{0.28\textwidth}Y}
\toprule
Method & What it gives & Addressability structure \\
\midrule
Binary Gaussian elimination & Logical quotient representatives of the expanded binary CSS code. & Valid representatives with supports; labels depend on binary coordinate pivots rather than packets. \\ \midrule
Symplectic Gram--Schmidt~\cite{Wilde2009,Williamson_2026} & A conjugate binary logical basis. & Valid $X/Z$ pairing after binary-level processing; packet and K\"unneth labels are not identified by itself. \\ \midrule
This work & Packet-labeled LP logical representatives. & Packet and K\"unneth-summand labels, reciprocal-packet partners with trace-dual coefficient bases, and idempotent lifts to binary representatives. \\
\bottomrule
\end{tabularx}
\caption{{\bf Logical-coordinate information exposed by different constructions.} Generic binary methods correctly produce logical representatives, their supports, and, after pairing, conjugate bases.  The distinction here is the coordinate system: the spectral construction attaches packet labels, K\"unneth-summand labels, and reciprocal-packet partners with trace-dual coefficient bases to its representatives.}
\label{tab:logical-coordinate-comparison}
\end{table*}

\section{Spectral Decomposition}
\label{sec:pivots-to-packets}

This section isolates the algebraic obstruction that prevents a direct HGP-style logical-basis construction for lifted products, and then introduces the spectral coordinates that replace it.  Appendix~\ref{app:reader-guide} reviews hypergraph-product and lifted-product codes, together with the mathematical tools used below; here we keep only the main mechanism.

\subsection{The pivot picture and its failure over a ring}

For an ordinary binary input matrix, the HGP logical basis is built from row reduction over a field.  If a full-row-rank matrix over $\F_2$ has $r$ rows, one chooses $r$ pivot columns, solves for the pivot coordinates in terms of the free coordinates, and obtains kernel vectors and quotient-complement coordinates.  The product construction then tensors these single-factor data into a grid of logical labels.

A direct lifted-product analogue would require one global pivot block over a group algebra.  In the single-variable quasi-cyclic setting this ring is $R_\ell=\F_2[\Z_\ell]=\F_2[x]/(x^\ell+1)$.
Suppose $A\in R_\ell^{r\times n}$ and one chooses a set $P$ of $r$ seed columns.  Writing the corresponding square block as $A_P$, a global systematic form would need
\begin{align}
A=\begin{bmatrix}A_P&A_F\end{bmatrix}, \quad
A_P^{-1}A=\begin{bmatrix}I&A_P^{-1}A_F\end{bmatrix}.
\end{align}
The issue is not whether $A_P$ has a nonzero determinant as a polynomial class, but whether that determinant is a unit of $R_\ell$.  When $\ell$ is odd, $x^\ell+1$ is square-free and the Chinese remainder theorem (CRT) decomposes $R_\ell$ into finite fields:
\begin{align}
x^\ell+1=\prod_{g\in\mathcal G_\ell}g(x), \quad
R_\ell\simeq\prod_{g\in\mathcal G_\ell}K_g,
\end{align}
where $K_g=\F_2[x]/(g)$ is a field. Accordingly, $A_P$ is invertible over $R_\ell$ if and only if $\det A_P$ is nonzero in every $K_g$.
Equivalently, in the single-variable quasi-cyclic setting, $\gcd(\det A_P,x^\ell+1)=1$.

Dense monomial seeds fail this global-pivot test in the simplest packet.  Reducing modulo $x+1$ is the same as setting $x=1$, so every monomial entry $x^a$ becomes $1$.  A dense monomial $r\times n$ seed reduces in the $x+1$ packet to the all-ones matrix $J_{r\times n}$.  For $r\geq2$, every $r\times r$ minor of this matrix vanishes.  Every candidate pivot determinant is zero in the $x+1$ packet and is not a unit of $R_\ell$.  The obstruction is the absence of a single pivot block that works simultaneously in all packets.

This is a coordinate obstruction.  It means that the logical degrees of freedom are distributed across packets, and that one should not try to expose them by forcing one global module row reduction.  The replacement is to decompose the lift ring first and do ordinary finite-field linear algebra packet by packet.

\subsection{Spectral decomposition}

We now define the CRT/Frobenius packets used for finite Abelian lifts.  Let
\begin{align} \label{eq:group_def}
G=\mathbb Z_{\ell_1}\times\cdots\times\mathbb Z_{\ell_m}, \quad
|G|\text{ odd}, \quad
R_G=\F_2[G].
\end{align}
The odd-order assumption implies that $\operatorname{char}\F_2$ does not divide $|G|$, so $R_G$ is semisimple.  Since $G$ is Abelian, this semisimple algebra is a product of finite fields~\cite{DummitFoote2004,LidlNiederreiter1997}.

Let $\overline{\F}_2$ denote an algebraic closure of $\F_2$, and let
\begin{align}
\widehat G(\overline{\F}_2)=\operatorname{Hom}(G,\overline{\F}_2^\times)
\end{align}
be the set of character modes, equivalently the Fourier modes of the regular $G$-action after extending scalars to $\overline{\F}_2$.

For $G$ as in \eqnref{eq:group_def}, a character may be written as a tuple
\begin{align}
\omega=(\zeta_1,\ldots,\zeta_m),\qquad \zeta_i^{\ell_i}=1.
\end{align}
The Frobenius map acts by squaring character values:
\begin{align}
\operatorname{Fr}(\omega)=\omega^2=(\zeta_1^2,\ldots,\zeta_m^2).
\end{align}
A Frobenius packet, or simply a packet, is an orbit
\begin{align}
\Omega=\{\omega,\omega^2,\omega^4,\ldots\}.
\end{align}
Equivalently, a packet is one $\F_2$-rational collection of Fourier modes: individual characters may live over extension fields, but binary coefficients force all Frobenius conjugates to appear together.  In classical coding terminology, a packet is a $2$-cyclotomic coset of characters.

For a packet $\Omega$, let $K_\Omega$ be the finite field generated by the coordinates of any representative character in $\Omega$, and let $e_\Omega\in R_G$ be the corresponding primitive central idempotent attached to that packet (see \appref{app:idempotent}).  The semisimple CRT decomposition can be written as
\begin{align}
R_G\simeq\prod_\Omega K_\Omega, \quad
1=\sum_\Omega e_\Omega, \quad
e_\Omega e_\Lambda=\delta_{\Omega\Lambda}e_\Omega.
\end{align}
Multiplication by $e_\Omega$ is the algebraic projector onto the $\Omega$ packet subspace.  These projectors are the coordinates in which the global ring obstruction from the previous subsection becomes ordinary finite-field linear algebra.

\subsection{Packet evaluation}

Projecting a seed matrix to a packet turns the group-algebra problem into ordinary linear algebra over a finite field.  For a seed map
\begin{align}
A:R_G^{n_A}\to R_G^{r_A},
\end{align}
the $\Omega$ packet is the map
\begin{align}
A_\Omega:K_\Omega^{n_A}\to K_\Omega^{r_A}
\end{align}
obtained by evaluating every group-algebra entry of $A$ on a representative character in $\Omega$.  Choosing another representative in the same Frobenius orbit gives a Frobenius-conjugate matrix, hence the same ranks, kernel and quotient dimensions, and binary packet subspace after lifting.

Row reduction is now performed over the field $K_\Omega$.  Every nonzero pivot is invertible inside the packet.  Different packets may use different pivot columns, kernel bases, and quotient complements.  This packet-by-packet freedom is precisely what is unavailable in a single global systematic form over $R_G$.

At the module level, the seed map restricts as
\begin{align}
A:e_\Omega R_G^{n_A}\to e_\Omega R_G^{r_A},
\end{align}
and, after the identification $e_\Omega R_G\simeq K_\Omega$, this restricted map is exactly $A_\Omega$.  Packet evaluation is the finite-field coordinate description of the original binary block-circulant map restricted to one invariant packet subspace.

\subsection{Two packet rules}
\label{subsec:two-packet-rules}

The construction uses two packet rules.  The first governs product representatives, and the second governs conjugate $X/Z$ pairing.

\begin{lemma}[Same-packet tensor rule]
\label{lem:same-packet-tensor}
If $M$ and $N$ are $\RG$-modules, then
\begin{align}
(e_\Omega M)\otimes_{\RG}(e_\Lambda N)=0
\end{align}
whenever $\Omega\neq\Lambda$.  Only diagonal packet components survive in the tensor product over $\RG$: an elementary tensor contribution must use the same packet on both tensor legs, while a general product representative may be a sum over such same-packet contributions.
\end{lemma}

\begin{proof}
For $m\in e_\Omega M$ and $n\in e_\Lambda N$,
\begin{align}
m\otimes n
=
m\otimes e_\Omega n
=
m\otimes e_\Omega e_\Lambda n
=
0.
\end{align}
\end{proof}

\begin{lemma}[Reciprocal binary pairing]
\label{lem:reciprocal-binary-pairing}
Let $\Omega^\vee$ denote the packet obtained by inverting the characters in $\Omega$.  If
\begin{align}
a\in e_\Omega \RG^m,
\quad
b\in e_\Lambda \RG^m,
\end{align}
then the physical binary pairing between $a$ and $b$ can be nonzero only when $\Lambda=\Omega^\vee$.
\end{lemma}

\begin{proof}
The physical CSS commutation pairing is the binary dot product.  In one cyclic coordinate it can be written as
\begin{align}
\langle a,b\rangle_{\operatorname{bin}}
=
[x^0]\sum_i a_i(x)b_i(x^{-1}).
\end{align}
Here $[x^0]$ is coefficient-extraction notation: if $f(x)\in R_\ell$ is represented as $f(x)=c_0+c_1x+\cdots+c_{\ell-1}x^{\ell-1}$, then $[x^0]f(x)=c_0$.  The sum is over the block coordinates of the vector, i.e. if $a,b\in R_\ell^m$, then $i=1,\ldots,m$.  This gives the binary dot product because the constant term of $a_i(x)b_i(x^{-1})$ is $\sum_t a_{i,t}b_{i,t}$.
The inverse appears because transposing a cyclic shift sends $x$ to $x^{-1}$.  For a general Abelian group, this is the group-inversion involution $u\mapsto u^{-1}$.

Algebraically, the involution sends $e_\Lambda$ to $e_{\Lambda^\vee}$, and the pairing contains the product $e_\Omega e_{\Lambda^\vee}$.  This product vanishes unless $\Omega=\Lambda^\vee$, equivalently $\Lambda=\Omega^\vee$.
\end{proof}

In the single-variable quasi-cyclic case, $\Omega^\vee$ corresponds to the reciprocal factor $g^\vee$, whose roots are the inverses of the roots of $g$.  Elementary tensor contributions are same-packet, a general representative may sum over packets, and $X/Z$ pairing is reciprocal-packet.

Applying ordinary finite-field homology in each packet, forming same-packet tensor contributions, lifting representatives back with primitive idempotents, and pairing reciprocal packets by trace duality assembles the basis theorem of the next section: the spectral construction restores the HGP kernel-and-quotient picture after the group algebra has been resolved into Frobenius packets.

\section{Finite Abelian basis theorem}
\label{sec:finite-abelian-basis}

We now state the main construction for the lifted product codes $\LP(A,B)$.
The single-variable quasi-cyclic formulas used in the examples will be treated later as an implementation layer. Throughout this section, we work in the setting of \eqnref{eq:group_def} where $G$ is an odd-order Abelian group.
The packets $\Omega$, fields $K_\Omega$, idempotents $e_\Omega$, and reciprocal packets $\Omega^\vee$ are those of Sec.~\ref{sec:pivots-to-packets}.

\subsection{The product complex}

We refer to \appref{app:LP_code} for the review of lifted-product codes.  Let
\begin{align}
A:\RG^{n_A}\to\RG^{r_A},
\quad
B:\RG^{n_B}\to\RG^{r_B}
\end{align}
be two seed maps.  Let $\bar{\cdot}:\RG\to\RG$ denote the group-inversion involution, so $\overline{\sum_g a_g g}=\sum_g a_g g^{-1}$ and, in the single-variable quasi-cyclic case, $\overline{p(x)}=p(x^{-1})$.  For a matrix $M$ over $\RG$, define the group-algebra adjoint by
\begin{align}
M^\dagger=\overline{M}^{\transpose},
\end{align}
where the bar is applied entrywise.  Here $\transpose$ denotes the ordinary transpose of the displayed module indices.  If $\rho$ denotes regular binary expansion, then $\rho(M^\dagger)=\rho(M)^\transpose$.  The symbol $\dagger$ is the physical adjoint after binary expansion, whereas $\transpose$ alone is only an index transpose unless the matrix is already binary.  The lifted-product code $\LP(A,B)$ is defined from the product complex
\begin{align}
Q_2
\xrightarrow{\partial_2}
Q_1
\xrightarrow{\partial_1}
Q_0,
\end{align}
where
\begin{align}
Q_2&=\RG^{n_A}\otimes_{\RG}\RG^{n_B},\\
Q_1&=(\RG^{n_A}\otimes_{\RG}\RG^{r_B})
\oplus
(\RG^{r_A}\otimes_{\RG}\RG^{n_B}),\\
Q_0&=\RG^{r_A}\otimes_{\RG}\RG^{r_B}.
\end{align}
The boundary maps $\rd_{1,2}$ are given in \eqnref{eq:boundary_map}.
After binary expansion, the number of physical qubits is
\begin{align}
n=|G|(n_A r_B+r_A n_B).
\end{align}
We take $H_X$ from the binary expansion of $\partial_1$, and $H_Z$ from the binary transpose of the expansion of $\partial_2$, equivalently from the expansion of $\partial_2^\dagger$.  The $Z$-logical space is the degree-one homology
\begin{align}
\cL_Z=\cH_1=\ker\partial_1/\Im\partial_2,
\end{align}
while the $X$-logical space is the dual cohomology
\begin{align}
\cL_X=\cH^1=\ker\partial_2^\dagger/\Im\partial_1^\dagger.
\end{align}

\noindent\textbf{Addressability convention.}
In this paper, a conjugate logical basis is addressable if its labels are reproducible, meaning determined by the construction and by fixed conventions rather than by arbitrary pivot or coordinate-order choices; each label has explicit stabilizer-commuting binary $X$- and $Z$-representatives, and the final binary Gram matrix is verified:
\begin{align}
Z_\lambda X_\mu^\transpose=\delta_{\lambda\mu}.
\end{align}
Here $Z_\lambda$ and $X_\mu$ denote binary support vectors for the corresponding $Z$- and $X$-type logical Pauli representatives, with Pauli phases suppressed; the displayed product is the mod-two commutation pairing.
The spectral construction below gives these labels from the packet index, finite-field kernel and quotient coordinates, and a coefficient-basis coordinate.

\subsection{Packet homology}

Projecting the product complex to a packet $\Omega$ identifies $e_\Omega\RG$ with the finite field $K_\Omega$.  The seed maps become ordinary finite-field matrices
\begin{align}
A_\Omega:K_\Omega^{n_A}\to K_\Omega^{r_A}, \quad
B_\Omega:K_\Omega^{n_B}\to K_\Omega^{r_B}.
\end{align}
Set
\begin{align}
s_A^\Omega=\rankop_{K_\Omega}A_\Omega,\quad
s_B^\Omega=\rankop_{K_\Omega}B_\Omega,\quad
d_\Omega=[K_\Omega:\F_2].
\end{align}
The packet product complex is an ordinary finite-field product complex.  Its degree-one homology therefore has the same K\"unneth decomposition as a hypergraph-product complex over the field $K_\Omega$:
\begin{align}
\cH_{1,\Omega} \simeq & \ker A_\Omega\otimes_{K_\Omega}\coker B_\Omega \nonumber \\
& \oplus  \coker A_\Omega\otimes_{K_\Omega}\ker B_\Omega.
\label{eq:packet-homology}
\end{align}
Each packet contribution has finite-field dimension
\begin{align}
\dim_{K_\Omega}\cH_{1,\Omega}
&=
(n_A-s_A^\Omega)(r_B-s_B^\Omega) \nonumber \\
&+
(r_A-s_A^\Omega)(n_B-s_B^\Omega).
\label{eq:packet-dimension}
\end{align}

\begin{lemma}[Reciprocal trace pairing]
\label{lem:reciprocal-trace-pairing}
Let $\tau_\Omega:K_{\Omega^\vee}\to K_\Omega$ be the field isomorphism induced by group inversion on characters.  Define
\begin{align}
\langle a,b\rangle_{\Omega,\operatorname{tr}}
=
\operatorname{Tr}_{K_\Omega/\F_2}\bigl(a\,\tau_\Omega(b)\bigr),
\quad
a\in K_\Omega,\quad b\in K_{\Omega^\vee}.
\end{align}
This pairing is nondegenerate over $\F_2$.  For any $\F_2$-basis $\{\beta_a\}$ of $K_\Omega$, there is a unique trace-dual basis $\{\beta_a^\#\}$ in the reciprocal packet field $K_{\Omega^\vee}$ satisfying
\begin{align}
\operatorname{Tr}_{K_\Omega/\F_2}\bigl(\beta_a\,\tau_\Omega(\beta_b^\#)\bigr)
=
\delta_{ab}.
\end{align}
Equivalently, this trace-dual basis in $K_{\Omega^\vee}$ represents the dual coordinate functionals on $K_\Omega$ under the physical binary pairing.
\end{lemma}

\begin{proof}
Finite fields are separable over $\F_2$, so the ordinary trace form $(u,v)\mapsto\operatorname{Tr}_{K_\Omega/\F_2}(uv)$ on $K_\Omega$ is nondegenerate~\cite{LidlNiederreiter1997}.  Composing the second argument with the field isomorphism $\tau_\Omega$ gives a nondegenerate pairing between $K_\Omega$ and $K_{\Omega^\vee}$.  The trace-dual basis is then the ordinary dual basis for this nondegenerate bilinear form.
\end{proof}

\noindent\textbf{Example.}
Consider the self-reciprocal packet $g=x^2+x+1$ in $R_3=\F_2[x]/(x^3+1)$, and write $K_g=\F_2[\alpha]/(\alpha^2+\alpha+1)$ with $\alpha=x\bmod g$.  The reciprocal map $\tau_g$ sends $\alpha$ to $\alpha^{-1}=\alpha^2$, and the field trace is $\operatorname{Tr}_{K_g/\F_2}(z)=z+z^2$.  For the basis $\{\beta_1,\beta_2\}=\{1,\alpha\}$, the trace-dual basis in the reciprocal packet is $\{\beta_1^\#,\beta_2^\#\}=\{\alpha,1\}$, since
\begin{align}
\operatorname{Tr}_{K_g/\F_2}\bigl(\beta_a\,\tau_g(\beta_b^\#)\bigr)
=
\delta_{ab}.
\end{align}
This is the coefficient-basis operation used below when a representative labeled by $\alpha$ is paired with a conjugate representative labeled by $\alpha^\#$ in the reciprocal packet.

\begin{lemma}[Packet homology--cohomology pairing]
\label{lem:packet-homology-cohomology-pairing}
For each packet $\Omega$, the physical binary pairing induces a perfect pairing between the homology packet $\cH_{1,\Omega}$ and the reciprocal cohomology packet $\cH^1_{\Omega^\vee}$.  More explicitly, if
$A^\dagger_{\Omega^\vee}:=(A^\dagger)_{\Omega^\vee}$ and
$B^\dagger_{\Omega^\vee}:=(B^\dagger)_{\Omega^\vee}$, then
\begin{align}
\cH^1_{\Omega^\vee}\simeq&
\coker A^\dagger_{\Omega^\vee}\otimes_{K_{\Omega^\vee}}\ker B^\dagger_{\Omega^\vee}
\nonumber\\
&\oplus
\ker A^\dagger_{\Omega^\vee}\otimes_{K_{\Omega^\vee}}\coker B^\dagger_{\Omega^\vee}.
\label{eq:packet-cohomology}
\end{align}
Under the decompositions in Eqs.~\eqref{eq:packet-homology} and \eqref{eq:packet-cohomology}, the pairing is the direct sum of the tensor-product pairings obtained from
\begin{align}
\ker A_\Omega
&\leftrightarrow
\coker A^\dagger_{\Omega^\vee},
\quad
\coker B_\Omega
\leftrightarrow
\ker B^\dagger_{\Omega^\vee},
\label{eq:packet-dualities-left}\\
\coker A_\Omega
&\leftrightarrow
\ker A^\dagger_{\Omega^\vee},
\quad
\ker B_\Omega
\leftrightarrow
\coker B^\dagger_{\Omega^\vee},
\label{eq:packet-dualities-right}
\end{align}
together with the reciprocal trace form of Lemma~\ref{lem:reciprocal-trace-pairing} on coefficient fields.
\end{lemma}

\begin{proof}
The binary expansion pairing on a free packet module is the coordinatewise reciprocal trace pairing between $K_\Omega$ and $K_{\Omega^\vee}$.  By the definition of the group-algebra adjoint,
\begin{align}
\langle A_\Omega u,v\rangle_{\Omega,\operatorname{tr}}
=
\langle u,A^\dagger_{\Omega^\vee}v\rangle_{\Omega,\operatorname{tr}},
\end{align}
and similarly for $B$.  We get
\begin{align}
(\ker A_\Omega)^\perp
&=\Im A^\dagger_{\Omega^\vee},
\quad
(\Im A_\Omega)^\perp=\ker A^\dagger_{\Omega^\vee},
\end{align}
with the analogous identities for $B$.  These identities give the perfect pairings between kernel and quotient spaces in Eqs.~\eqref{eq:packet-dualities-left} and \eqref{eq:packet-dualities-right}.  Taking tensor products of these perfect pairings gives a perfect pairing on the two K\"unneth summands, and the two summands live in the two direct summands of $Q_1$, so their cross-pairings vanish.  This proves the claimed perfect pairing between $\cH_{1,\Omega}$ and $\cH^1_{\Omega^\vee}$.
\end{proof}

\begin{theorem}[Abelian spectral basis for $\LP(A,B)$]
\label{thm:finite-abelian-basis}
Assume $G$ is finite Abelian of odd order.  For the lifted-product CSS code $\LP(A,B)$ over $\RG=\F_2[G]$, the binary logical dimension is
\begin{align}
k
&=
\sum_\Omega d_\Omega \bigl[ (n_A-s_A^\Omega)(r_B-s_B^\Omega) + (r_A-s_A^\Omega)(n_B-s_B^\Omega) \bigr],
\label{eq:finite-abelian-K}
\end{align}
where the sum runs over Frobenius packets.

Moreover, choosing finite-field bases of the packet homology spaces in Eq.~\eqref{eq:packet-homology}, lifting them with the idempotents $e_\Omega$, and pairing reciprocal packets by the trace and kernel/quotient dualities of Lemma~\ref{lem:packet-homology-cohomology-pairing} gives a complete addressable conjugate logical basis after binary expansion.
\end{theorem}

\begin{proof}
The semisimple CRT decomposition turns every free $\RG$-module into a product of finite-dimensional vector spaces. Thus, the product complex decomposes as
\begin{align}
Q_\bullet\simeq\prod_\Omega Q_{\bullet,\Omega},
\end{align}
where $Q_{\bullet,\Omega}$ is the finite-field product complex built from $A_\Omega$ and $B_\Omega$.

Homology commutes with finite products.  Since each packet component is a complex of vector spaces over the field $K_\Omega$, the K\"unneth formula has no Tor term and gives Eq.~\eqref{eq:packet-homology}.  Concretely, the middle term of the packet product complex is the direct sum
\begin{align}
Q_{1,\Omega}
=
\bigl(K_\Omega^{n_A}\otimes K_\Omega^{r_B}\bigr)
\oplus
\bigl(K_\Omega^{r_A}\otimes K_\Omega^{n_B}\bigr).
\end{align}
The first K\"unneth summand is represented in the first direct summand by tensors $u\otimes q$, with $u\in\ker A_\Omega$ and $q$ a representative of $\coker B_\Omega=K_\Omega^{r_B}/\Im B_\Omega$.  The second is represented in the second direct summand by tensors $q'\otimes v$, with $q'$ representing $\coker A_\Omega$ and $v\in\ker B_\Omega$.  Thus the two displayed K\"unneth summands are not only abstract direct summands of homology; with these representatives they live in the two displayed direct summands of $Q_{1,\Omega}$.  Its $K_\Omega$-dimension is Eq.~\eqref{eq:packet-dimension}, and multiplication by $d_\Omega=[K_\Omega:\F_2]$ converts this to binary dimension.  Summing over packets gives Eq.~\eqref{eq:finite-abelian-K}.

It remains to explain the conjugate pairing.  By Lemma~\ref{lem:same-packet-tensor}, each elementary tensor contribution uses the same idempotent $e_\Omega$ on both tensor legs, and general product representatives are sums of such packetwise contributions.  By Lemma~\ref{lem:reciprocal-binary-pairing}, a $Z$-representative in packet $\Omega$ can pair only with an $X$-representative in the reciprocal packet $\Omega^\vee$.  Lemma~\ref{lem:packet-homology-cohomology-pairing} identifies $\cH^1_{\Omega^\vee}$ as the perfect dual of $\cH_{1,\Omega}$ under the physical binary pairing.  Choosing trace-dual coefficient bases for reciprocal packets and finite-field kernel/quotient dual bases makes each packet Gram block the identity.  For arbitrary packet bases, one obtains the same result by inverting the finite Gram block.  The lifted binary representatives form a complete conjugate basis; in the implementation, this exact pairing is certified by the binary Gram matrix.
\end{proof}

In the single-variable quasi-cyclic setting, the theorem becomes a finite-field packet algorithm over $R_\ell=\F_2[x]/(x^\ell+1)$ with $\ell$ odd.  Algorithm~\ref{alg:packet-certificate} summarizes the construction behind the certificates of Table~\ref{tab:main}; Appendix~\ref{app:deterministic-implementation} gives the deterministic choices of kernel bases, quotient representatives, coefficient bases, and idempotent lifts.

\begin{algorithm}[t]
\DontPrintSemicolon
\small
\caption{Packetwise logical-basis certificate in the single-variable quasi-cyclic setting}
\label{alg:packet-certificate}
\KwIn{Seed matrices $A,B$ over $R_\ell=\F_2[x]/(x^\ell+1)$ with $\ell$ odd.}
\KwOut{Expanded binary representatives $\{Z_i,X_i\}$, packet labels, and the basis-width certificate $\Wcert$ (\secref{sec:representatives-certificates}).}

Factor $x^\ell+1=\prod_{g\in\mathcal G_\ell}g$ over $\F_2$.\;
\ForEach{packet $g\in\mathcal G_\ell$}{
    Set $K_g=\F_2[x]/(g)$ and evaluate $A_g$ and $B_g$ over $K_g$.\;
    Compute finite-field bases for $\ker A_g$, $\coker A_g$, $\ker B_g$, and $\coker B_g$.\;
    Form the same-packet K\"unneth representatives
    $\ker A_g\otimes\coker B_g$ and $\coker A_g\otimes\ker B_g$.\;
    Construct the primitive idempotent $e_g$ attached to packet $g$ and lift each packet representative by multiplying every tensor leg by $e_g$.\;
}
\ForEach{reciprocal packet orbit $\{g,g^\vee\}$}{
    Pair the $Z$ data in packet $g$ with the $X$ data in packet $g^\vee$ using trace-dual coefficient bases, or invert the finite packet Gram block.\;
}
Expand every lifted representative into binary coordinates.\;
Verify stabilizer commutation and the conjugate Gram identity $ZX^\transpose=I$.\;
\KwRet{The packet-labeled basis and $\Wcert=\max_i\{\wt(Z_i),\wt(X_i)\}$.}
\end{algorithm}

Theorem~\ref{thm:finite-abelian-basis} proves correctness in exact algebra.  The final binary step expands the basis into explicit representatives, confirms stabilizer commutation and the conjugate Gram identity exactly, and measures the support weights that enter the certificates; the weights depend on the deterministic choices of Appendix~\ref{app:deterministic-implementation} and are not fixed by the theorem.

\vspace{3pt}\noindent\textbf{Complexity.}
Every stage of Algorithm~\ref{alg:packet-certificate} is polynomial in the code size.  Factoring $x^\ell+1$ over $\F_2$ and constructing the primitive idempotents by the extended Euclidean algorithm cost $\operatorname{poly}(\ell)$~\cite{vonZurGathenGerhard2013}.  Packet evaluation and row reduction cost $O(rn\min(r,n))$ arithmetic operations in $K_g$ per packet; since one $K_g$ operation costs $O(d_g^2)$ bit operations and $\sum_g d_g=\ell$, the total over all packets is $O(rn\min(r,n)\,\ell^2)$ bit operations.  Forming the lifted representatives and running the final binary verification are dominated by the size of the output itself: the basis consists of $2k$ binary vectors of length $n$, and the stabilizer and Gram checks are sparse-matrix products.  The same counts hold for a general odd-order Abelian lift with $\ell$ replaced by $|G|$.  For comparison, at fixed seed shape, dense coordinate-order Gaussian elimination on the expanded matrices of side $\Theta(\ell)$ costs $\Theta(\ell^3)$ bit operations without fast matrix multiplication, followed by a global symplectic Gram--Schmidt pass over the $2k$ quotient representatives; the packetwise route costs $O(\ell^2)$ for the same seed-level linear algebra and needs no global symplectic pass, since reciprocal packets and trace-dual coefficient bases yield the conjugate Gram identity blockwise, with blocks of size at most $\max_\Omega d_\Omega$.  Structured binary methods can also exploit the block-circulant form; the packet route is one systematic way to do so that additionally preserves the labels.

\vspace{3pt}\noindent\textbf{Addressable output.}
For odd-order finite Abelian $G$, the construction outputs logical labels of the form
\begin{align}
\lambda=(\Omega,\sigma,i,j,b),
\end{align}
where $\Omega$ is a Frobenius packet, $\sigma$ specifies the left or right K\"unneth summand, $i$ and $j$ index the corresponding kernel and quotient basis elements, and $b$ indexes an $\F_2$-basis element of $K_\Omega$.  These labels have two levels.  The packet label $\Omega$, the K\"unneth-summand label $\sigma$, and the reciprocal pairing $\Omega\leftrightarrow\Omega^\vee$ are canonical: they are fixed by the primitive central idempotents and the natural K\"unneth splitting, independently of any basis choice, and the translation action of Corollary~\ref{cor:equivariance} is canonical at the same block level.  The indices $i$, $j$, and $b$ are deterministic but convention-dependent coordinates inside these canonical blocks, fixed by Appendix~\ref{app:deterministic-implementation} together with the representative-character choice of Appendix~\ref{app:idempotent}; the specific binary supports, and hence $\Wcert$, inherit this convention dependence.  The situation is that of a basis chosen inside a canonically determined eigenspace.  For each label, the construction gives binary support vectors $Z_\lambda,X_\lambda\in\F_2^n$ that commute with the stabilizers and satisfy the Gram identity $Z_\lambda X_\mu^\transpose=\delta_{\lambda\mu}$ of the addressability convention, and it stores the reciprocal conjugate packet, the primitive-idempotent lift, and the support weight of each representative.

\vspace{3pt}\noindent\textbf{Self-adjoint specialization.}
For the self-adjoint lifted product $\LP(A,A^\dagger)$ with $A:\RG^n\to\RG^r$ and $s_\Omega=\rankop_{K_\Omega}A_\Omega$, the group-inversion involution sends the $\Omega$-packet of $A^\dagger$ to the reciprocal packet of $A$, so $s_B^\Omega=s_A(\Omega^\vee)$ and Eq.~\eqref{eq:finite-abelian-K} becomes
\begin{align}
k
=
\sum_\Omega
d_\Omega
\bigl[
(n-s_\Omega)(n-s_{\Omega^\vee})
+
(r-s_\Omega)(r-s_{\Omega^\vee})
\bigr].
\label{eq:self-adjoint-K}
\end{align}

\begin{corollary}[Translation equivariance of the spectral basis]
\label{cor:equivariance}
For $h\in G$, let $U_h$ denote the CSS-preserving qubit permutation induced by the regular translation action of $h$ on every lift orbit, and let $\xi_\Omega(h)\in K_\Omega^\times$ be the image of $h$ under the CRT projection $\RG\to K_\Omega$; equivalently, $\xi_\Omega(h)=\omega(h)$ for the fixed representative character $\omega$ of $\Omega$, and in the cyclic case $K_g=\F_2[x]/(g)$ one has $\xi_g(x^a)=\xi_g^{\,a}$ with $\xi_g=x\bmod g$ as in Appendix~\ref{app:idempotent}.  Then $U_h$ maps the reported representatives to exact $\F_2$-linear combinations within single coefficient blocks:
\begin{align}
U_h\,Z_{(\Omega,\sigma,i,j,b)}
&=
\sum_{b'}\bigl[M_\Omega(h)\bigr]_{b'b}\,Z_{(\Omega,\sigma,i,j,b')},
\nonumber\\
U_h\,X_{(\Omega^\vee,\sigma,i,j,b)}
&=
\sum_{b'}\bigl[M_{\Omega^\vee}(h)\bigr]_{b'b}\,X_{(\Omega^\vee,\sigma,i,j,b')},
\end{align}
where $M_\Omega(h)\in\mathrm{GL}_{d_\Omega}(\F_2)$ is the matrix of multiplication by $\xi_\Omega(h)$ in the chosen $\F_2$-basis of $K_\Omega$.  In particular, the lift group acts on the logical algebra block-diagonally in the labels $(\Omega,\sigma,i,j)$, through the abelian unit groups $K_\Omega^\times$, and the conjugate Gram identity is preserved.  For $G=\Z_\ell$ and the cyclic generator $x$, $\xi_\Omega(x)=1$ exactly on the trivial packet, so the induced action $U^Z_x$ of $U_x$ on the $Z$-logical coordinates satisfies
\begin{align}
\rank(U^Z_x-I)=k-k_{x+1},
\label{eq:equivariance-rank}
\end{align}
with the same rank for the $X$-side action $U^X_x$,
where $k_{x+1}$ is the trivial-packet contribution to Eq.~\eqref{eq:finite-abelian-K}, and the order of the logical action is the least common multiple of the multiplicative orders of $\xi_\Omega(x)$ over packets with nonzero homology.
\end{corollary}

\begin{proof}
Translation by $h$ is multiplication by $h$, which is $\RG$-linear, hence commutes with the boundary maps and with every idempotent $e_\Omega$, and acts on the coefficient of a packet representative as multiplication by $\xi_\Omega(h)$.  The equality holds at the level of the stored binary representatives, not only of homology classes: in the semisimple case $e_\Omega\RG\to K_\Omega$ is an isomorphism, so $h$ times the deterministic lift of $\beta_b u$ and the deterministic lift of $\xi_\Omega(h)\beta_b u$ have the same packet image and vanish in all other packets, hence coincide.  Since $U_h$ is a qubit permutation, all binary pairings, stabilizer memberships, and the Gram identity are preserved.  Finally, $M_\Omega(x)-I$ is multiplication by $\xi_\Omega(x)-1$, which is invertible on every nontrivial packet and zero on the trivial packet, giving Eq.~\eqref{eq:equivariance-rank}.
\end{proof}

Corollary~\ref{cor:equivariance} is the property that distinguishes the spectral basis from a generic conjugate basis: for a coordinate-order quotient basis, the same translation generically acts by a dense $k\times k$ matrix with no label structure, whereas here it acts through blocks of size at most $\max_\Omega d_\Omega$, trivially on all remaining labels.

\vspace{3pt}\noindent\textbf{Template count.}
For a cyclic lift, let
\begin{align}
m_\Omega=\dim_{K_\Omega}\cH_{1,\Omega},
\qquad
T=\sum_\Omega m_\Omega.
\end{align}
Here $m_\Omega$ counts the $K_\Omega$-linear homology labels before choosing an $\F_2$ coefficient basis.  After binary expansion, packet $\Omega$ contributes $d_\Omega m_\Omega$ logical coordinates, so
\begin{align}
k=\sum_\Omega d_\Omega m_\Omega .
\end{align}
Thus $k/T$ is the $m_\Omega$-weighted average of the active packet degrees.  In the cyclic case $K_\Omega=\F_2[\xi_\Omega(x)]$, so multiplication by the cyclic shift generates the coefficient field.  For each fixed packet, K\"unneth summand, kernel index, and quotient index, the $d_\Omega$ binary coefficient representatives lie in the $\F_2$-span of cyclic translates of one representative template.  Large packet degrees therefore give large reuse of stored representative patterns.  If $\ell$ is prime and $2$ is primitive modulo $\ell$, all nontrivial cyclic packets have degree $\ell-1$; after separating the $x+1$ packet, nontrivial logical coordinates have reuse factor $\ell-1$ at the level of stored representative templates.  Circuit depth and scheduling require additional protocol data beyond this coordinate-storage count.

The odd-order assumption enters only through semisimplicity.  Appendix~\ref{app:even-lifts} develops what survives without it, summarized as follows.

\begin{theorem}[Even Abelian lifts]
\label{thm:beyond-odd}
Let $G$ be a finite Abelian group of arbitrary order.
\begin{enumerate}
\item[(i)] $R_G=\F_2[G]$ is a finite product of local rings (the primary packets), and the product complex and its homology decompose packet by packet.
\item[(ii)] For $G=\Z_\ell$ with $\ell=2^sm$ and $m$ odd, the logical dimension of $\LP(A,B)$ is an exact function of the nilpotent Smith profiles of the packet matrices, and the formula reduces to Eq.~\eqref{eq:finite-abelian-K} at $s=0$.
\item[(iii)] For even $\ell$ and $u=x^{\ell/2}+1$, multiplication by $u$ gives a canonical short exact sequence relating the product complex to that of its half-length code, with the exact recursion $\dim_{\F_2}\cH_1(\ell)=2\dim_{\F_2}\cH_1(\ell/2)-\rank\delta_1-\rank\delta_2$ in the ranks of the connecting Bockstein maps $\delta_i$ in the standard homological-algebra sense~\cite{Weibel1994}; iterating one factor of two per variable reduces any even Abelian lift to its odd skeleton.
\end{enumerate}
\end{theorem}

\begin{proof}
Appendix~\ref{app:even-lifts}: (i) is Proposition~\ref{prop:primary-decomposition}, (ii) is Theorem~\ref{thm:even-dimension}, and (iii) is Proposition~\ref{prop:halving-ses}.
\end{proof}

\noindent In particular, parts (i) and (iii) apply directly to bivariate-bicycle codes~\cite{BravyiCross2024}.

\section{Representatives, certificates, and distance witnesses}
\label{sec:representatives-certificates}

The theorem gives a packetwise logical basis abstractly.  This section sets out the representative choices and weight certificates used in the single-variable quasi-cyclic examples.  Nothing in this section changes the finite Abelian theorem; it fixes a concrete implementation layer for $R_\ell=\F_2[x]/(x^\ell+1)$ with $\ell$ odd, and for the self-adjoint family $\LP(A,A^\dagger)$.  Detailed worked examples and seed data are deferred to the appendices.

\begin{table}[t]
\caption{Main distance-witness and basis-width quantities used in the text and in Table~\ref{tab:main}.  ``Reported'' means the concrete representatives and values constructed in this manuscript.}
\label{tab:certificate-notation}
\small
\renewcommand{\arraystretch}{1.12}
\begin{tabular}{@{}L{0.25\columnwidth}L{0.68\columnwidth}@{}}
\toprule
Quantity & Meaning \\
\midrule
$\Dbest$ & tightest reported verified distance upper bound used in the parameter column. \\
$\Ddiag$ & verified low-cost distance diagnostic; Appendix~\ref{sec:seed-prefilters} defines the duplicate-column, minor, quotient-lift, and bounded-syzygy subclasses used to compute it. \\
$\Dwit$ & standalone distance witness from the literature, QDist, fast distance sampling~\cite{ZhouMaskara2026}, packet-sum search, or primary-packet search. \\
$\Dbasis$ & lightest representative in the reported addressable basis; for the odd rows this is the spectral basis of Theorem~\ref{thm:finite-abelian-basis}, while the even row uses the primary-packet-projected binary basis described in Appendix~\ref{app:even-lifts}. \\
$\Wcert$ & full reported conjugate-basis width certificate, so $B(C)\le\Wcert$. \\
RREF columns & coordinate-order comparison values: lightest RREF quotient representative, paired RREF width, and paired RREF width after deterministic greedy stabilizer-coset reduction. \\
\bottomrule
\end{tabular}
\end{table}

\subsection{Full-row-rank packets in the quasi-cyclic case}

Let $e_g\in R_\ell$ be the primitive idempotent attached to packet $g$, and let $g^\vee$ be the reciprocal factor.  For a seed $A:R_\ell^n\to R_\ell^r$, write $A_g$ for the packet evaluation of $A$ in $K_g$.  The formulas below give the explicit common-pivot specialization used for the single-variable examples.

Assume first that $A_g$ has full row rank $r$.  Choose pivot columns
\begin{align}
P=\{p_1,\ldots,p_r\},\quad F=\{1,\ldots,n\}\setminus P,
\end{align}
with $A_{g,P}$ invertible over $K_g$.  For each free column $j\in F$, solve
\begin{align}
A_{g,P}z_{g,j}=A_{g,j}
\end{align}
over $K_g$.  Then
\begin{align}
u_{g,j}=e_j+\sum_{a=1}^r (z_{g,j})_a e_{p_a}\in K_g^n
\label{eq:kernel-vector-full-rank}
\end{align}
is a kernel vector of $A_g$.  The free coordinate vector
\begin{align}
q_{g,j}=e_j
\end{align}
represents the corresponding quotient coordinate in $K_g^n/\operatorname{row}(A_g)$, and the choices satisfy
\begin{align}
q_{g,i}^{\transpose}u_{g,j}=\delta_{ij}.
\end{align}
Here $e_j$ denotes the standard coordinate vector; the primitive idempotent is always written with the packet subscript $e_g$.

Let
\begin{align}
\tau_g:K_{g^\vee}\to K_g
\end{align}
be the reciprocal field isomorphism induced by $x\mapsto x^{-1}$.  If one pivot set is valid for both $g$ and $g^\vee$, the left K\"unneth representatives can be written in a particularly transparent form.  For $h\in\{g,g^\vee\}$, for $j,j'\in F$, and for an $\F_2$-basis element $\alpha\in K_h$, set
\begin{align}
Z^{(L)}_{h;j,j';\alpha}
=
e_h\bigl((\alpha u_{h,j})\otimes \tau_h(q_{h^\vee,j'})\bigr).
\label{eq:left-Z-rep}
\end{align}
The conjugate representative is chosen in the reciprocal packet:
\begin{align}
X^{(L)}_{h^\vee;j,j';\alpha^\#}
=
e_{h^\vee}\bigl((\alpha^\# q_{h^\vee,j})\otimes \tau_{h^\vee}(u_{h,j'})\bigr),
\label{eq:left-X-rep}
\end{align}
where $\alpha^\#$ belongs to the trace-dual basis in the reciprocal packet in the sense of Lemma~\ref{lem:reciprocal-trace-pairing}.  The right K\"unneth summand is obtained by exchanging the $n$- and $r$-summand roles.  In a full-row-rank self-adjoint packet, this right summand vanishes.

The common-pivot assumption is only a convenience.  With different pivots, or with arbitrary finite-field bases, one forms the packet representatives from Theorem~\ref{thm:finite-abelian-basis} and then inverts the finite packet Gram block.  The final binary verification is the same.

\begin{lemma}[Factor-local conjugacy]
\label{lem:factor-local-conjugacy}
With a common pivot and trace-dual coefficient bases for reciprocal packets, the representatives in Eqs.~\eqref{eq:left-Z-rep} and \eqref{eq:left-X-rep} satisfy
\begin{align}
\langle Z^{(L)}_{g;j,j';\beta_a},X^{(L)}_{g^\vee;k,k';\beta_b^\#}\rangle_{\mathrm{bin}}
=
\delta_{jk}\delta_{j'k'}\delta_{ab}.
\end{align}
Pairings between nonreciprocal factor orbits vanish.
\end{lemma}

\begin{proof}
The binary pairing factors over the two tensor coordinates.  Lemma~\ref{lem:reciprocal-binary-pairing} forces the $Z$-packet to pair only with the reciprocal $X$-packet.  The kernel and quotient coordinates give the Kronecker factors in $j$ and $j'$, while Lemma~\ref{lem:reciprocal-trace-pairing} gives the Kronecker factor in $a$.
\end{proof}

\subsection{\texorpdfstring{The trivial $(x+1)$ packet}{The trivial (x+1) packet}}
\label{subsec:trivial-packet}

For dense monomial seeds, the trivial packet is explicit.  Let
\begin{align}
\chi_\ell=1+x+\cdots+x^{\ell-1}=\frac{x^\ell+1}{x+1}.
\end{align}
For odd $\ell$, this is the primitive idempotent for the factor $x+1$.  Since every monomial evaluates to $1$ at $x=1$, a dense monomial $r\times n$ seed reduces in this packet to the all-ones matrix $J_{r\times n}$, of rank one.

Let $e_i^{(n)}$ denote the standard basis of $\F_2^n$.  A convenient basis of $\ker J_{r\times n}$ is
\begin{align}
e_i^{(n)}+e_n^{(n)},
\quad
i=1,\ldots,n-1.
\end{align}
A quotient complement may be represented by $e_j^{(n)}$, $j=1,\ldots,n-1$.  When this packet has the rank-one defect above, its left contribution contains $(n-1)^2$ conjugate pairs
\begin{align}
Z^{(L)}_{ij}
=
\chi_\ell\bigl((e_i^{(n)}+e_n^{(n)})\otimes e_j^{(n)}\bigr),
\label{eq:trivial-packet-Z}
\\
X^{(L)}_{ij}
=
\chi_\ell\bigl(e_i^{(n)}\otimes(e_j^{(n)}+e_n^{(n)})\bigr),
\label{eq:trivial-packet-X}
\end{align}
each of weight $2\ell$.  The right contribution from the same packet similarly contains $(r-1)^2$ pairs.  If all nontrivial factors have full row rank, then the self-adjoint logical dimension is
\begin{align}
k
=
(n-r)^2(\ell-1)+(n-1)^2+(r-1)^2.
\label{eq:monomial-dimension}
\end{align}

This trivial packet is also diagnostically useful: it shows exactly how the low-frequency rank defect contributes to $k$, and, when it is non-full-rank, it often determines the full conjugate-basis width certificate.

\subsection{\texorpdfstring{Tuning the $x+1$ parity skeleton}{Tuning the x+1 parity skeleton}}

Sparse polynomial entries can tune the trivial packet while keeping bounded group-basis support.  If
\begin{align}
p_{ij}(x)=\sum_a c_{ij,a}x^a
\end{align}
is a seed entry, then
\begin{align}
p_{ij}(1)=\sum_a c_{ij,a}
\end{align}
is its coefficient parity.  Monomial entries give $1$ in the $x+1$ packet, while binomial entries give $0$.  By choosing odd- and even-support entries, one prescribes a binary parity skeleton
\begin{align}
M=A_{x+1}\in\F_2^{r\times n}.
\end{align}

If $M$ has rank $s$, and if $A_g$ has full row rank $r$ for every $g\neq x+1$, then Eq.~\eqref{eq:self-adjoint-K} gives
\begin{align}
k
=
(n-r)^2(\ell-1)+(n-s)^2+(r-s)^2.
\label{eq:parity-skeleton-dimension}
\end{align}
In particular, if $s=r$, the trivial packet has full row rank; its left contribution has the same $(n-r)^2$ labels as an ordinary full-row-rank packet.  This is the simplest sparse packet-rank engineering mechanism used in the examples.

\subsection{Basis width versus distance}

The construction outputs one complete conjugate basis.  Its support width is a property of that chosen basis, while the code distance is optimized over all nontrivial logical operators.  We use the following notation to keep these quantities separate.

The code distance $d(C)$ is the minimum weight of any nontrivial logical Pauli operator. The optimal conjugate logical-basis width is
\begin{align}
B(C)=
\min_{\{X_i,Z_i\}_{i=1}^k}
\max_i
\max\{\operatorname{wt}(X_i),\operatorname{wt}(Z_i)\},
\label{eq:basis-width}
\end{align}
where the minimum is over complete conjugate bases satisfying
\begin{align}
Z_iX_j^{\transpose}=\delta_{ij}.
\end{align}
The reported spectral basis gives one feasible conjugate basis, hence the certificate
\begin{align}
B(C)\leq W_{\mathrm{full}}^{\mathrm{cert}}.
\end{align}
These quantities obey
\begin{align}
d(C)\leq B(C)\leq W_{\mathrm{full}}^{\mathrm{cert}}.
\label{eq:distance-basis-certificate}
\end{align}
The first inequality follows because every representative in any conjugate basis is a nontrivial logical operator.  The second inequality is a construction-level certificate for the reported spectral basis.  Pivot choices, coefficient bases, reciprocal-stable packet grouping, and stabilizer-coset reduction can change this support cost while preserving the logical labels.  A packet union is \emph{reciprocal-stable} if $\Omega\in\mathcal P$ implies $\Omega^\vee\in\mathcal P$; this keeps a coarser packet label closed under the $X/Z$ reciprocal pairing.

The lightest representative inside this basis gives a separate distance witness:
\begin{align}
d(C) &\leq D_{\mathrm{basis}}^{\mathrm{ub}},
\\
D_{\mathrm{basis}}^{\mathrm{ub}}
&=
\min_\lambda
\min\{\operatorname{wt}(X_\lambda),\operatorname{wt}(Z_\lambda)\}.
\end{align}
$\Dbasis$ exhibits one nontrivial logical operator, while $\Wcert$ bounds the maximum support in a full addressable conjugate basis.

For a full-row-rank reciprocal factor orbit $O$ and a common pivot choice $P$, define the preliminary lifted-kernel support
\begin{align}
W_O^{\mathrm{pair}}(P)
=
\max_{h\in O}
\max_{j\in F(P)}
\operatorname{wt}(e_hu_{h,j}(P)).
\end{align}
This preliminary quantity captures the lifted kernel-vector support before coefficient-basis and conjugate-basis effects are included.  Let $B_h$ be the chosen coefficient basis of $K_h$, and let $B_h^\#$ be the trace-dual basis in the reciprocal packet.  The reciprocal-packet conjugate certificate is
\begin{align}
W_O^{\mathrm{conj}}(P)
=
\max_{h\in O}
\max_{j\in F(P)}
\max_{\beta\in B_h\cup B_h^\#}
\operatorname{wt}(e_h\beta u_{h,j}(P)).
\end{align}
The generic-packet certificate takes the best available pivot within each reciprocal orbit and the worst orbit overall:
\begin{align}
W_{\mathrm{gen}}^{\mathrm{conj}}
=
\max_O
\min_P
W_O^{\mathrm{conj}}(P).
\end{align}
The full certificate is then
\begin{align}
W_{\mathrm{full}}^{\mathrm{cert}}
=
\max\{W_{\mathrm{gen}}^{\mathrm{conj}},W_{\mathrm{nonfull}}\},
\end{align}
where $W_{\mathrm{nonfull}}$ includes all non-full-rank packet contributions, such as the $x+1$ packet above.  The complete conjugate basis is certified by $W_{\mathrm{full}}^{\mathrm{cert}}$; the preliminary quantity $W_O^{\mathrm{pair}}$ omits the coefficient-basis and non-full-rank packet effects and does not certify it.

\subsection{Seed-level distance witnesses}
\label{subsec:seed-distance-witnesses}

A seed can have a compact verified conjugate basis and still contain a much lighter logical operator from a sparse relation in the seed matrix or in a proper-divisor quotient of it.  The table records such low-cost verified mechanisms in the diagnostic column $\Ddiag$; the appendix separates duplicate-column, Cramer/minor, quotient-lift, and bounded-syzygy subclasses.  Each candidate is counted only after expanded-binary verification as a nontrivial logical operator.  The design goals are separate: distance should be large for protection, and, at a useful distance scale, the chosen logical basis should stay reasonably close to that scale.

In the single-variable self-adjoint quasi-cyclic setting
\begin{align}
R_\ell=\F_2[x]/(x^\ell+1),
\quad
A:R_\ell^n\to R_\ell^r,
\quad
\LP(A,A^\dagger),
\end{align}
we write
\begin{align}
\operatorname{wt}_{R}(u)=\sum_{j=1}^n|\operatorname{supp}(u_j)|
\label{eq:group-basis-weight}
\end{align}
for the group-basis weight, i.e., the binary Hamming weight after expanding every polynomial entry as a length-$\ell$ coefficient vector.  If $u\in\ker A$ and $q\in R_\ell^n$ represents a nonzero adjoint-quotient class, then the left degree-one summand contains the $Z$-type cycle
\begin{align}
Z=(u\otimes q,0)
\end{align}
because $\partial_1Z=Au\otimes q=0$.  Here $[u]$ and $[q]$ denote the induced kernel and quotient classes, respectively.  This cycle is a valid distance witness only when the K\"unneth product class $[u]\otimes[q]$ is nonzero.  Whenever a diagnostic distance-witness value enters $\Dbest$, we verify the expanded binary representative directly: for a $Z$-type distance witness, $H_XL^\transpose=0$ and $L\notin\row(H_Z)$; for an $X$-type distance witness, $H_ZL^\transpose=0$ and $L\notin\row(H_X)$.  In both cases we also check the reported support weight.  When $q$ is a monomial coordinate vector, every verified representative of this form satisfies
\begin{align}
\operatorname{wt}(Z)=\operatorname{wt}_{R}(u).
\end{align}

The table convention follows this hierarchy.  Parameter columns use the sharpest reported verified upper bound $\Dbest$, while $\Dbasis$ and $\Wcert$ remain properties of the constructed addressable basis.  The toy seed \seedlabel{T}{15}{2}{4} has $\Dbest=2$ even though its constructed-basis bound is larger, so it is used only as a warning example and not as an example row.  Appendix~\ref{app:seed-distance-witness-examples} works out this example, Appendix~\ref{sec:seed-prefilters} defines the pre-QDist diagnostic subclasses, and Appendix~\ref{app:packet-local-distance} gives a small example where a low-weight distance witness combines several packets.

\subsection{Stabilizer-coset post-processing}

The spectral representatives are fixed representatives of logical cosets.  Adding rows of $H_Z$ to a $Z$ representative or rows of $H_X$ to an $X$ representative preserves both the represented logical cosets and the conjugate Gram matrix, since stabilizers commute with all logical Pauli operators and pair trivially with logicals of the opposite type.  Exact minimization inside a stabilizer coset is a decoding problem; greedy local search, information-set decoding, BP-OSD-style reduction, or integer programming on small instances can nevertheless improve a certificate when the output is reverified.  If post-processing produces a complete conjugate basis with maximum support $W_{\mathrm{red}}^{\mathrm{cert}}$, then
\begin{align}
B(C)\leq W_{\mathrm{red}}^{\mathrm{cert}}\leq W_{\mathrm{full}}^{\mathrm{cert}}.
\end{align}
A packet-preserving variant restricts the added stabilizers to their packet projections, so that the representative stays inside its packet-labeled subspace at the cost of a more constrained coset problem; Appendix~\ref{app:packet-local-distance} details the projected formulation and the resulting packet-labeled certificate refinement.  The examples below report raw spectral certificates unless explicitly labeled as reduced; this convention keeps the certificate reproducible from packet data alone.

\section{Applications}
\label{sec:applications}

This section applies logical spectroscopy to concrete codes: quasi-cyclic construction examples, including high-rate seeds, and a noncyclic two-variable example.  The computations expand the packet bases into explicit binary representatives, confirm stabilizer commutation and the conjugate Gram matrix, and measure the support weights used in the reported certificates.  Those weights depend on the deterministic representative choices of Appendix~\ref{app:deterministic-implementation}; they are implementation data attached to the theorem, not consequences of the dimension formula alone.  The same packet data also produce basis-width certificates, verified distance witnesses, and structured-erasure diagnostics, which the tables keep as separate quantities.

\subsection{Quasi-cyclic construction examples}

The examples in Table~\ref{tab:main} exercise the construction across seed shapes, lift lengths, and rates, and illustrate the certificate hierarchy.
The seed labels are internal bookkeeping names.  The superscript gives the lift length $\ell$ and the subscript the matrix shape; the prefix indicates the source or role of the seed: $\mathsf A$ denotes a reference seed whose distance is cross-checked against available literature, $\mathsf R$ denotes a normalized random monomial probe, $\mathsf P$ denotes a sparse parity-systematic probe, $\mathsf M$ denotes an auxiliary normalized monomial probe of the corresponding rectangular seed shape, $\mathsf H$ denotes a high-rate sparse search finalist, $\mathsf E$ denotes an even-order exponent-lift example, and $\mathsf T$ denotes an appendix-only toy example used for illustration.  Here normalized means gauge-fixed up to row and column monomial scalings, for example by shifting exponents so the first row and first column are zero; it is a convention for avoiding duplicate seed descriptions, not an additional distance or rank assumption.
This identifies ${\mathsf R}^{75}_{3,7}$ as a random monomial probe of shape $3\times7$ at $\ell=75$.

Except for the small \seedlabel{P}{15}{2}{4} row, whose distance is verified exactly below, the distance entries in Table~\ref{tab:main} are upper bounds, not exact distances or lower bounds.  The table separates three sources.  The column $\Dbasis$ is the lightest representative in the reported addressable basis.  The column $\Dwit$ is a standalone distance witness: a fast distance-sampling bound from Ref.~\cite{ZhouMaskara2026}, a completed QDist search~\cite{QDistRnd2023}, a literature value, a verified packet-sum witness, or a verified primary-packet witness.  The column $\Ddiag$ lists low-cost distance diagnostics.  In the odd-order rows these are seed-level duplicate-column, Cramer/minor, proper-divisor quotient-lift, or bounded-syzygy distance witnesses; here a syzygy is a sparse polynomial relation among seed columns, with row relations obtained by applying the same test to $A^\dagger$.  For the even row, $\Ddiag$ is the verified primary-packet witness described below.  The appendix gives the subclass definitions and the Cramer/minor and quotient-lift calculations used for the $\mathsf M$, $\mathsf R$, and $\mathsf P$ example rows.

The code-parameter column uses an exact distance when available and otherwise uses $\Dbest$.  A value enters $\Dbest$ only when it has a verified nontrivial logical representative, a cited literature distance witness, a completed QDist distance witness, a fast-sampling distance witness of Ref.~\cite{ZhouMaskara2026}, a verified packet-sum distance witness, a verified quotient-lift distance witness, or a verified primary-packet distance witness.  Unverified diagnostics and unavailable entries are not used in the parameter column.  The $\mathsf A$ rows are from Ref.~\cite{Cain2026}.  Their $\Ddiag$ entries are low-cost distance diagnostics ($18,24,24$).  Appendix~\ref{app:score} provides the pre-QDist diagnostic summary, and the companion data record the row-by-row distance-witness source and verification status for the table entries.

The high-rate rows were selected by a small sparse parity-systematic search over binomial entries with $500\leq n_{\rm phys}\leq3000$ and rate above $0.3$.  The Cramer lower-tail scores used for ranking, reported as $\min/q10/{\rm median}$, are $80/82/86$ for \seedlabel{H}{41}{3}{8} and $62/66/72$ for \seedlabel{H}{33}{3}{9}.  Here $q10$ is the tenth percentile of the sampled Cramer/minor weights.  Their minima are the $\Ddiag$ entries for these rows.  In both high-rate rows, the parameter column is determined by the tighter $\Dwit$ entry.

The small-scale rows were selected from a bounded scan with $300\leq n_{\rm phys}\leq900$ to illustrate finite-size diagnostics at modest blocklength, without implying high-rate behavior.  Their verified Cramer/minor diagnostics are $\Ddiag=17,58,75,10$ for \seedlabel{P}{15}{2}{4}, \seedlabel{R}{21}{4}{5}, \seedlabel{P}{35}{3}{4}, and \seedlabel{P}{31}{2}{5}, respectively.  The independent distance witnesses are $\Dwit=9,26,35,8$, which determine the parameter-column bounds except where they tie the constructed-basis bound.  For \seedlabel{P}{15}{2}{4}, the weight-nine witness is complemented by an exact information-set enumeration: on both CSS sides, the kernel has dimension $180$, and every information-set pattern of weight at most $8$ was re-encoded.  Every resulting word of physical weight at most $8$ was then tested against the opposite stabilizer row space.  No nontrivial logical operator of weight at most $8$ was found, so $d_X,d_Z\geq9$; together with the verified weight-nine witness, this gives the exact distance $d=9$ for this row.

The row \seedlabel{E}{14}{3}{4} is different: it has even lift length $\ell=14=2\cdot7$, so the finite-field packet basis theorem is replaced by the primary-packet and Bockstein tools summarized in Theorem~\ref{thm:beyond-odd} and Appendix~\ref{app:even-lifts}.  Here
\begin{align}
x^{14}+1
=
(x+1)^2(x^3+x+1)^2(x^3+x^2+1)^2,
\end{align}
and the primary-packet homology dimensions are $14,6,6$, giving $k=26$.  The displayed $\Ddiag=14$ is a verified primary-packet witness in the $(x+1)^2$ packet, and the same weight is reproduced as $\Dwit=14$ by a $50{,}000$-iteration randomized QDist run.  The reported $\Dbasis=24$ and $\Wcert=168$ are obtained by projecting the expanded binary kernel and stabilizer spaces with the three primary idempotents, choosing packetwise quotient representatives, and pairing the resulting $26$ $Z$- and $26$ $X$-representatives by a finite binary Gram inversion.  This is an even-order primary-packet implementation certificate, not the odd-order trace-dual field-basis construction.  The RREF comparison values $\Drref=35$, $\Wrref=101$, and $\Wrrefred=75$ are obtained by ordinary binary expansion followed by the same coordinate-order RREF, Gram-pairing, and greedy stabilizer-reduction baseline used for the other rows.  The corresponding seed matrices are displayed in Appendix~\ref{app:seed-data}.

For the odd-order rows in Table~\ref{tab:main}, the companion computation constructs the binary representatives, reports $\Dbasis$ and $\Wcert$, and confirms $ZX^\transpose=I_k$ over $\F_2$; this establishes the conjugate pairing of the displayed representatives.  For the even row, the same binary checks verify the primary-packet-projected representatives after Gram pairing, while the trace-dual finite-field basis theorem is not invoked.

Table~\ref{tab:main} illustrates why the certificate hierarchy matters. The \seedlabel{R}{81}{3}{7} monomial probe has generic trace-dual certificate $108$ before the rank-defective packet is included; its rank-one $x+1$ packet has representatives of weight $2\ell=162$, giving $\Wcert=162$.  The distance witness $\Dwit=16$ is much tighter than the constructed-basis value $\Dbasis=60$.  The \seedlabel{P}{75}{3}{7} sparse probe makes the $x+1$ packet full rank, while its generic trace-dual certificate is $180$ rather than the preliminary lifted-kernel value $160$.

Another pattern in Table~\ref{tab:main} is that, in several rows, verified low-cost distance diagnostics are close to the best independent randomized QDist or literature distance witnesses.  Once a randomized QDist search returns a logical support, the corresponding upper bound is an explicit distance witness; the randomized aspect concerns how close the found distance witness is to the true minimum.  When a verified seed-level or quotient-level relation has weight comparable to an independent QDist or literature distance witness, the low-weight representative can be traced to that relation before full-code expansion.  In these cases, the relation is explicit, such as a duplicate-column relation, Cramer/minor relation, proper-divisor quotient-lift relation, or bounded syzygy, and its expanded product representative is verified as a nontrivial logical operator.

The Cramer/minor diagnostics detect only one low-weight mechanism.  The seeds \seedlabel{M}{45}{5}{9}, \seedlabel{P}{75}{3}{7}, and \seedlabel{P}{91}{3}{7} have large Cramer/minor diagnostics, while proper-divisor quotient-lift witnesses give smaller verified representatives.  Their best reported verified distance witnesses have weights
\begin{align}
\Dbest
&=
50,\quad 25,\quad 70,
\end{align}
respectively.  For the two larger remaining values, long-running randomized QDist attempts recorded in the companion data did not return tighter distance witnesses.
A randomized search that terminates without a lighter distance witness gives no distance lower bound.  The diagnostic conclusion is narrower: the listed duplicate-column and small-determinantal mechanisms are absent or pushed to substantially larger weight in these rows, while quotient-lift relations can still produce smaller verified representatives.  Appendix~\ref{app:cramer-spectra} explains this separation through seed-level Cramer/minor spectra, including a lift-length-independent minor ceiling for monomial seeds (Proposition~\ref{prop:cramer-ceiling}).

\vspace{5pt} \noindent {\bf Automorphism.} Beyond the certificates, the binary representatives permit a direct expanded-binary demonstration of the natural quasi-cyclic translation symmetry on the three highlighted rows, namely \seedlabel{M}{45}{5}{9}, \seedlabel{P}{75}{3}{7}, and \seedlabel{P}{91}{3}{7}.  Let $U_x$ shift every length-$\ell$ cyclic orbit by one coordinate, and let $U^Z_x$ and $U^X_x$ denote the induced source-indexed matrices on the logical $Z$ and $X$ quotient coordinates extracted from the conjugate pairings.  The computation confirms that the shifted representatives remain in $\ker H_X$ and $\ker H_Z$ and that the shifted basis preserves the identity $ZX^\transpose$ pairing.  The shifted action is a CSS-preserving qubit-permutation Clifford, in the usual wire-permutation sense.  For these three rows, respectively, the smallest positive power of $U_x$ acting trivially on both logical quotients is $45,75,91$, and
\begin{align}
\rank(U^Z_x-I)=\rank(U^X_x-I)
=704,\,\, 1184,\,\, 1440.
\end{align}
This realizes a cyclic subgroup of permutation Clifford gates generated by the LP shift symmetry.  Both measurements are exact predictions of Corollary~\ref{cor:equivariance}: the trivial-packet contributions to Eq.~\eqref{eq:self-adjoint-K} are $k_{x+1}=80,16,16$ for these rows, so Eq.~\eqref{eq:equivariance-rank} gives $\rank(U^Z_x-I)=k-k_{x+1}=704,1184,1440$, and the least common multiple of the packet orders of $\xi_\Omega(x)$ is $\ell=45,75,91$.  Transversal Hadamard or phase gates would require separate self-duality and phase-compatibility conditions.

\vspace{5pt} \noindent {\bf Addressing arithmetic.}
Corollary~\ref{cor:equivariance} also quantifies what the packet labels save when logical qubits must be addressed individually.  In the single-variable case, $K_\Omega=\F_2[\xi_\Omega(x)]$ by construction, so the cyclic translates of a single stored representative span its entire $d_\Omega$-dimensional coefficient block over $\F_2$: every representative $Z_{(\Omega,\sigma,i,j,b)}$ is an exact $\F_2$-linear combination of translates of one fixed $Z_{(\Omega,\sigma,i,j,b_0)}$, with coefficients supplied in closed form by the matrices $M_\Omega(h)$, and this equality holds at the level of the stored binary vectors.  A complete conjugate basis is therefore determined by
\begin{align}
T=\sum_\Omega \dim_{K_\Omega}\cH_{1,\Omega}
\end{align}
stored representative templates before the $\F_2$ coefficient expansion.  An ancilla or measurement gadget laid out for one template transports to every translate by the same global shift whenever the hardware layout implements that group translation uniformly on the relevant block orbits; alternatively, $d_\Omega$ linearly independent translates of one template themselves form a basis of the coefficient block, at the cost of a non-identity but closed-form Gram block.  For the highlighted rows, the packet counts give $T=80+7\times16=192$ templates for the $k=784$ logical qubits of \seedlabel{M}{45}{5}{9}, $T=16+7\times16=128$ for the $k=1200$ of \seedlabel{P}{75}{3}{7}, and $T=16+9\times16=160$ for the $k=1456$ of \seedlabel{P}{91}{3}{7}---reuse factors of roughly $4$, $9$, and $9$.  The two-variable example of the next subsection needs $T=20$ templates for $k=36$, with translates taken under the full group $\Z_3\times\Z_3$ since $K_\Omega$ is generated by the character coordinates jointly.  The companion rank-profile data verify the packet profiles used in these counts: for the three highlighted rows, respectively, the nontrivial packets all have full row rank and there are $7$, $7$, and $9$ nontrivial Frobenius packets; the $x+1$ packet contributions are the values $k_{x+1}=80$, $16$, and $16$ quoted above.  A coordinate-order quotient basis does not expose this packet-level template reuse.  Circuit-level scheduling and surgery overhead require additional protocol data.

\begin{figure}[t]
\includegraphics[width=0.99\columnwidth]{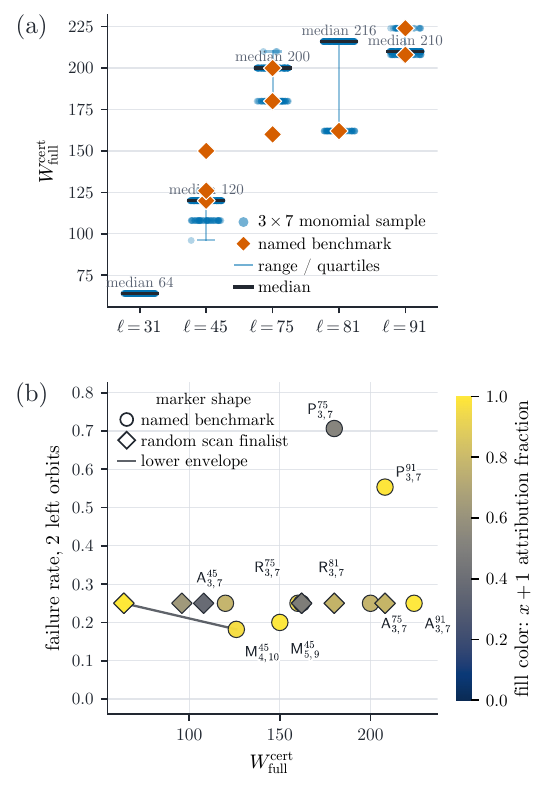}
\caption{ {\bf Certificate and erasure diagnostics.} (a) Distribution of the full conjugate-basis certificate $\Wcert$ over reproducible random monomial scans of $3\times7$ seeds for $\ell=31,45,75,81,91$. Blue circles are individual scan samples; vertical whiskers show the observed range; short horizontal blue marks show the interquartile range; orange diamonds mark named construction examples where shown; horizontal annotations give the median. (b) Certificate-erasure tradeoff for named examples and selected random monomial finalists passing the seed filter. Marker shape encodes source: circles are named examples, and diamonds are selected random monomial finalists. Fill color is independent of marker shape and gives the fraction of total erased-logical dimension attributed to the $x+1$ packet. The horizontal axis is $\Wcert$, the vertical axis is the exact failure rate after erasing two left-block-orbits, and the gray polyline is the lower envelope among the plotted candidates. Named-point labels identify the example seeds; the corresponding $\Dbest$ values are listed in Table~\ref{tab:main}. The original one- and two-orbit erasure values are reported in Table~\ref{tab:erasure}.}
\label{fig:codesignlandscape}
\end{figure}

\subsection{Distance witnesses versus basis certificates}

Table~\ref{tab:main} should be read with the seed-level distance-witness mechanism of Sec.~\ref{subsec:seed-distance-witnesses} in mind: $\Dbasis$ and $\Wcert$ describe the constructed addressable basis, while $\Dbest$ uses the lightest verified distance witness available.  The toy seed \seedlabel{T}{15}{2}{4} worked out in Appendix~\ref{app:worked-example} illustrates how large the gap can be ($\Dbest=2$ against $\Dbasis=12$), and the first small example row is the sparse \seedlabel{P}{15}{2}{4} seed precisely because it carries no such duplicate-column distance witness.

For any finite CSS code, binary Gaussian elimination can compute logical spaces as quotients $\ker H_X/\row(H_Z)$ and $\ker H_Z/\row(H_X)$, and symplectic Gram--Schmidt can turn such representatives into a conjugate binary basis.  These methods also provide explicit supports.  Table~\ref{tab:logical-coordinate-comparison} summarizes the coordinate data exposed by each construction, while the RREF columns in Table~\ref{tab:main} give a coordinate-order baseline.  Here $\Drref$ is the lightest nontrivial logical representative among the coordinate-order RREF quotient complements.  The paired width $\Wrref$ is obtained by forming $Z$- and $X$-side RREF quotient complements, computing their binary Gram matrix, replacing one side by the corresponding inverse-Gram linear combinations when the block is nonsingular, and verifying the resulting conjugate basis.  The reduced paired width $\Wrrefred$ then applies the deterministic greedy stabilizer-coset reduction of Sec.~\ref{sec:representatives-certificates} to each paired representative and verifies the conjugate CSS basis again.  The coordinate systems expose different data: a coordinate-order quotient complement is basis-dependent and generally does not identify the packet, K\"unneth summand, or reciprocal-packet trace pairing behind a representative.  The spectral construction supplies these labels together with the maximum support weight of the resulting conjugate basis.

Two consequences of the packet coordinate system are directly visible in the data.  First, the packet-labeled conjugate bases remain substantially lighter than coordinate-order RREF bases even after applying the deterministic greedy stabilizer-coset reduction of Sec.~\ref{sec:representatives-certificates} to the paired RREF representatives.  For example, \seedlabel{M}{45}{4}{10} has $\Wcert=126$, while the paired RREF width drops from $\Wrref=1672$ to $\Wrrefred=1422$ after this reduction.  The reduced RREF column is still not an optimized coset-leader computation, but it is a fairer baseline than the unreduced paired RREF width.  In this fixed-algorithm comparison the spectral entries are raw, unreduced representatives, while the RREF representatives receive the greedy coset reduction; in that limited sense the reported support gap is conservative.  A stronger coset reducer, for example BP-OSD-based reduction, could narrow both columns.  The comparison therefore concerns coordinate systems and their typical representative supports; it does not bound the optimal basis width $B(C)$ achievable from either starting point.  The weight comparison is illustrative; the structural distinction is the packet labels and the block-diagonal translation action of Corollary~\ref{cor:equivariance}, which no amount of coset reduction confers on a coordinate-order basis.  Second, the measured automorphism data above match the exact predictions of the translation equivariance of Corollary~\ref{cor:equivariance}: a single representative determines a whole orbit of translated representatives.

Figure~\ref{fig:codesignlandscape} summarizes the expanded certificate scan and exact erasure diagnostics used below.  It displays $\Wcert$, packet attribution, structured-erasure exposure, and verified distance witnesses as separate quantities that can be evaluated together during seed search.  In that figure, left-block-orbits and the erasure fractions $\rho_{s\rm orb}$ are as defined in \secref{sec:erasure}.

\subsection{A genuine two-variable finite Abelian example}

Consider
\begin{align}
G=\Z_3\times\Z_3, \quad R_G=\F_2[x,y]/(x^3+1,y^3+1).
\end{align}
This group is not cyclic: every nonidentity group element has order three. Take the sparse-polynomial seed
\begin{align}
A=\begin{pmatrix}
x^2&y&x+xy^2&x\\
y^2+xy^2&y^2+xy^2&x&y^2
\end{pmatrix}
\end{align}
and form the self-adjoint product $\LP(A,A^\dagger)$. Since $|G|=9$, $r=2$, and $n_{\rm seed}=4$,
\begin{align}
n=(n_{\rm seed}^2+r^2)|G|=180.
\end{align}
Full binary expansion gives
\begin{align}
\rank H_Z=72,
\qquad
\rank H_X=72,
\qquad
k=36,
\end{align}
and the CSS commutation check passes.
Table~\ref{tab:twovarpackets} lists the packet ranks and dimensions used in this example.

\begin{table}[t]
\begin{tabular}{cccccc}
\toprule
Packet & Rep. & $[K_\Omega: \F_2]$ & $\rank A_\Omega$ & $k_\Omega$ & binary dimension\\
\midrule
$\Omega_{00}$ & $(0,0)$ & 1 & 2 & 4 & 4\\
$\Omega_{10}$ & $(1,0)$ & 2 & 2 & 4 & 8\\
$\Omega_{01}$ & $(0,1)$ & 2 & 2 & 4 & 8\\
$\Omega_{11}$ & $(1,1)$ & 2 & 2 & 4 & 8\\
$\Omega_{12}$ & $(1,2)$ & 2 & 2 & 4 & 8\\
\midrule
Total & & & & & 36\\
\bottomrule
\end{tabular}
\caption{ {\bf Packet data for the noncyclic two-variable example over $\F_2[\Z_3\times\Z_3]$}. The representative $(a,b)$ denotes the character exponent pair. The degree $d = [K_\Omega : \F_2]$. $k_\Omega=\dim_{K_\Omega} \cH_1(\Omega)$ is the logical dimension inside that packet. The binary dimension is $[K_\Omega: \F_2] \cdot k_\Omega$, whose sum agrees with the full binary rank computation, $k=36$.}
\label{tab:twovarpackets}
\end{table}

The companion computation constructs $36$ $X$-type and $36$ $Z$-type representatives and checks
\begin{align}
X_iZ_j^\transpose=\delta_{ij},
\qquad
X_iH_Z^\transpose=0,
\qquad
Z_iH_X^\transpose=0.
\end{align}
The constructed representatives have $12\le\wt(X_i)\le57$ and $6\le\wt(Z_i)\le12$. These weights are not optimized; the point is that the finite Abelian packet construction survives binary expansion and produces a complete conjugate basis in a genuinely noncyclic two-variable group algebra.

The example also demonstrates exact whole-orbit erasure attribution. There are $20$ block-orbits in $Q_1$ under the group action. Table~\ref{tab:twovarerasure} enumerates one-, two-, and three-block-orbit erasures, compares the packet formula of \secref{sec:erasure} with full binary rank tests, and attributes total erased-logical dimension to packets. The mismatch count is zero in all three enumerations.

\begin{table*}[t]
\caption{Exact whole-orbit erasure attribution for the two-variable example. A failure is an erased block-orbit pattern with nonzero erased-logical dimension. The packet-total vector is ordered as $(\Omega_{00},\Omega_{10},\Omega_{01},\Omega_{11},\Omega_{12})$ and sums $d_Z+d_X$ over failing erasure patterns.}
\label{tab:twovarerasure}
\begin{tabular}{ccccccc}
\toprule
erased block-orbits & patterns & failures & failure rate & $\sum d_Z$ & $\sum d_X$ & packet-total vector\\
\midrule
1 & 20 & 0 & 0 & 0 & 0 & $(0,0,0,0,0)$\\
2 & 190 & 38 & 0.200 & 54 & 54 & $(8,16,32,16,36)$\\
3 & 1140 & 552 & 0.484 & 1012 & 1012 & $(164,320,584,320,636)$\\
\bottomrule
\end{tabular}
\end{table*}

\subsection{Orbit-erasure spectroscopy}
\label{sec:erasure}

Sections~\ref{sec:finite-abelian-basis} and \ref{sec:representatives-certificates} construct logical representatives and their certificates.  One use of the same packet data is to ask which located-erasure patterns can support logical operators.  For known erasure locations this is an exact CSS rank question, not a decoder simulation.  Let $E$ be a set of erased qubits and $\bar E$ its complement.  Since $H_X$ detects $Z$ errors and $H_Z$ detects $X$ errors, the erased-logical dimension for $Z$-type logical operators supported entirely inside $E$ is
\begin{align}
d_Z(E)=|E|-\rank H_X|_E-\rank H_Z+\rank H_Z|_{\bar E},
\end{align}
and similarly the erased-logical dimension for $X$-type logical operators is
\begin{align}
d_X(E)=|E|-\rank H_Z|_E-\rank H_X+\rank H_X|_{\bar E}.
\end{align}
An erasure is information-theoretically unrecoverable exactly when $d_Z(E)+d_X(E)>0$.

Whole-orbit erasure does not model independent qubit loss: under independent loss, the probability of erasing a complete length-$|G|$ orbit is negligible, and the failure rates below are not loss thresholds.  It is used here as a structured pattern whose erased-logical dimensions are exactly computable from the same finite-field data as the basis.  The physical interpretation is conditional: if each lift orbit is assigned to one addressable row, zone, tweezer tone, or analogous hardware channel with cyclic order, then loss of that channel is a whole-orbit erasure and the group generator is a simultaneous one-step shift on all such channels.  In that layout model, the same relative-support pattern for a representative or measurement gadget transports under the global shift; constructing a concrete circuit schedule is a separate task.

The packet decomposition becomes exact for erasures that respect the lift-group orbit structure.  If a block coordinate is erased together with its full $G$-orbit, the erased subspace is an $R_G$-submodule.  The rank test then decomposes packet by packet.

\begin{theorem}[Packet decomposition of whole-orbit erasures]
\label{thm:orbit-erasure-packets}
Let $E_S$ erase complete $G$-orbits of physical coordinates indexed by a block set $S\subset Q_1$, and let $\bar S$ be the complementary block set.  If $H_{X,\Omega}$ and $H_{Z,\Omega}$ are the packet components of the CSS check matrices over $K_\Omega$, then
\begin{align}
d_Z(E_S)&=\sum_\Omega [K_\Omega:\F_2]d_{Z,\Omega}(S),\\
d_X(E_S)&=\sum_\Omega [K_\Omega:\F_2]d_{X,\Omega}(S),
\end{align}
where
\begin{align}
d_{Z,\Omega}(S)
&=|S|-\rank H_{X,\Omega}|_S-\rank H_{Z,\Omega}
\nonumber\\
&\quad+\rank H_{Z,\Omega}|_{\bar S},\\
d_{X,\Omega}(S)
&=|S|-\rank H_{Z,\Omega}|_S-\rank H_{X,\Omega}
\nonumber\\
&\quad+\rank H_{X,\Omega}|_{\bar S}.
\end{align}
All ranks in $d_{Z,\Omega}$ and $d_{X,\Omega}$ are taken over $K_\Omega$.  Here $|S|$ counts erased group-algebra block coordinates, not binary qubits.
\end{theorem}

\begin{proof}
A whole-orbit erasure is an $R_G$-submodule of the physical block module.  Column restriction to $S$ and $\bar S$ commutes with the semisimple decomposition $R_G\simeq\prod_\Omega K_\Omega$.  Applying the erasure-rank formulas componentwise over $K_\Omega$ gives the displayed $K_\Omega$-dimensions.  Multiplication by $[K_\Omega:\F_2]$ converts them to binary dimensions, and summing over packets gives the result.
\end{proof}

The odd-order assumption can in fact be dropped from this theorem: the same decomposition holds over the primary local packets of an arbitrary finite Abelian lift group, with the packet quantities computed as binary dimensions of the local restrictions (Proposition~\ref{prop:primary-decomposition} in Appendix~\ref{app:even-lifts}).

Figure~\ref{fig:codesignlandscape} and Table~\ref{tab:erasure} illustrate why basis width, verified distance witnesses, and structured-erasure exposure are distinct objectives.  In this comparison, the monomial rows all sit on the exact $2/(n_{\rm seed}+1)$ baseline derived below, so their $\rho_{2\rm orb}$ differences reflect grid size rather than seed quality.  The sparse parity-systematic probes \seedlabel{P}{75}{3}{7} and \seedlabel{P}{91}{3}{7} repair parts of the trivial-packet rank profile, but they are more exposed to this particular one- and two-left-block-orbit erasure model.  Conversely, an attractive erasure score does not compensate for a verified small distance witness; such seeds should be treated as warning examples rather than construction examples.

For the self-adjoint quasi-cyclic comparisons, ``left'' block-orbits are the $n_{\rm seed}^2$ complete cyclic orbits in the left tensor summand $R_\ell^{n_{\rm seed}}\otimes_{R_\ell}R_\ell^{n_{\rm seed}}\subset Q_1$, and $\rho_{s\rm orb}$ denotes the fraction of size-$s$ left-block-orbit erasure patterns with nonzero erased-logical dimension.  For monomial seeds, $\rho_{1\rm orb}=0$: a logical supported in a single grid cell would require a nonzero polynomial annihilating a monomial column of $A$ or $A^\dagger$, and monomials are units in every packet.  For two erased left-block-orbits, the same unit argument excludes every pair of cells that shares neither a row nor a column, while any sharing pair supports a trivial-packet representative of the form Eqs.~\eqref{eq:trivial-packet-Z}--\eqref{eq:trivial-packet-X}; the monomial failures are exactly the sharing pairs, $2n_{\rm seed}\binom{n_{\rm seed}}{2}$ of the $\binom{n_{\rm seed}^2}{2}$ patterns, giving $\rho_{2\rm orb}=2/(n_{\rm seed}+1)$.  Table~\ref{tab:erasure} reports the original $\rho_{s\rm orb}$ values.  Relative to these baselines, \seedlabel{P}{75}{3}{7} has excesses $0.449$ and $0.457$ for one and two erased left-block-orbits, while \seedlabel{P}{91}{3}{7} has excesses $0.327$ and $0.304$.

These exact whole-orbit diagnostics slot into a staged seed-search procedure---algebraic prefilters first, full-code distance and decoder tests last---whose concrete layers and diagnostic table are laid out in Appendix~\ref{sec:seed-prefilters}.

\section{Discussion and Outlook}
\label{sec:discussion}

Abelian lifted-product logical operators can be constructed packetwise rather than by forcing a global module systematic form. The $x+1$ obstruction explains why dense monomial seeds defeat the naive HGP pivot argument. Once the group algebra is decomposed into Frobenius packets, however, the problem becomes ordinary finite-field homology in each packet. Primitive idempotents lift representatives back to the physical code; the same-packet tensor rule excludes spurious cross-packet products; and reciprocal trace duality gives a local mechanism for pairing $X$ and $Z$ bases.

The same calculation produces an addressable conjugate logical basis for odd-order finite Abelian $\mathrm{LP}(A,B)$ codes, a verified basis-width certificate, and an exact spectral decomposition of whole-orbit erasure failures.  These outputs are exercised at scale on the quasi-cyclic examples of Table~\ref{tab:main}, including high-rate codes, and on a noncyclic $\F_2[\Z_3\times\Z_3]$ example where no single-variable factorization is available. In what follows, we note some notable consequences of the packet viewpoint and the associated diagnostics, as well as the outlook.

\vspace{3pt}\noindent\textbf{Rank-profile interpolation.}
The spectral viewpoint also separates packetwise rank profiles that are algebraically possible from those realizable with sparse or hardware-local seed entries.  Relative to the fixed CRT identification, the entrywise bijection $R_G^{r\times n}\simeq\prod_\Omega K_\Omega^{r\times n}$ realizes any family of matrices $M_\Omega\in K_\Omega^{r\times n}$ by a unique seed $A$ with packet evaluation $A_\Omega=M_\Omega$; in particular, any rank profile $0\le s_\Omega\le\min(r,n)$ is achievable, and the profiles of $A$ and $B$ can be interpolated independently for a general product $\mathrm{LP}(A,B)$, while for $\mathrm{LP}(A,A^\dagger)$ the $B$ profile at $\Omega$ is determined by the $A$ profile at the reciprocal packet $\Omega^\vee$.  The interpolating entries may be dense in the group basis, so sparse or hardware-local rank-profile engineering becomes a constrained optimization problem, addressed by sparse mechanisms such as parity-skeleton tuning and divisor rank defects.

\vspace{3pt}\noindent\textbf{Search diagnostics.}
The certificate $\Wcert$ is also useful as a diagnostic during seed search, before any basis is adopted.  Figure~\ref{fig:codesignlandscape} illustrates this use with a modest reproducible scan of random monomial self-adjoint $3\times7$ seeds at $\ell=31,45,75,81,91$: in Fig.~\ref{fig:codesignlandscape}(b), no single seed traces the lower envelope in this scan, suggesting that certificate size and structured-erasure exposure trade off rather than co-optimize.  The reported certificates fix deterministic pivot choices and use optimized common pivots where available; pivot selection is a construction-level choice, not a code invariant.  These examples motivate using spectral ranks, $\Wcert$, pivot selection, and orbit-erasure spectra as inexpensive filters before a larger code-search study with decoder simulations; whether smaller $\Wcert$ alone predicts better decoding is a question for those decoder studies.  Because these diagnostics are computed from packet-level finite-field data, the framework supplies an algebraic prefilter for finite-length LP seed search: it can reject obvious low-distance or high-overhead seeds and identify candidates for full-code distance or decoder simulations (Appendix~\ref{sec:seed-prefilters}).  A broader numerical study of this search procedure is deferred to Ref.~\cite{LeeCodeSearchInPrep}.

\vspace{3pt}\noindent\textbf{Hardware motivation.}
For multicyclic lifts, a monomial seed entry $x^ay^b$ reads as a displacement by $(a,b)$ on a periodic two-dimensional array, and a sparse polynomial entry specifies a small displacement palette, so bounded group-basis support doubles as a constraint on the number and range of hardware moves in reconfigurable neutral-atom and Rydberg-array architectures~\cite{Sahay2023,Ma2023,PecorariPupillo2025,Cain2026,zhao2026ultrahighrate,WangMueller2026}.  Under an orbit-respecting layout of the kind described in Sec.~\ref{sec:erasure}, the equivariance of Corollary~\ref{cor:equivariance} has a direct layout reading: an ancilla or measurement gadget laid out for one representative transports to the entire packet orbit of representatives by the same global array translation, so addressing cost is paid once per orbit instead of once per logical coordinate.  The erasure diagnostics are natural in this conditional setting as well: atom loss, leakage, and certain decay events can often be detected or converted into located erasures, and for whole-orbit patterns the exact erased logical space decomposes into the same packets used for the logical-basis construction, identifying which packets become exposed.  On multidimensional frequency tori the primitive idempotents of individual packets can be spatially dense; partial packet localization over reciprocal-stable packet unions (Appendix~\ref{app:packet-local-distance}) trades packet resolution for physical support.  Circuit-level schedules, threshold estimates, and resource estimates for specific atom-array architectures are the next layer built on these examples.

\vspace{3pt}
The construction is thus an algebraic interface for later distance, decoder, and hardware studies.  For Abelian lifted products, logical qubits have natural packet addresses; once those addresses are explicit, basis construction, support accounting, and structured-erasure attribution become outputs of the same finite-field computation, while distance estimation, BP-OSD simulations, erasure-decoder studies, and architecture-specific schedules remain separate optimization layers.  Combined with lifting techniques, the same coordinate data may also help search for more efficient surgery gadgets~\cite{Hirasaki2026}.  Non-Abelian lifts and full decoder simulations require additional work.

Finally, we remark that the odd-order assumption is mainly for simplification. Theorem~\ref{thm:beyond-odd} and Appendix~\ref{app:even-lifts} spell out what survives without it, from the general primary packet decomposition to the exact even-lift dimension formulas, with reciprocal pairing implemented by the Frobenius-ring duality of the local packets rather than by field trace duality (Lemma~\ref{lem:primary-pairing}).  In particular, for bivariate-bicycle codes~\cite{BravyiCross2024}, the primary packet decomposition, whole-orbit erasure attribution, and the Bockstein dimension recursion apply directly.  The associated sparse code-search moves, exponent lifts and nilpotent repairs, and the multivariable even theory will be explored in Ref.~\cite{LeeCodeSearchInPrep}.

\begin{acknowledgments}

J.Y.L. acknowledges Hengyun Zhou and Nishad Maskara for rapid communication on fast distance sampling results on the example seeds~\cite{ZhouMaskara2026}.

J.Y.L. thanks Yuta Hirasaki for ongoing collaboration~\cite{Hirasaki2026} on lifted product codes as well as Ehud Altman for hosting his stay at the University of California, Berkeley. The stay enabled his participation in the Simons Institute for the Theory of Computing program ``Error-Correcting Codes: Theory and Practice,'' which provided important background and motivation for this research direction.
J.Y.L. acknowledges support from the faculty startup grant at the University of Illinois, Urbana-Champaign, and from the Korea Institute for Advanced Study through the Quantum Universe Center scholar program.
The randomized distance searches used QDist~\cite{QDistRnd2023}, and BP-OSD decoder checks used the LDPC software package~\cite{RoffeLDPC2022} and its BP-OSD implementation for quantum codes~\cite{RoffeWhite2020}.
\end{acknowledgments}

\section*{Data and code availability}

Machine-readable seeds, random-scan records, certificate scripts, figure inputs, and decoder command lines used for the numerical tables and figures will be collected in a companion repository. Until that repository is public, the code and data are available from the corresponding author upon reasonable request.

\appendix

\section{Reader's Guide} \label{app:reader-guide}

This appendix is a self-contained guide to the main construction. It gives a short tutorial on the group-algebra language used for lifted-product codes. The main text gives the formal statements; here we emphasize meaning and common pitfalls.

\subsection{Lifted Product Codes} \label{app:LP_code}

The lifted-product (LP) construction starts from two seed maps over a group algebra.
Let $R = \F_2[G]$, where $G$ is a finite Abelian group, and let
\begin{align}
    A:R^{n_A}\to R^{r_A},\quad B:R^{n_B}\to R^{r_B}
\end{align}
be two one-step chain complexes.  Each entry of a seed matrix is an element of $R$, expanded as $a = \sum_g a_g g$ with $a_g \in \F_2$ and $g \in G$.
After choosing the regular representation of $G$, which maps
\begin{align}
    g \mapsto S_g \in \F_2^{|G| \times |G|},
\end{align}
every entry expands to a $|G| \times |G|$ binary matrix, and the seed maps become ordinary binary maps
\begin{align}
    A_{\rm b}:\mathbb F_2^{|G|n_A}\to \mathbb F_2^{|G|r_A}, \quad B_{\rm b}:\mathbb F_2^{|G|n_B}\to \mathbb F_2^{|G|r_B},
\end{align}
which can be used to construct respective classical codes. This is called \emph{binary expansion}. For example, in the single-variable quasi-cyclic setting $R = \F_2[\mathbb{Z}_\ell] = \F_2[x]/(x^\ell + 1)$, the group element $x$ expands to the cyclic shift matrix. Therefore, if the seed shape and the group-basis support of each seed entry remain bounded, the expanded CSS checks have bounded weight and the construction gives a qLDPC family.

The LP code $\LP(A,B)$ is obtained from the product complex
\begin{align}
    Q_2 \xrightarrow{\partial_2} Q_1 \xrightarrow{\partial_1} Q_0,
\end{align}
where each $Q_i$ is an $R$-module
\begin{align}
Q_2&=R^{n_A}\otimes_R R^{n_B}, \nonumber \\
Q_1&=(R^{n_A}\otimes_R R^{r_B})\oplus(R^{r_A}\otimes_R R^{n_B}), \nonumber \\
Q_0&=R^{r_A}\otimes_R R^{r_B},
\end{align}
with boundary maps
\begin{align}
\partial_2(w\otimes z)&=(w\otimes Bz,Aw\otimes z), \nonumber \\
\partial_1(u\otimes v,a\otimes b)&=Au\otimes v+a\otimes Bb,
\end{align}
satisfying $\rd_1 \rd_2 = 0$.  In matrix form one may write
\begin{align} \label{eq:boundary_map}
    \rd_1 &=\begin{bmatrix} A\otimes I_{r_B} & I_{r_A}\otimes B
\end{bmatrix}, \nonumber \\
\rd_2 &=
\begin{bmatrix}
I_{n_A}\otimes B^\transpose & A^\transpose\otimes I_{n_B}
\end{bmatrix}^\transpose.
\end{align}
In \eqnref{eq:boundary_map}, the superscript $\transpose$ is the ordinary transpose of the displayed module indices coming from the chosen tensor-product vectorization.  The pre-expansion operator that expands to the physical binary transpose is the group-algebra adjoint $\dagger$ defined in Sec.~\ref{sec:finite-abelian-basis}.
The associated CSS code takes $H_X$ from the binary expansion of $\rd_1$ and $H_Z$ from the binary transpose of the expansion of $\rd_2$, equivalently from the expansion of $\rd_2^\dagger$. Thus, the logical operators are represented by the degree-one homology
\begin{align} \label{APP:logical}
\cL_Z =\ker\partial_1/\Im\partial_2,
\quad
\cL_X =\ker\partial_2^\dagger/\Im\partial_1^\dagger.
\end{align}
For non-Abelian $G$, the same binary expansion via the regular representation still produces group-algebra LDPC matrices~\cite{PanteleevKalachev2022}, but the tensor-product description requires compatible left and right $R$-module conventions; we restrict here to Abelian $G$, where $R$ is commutative.

\subsection{Hypergraph product logical bases}

The above construction contains the usual hypergraph product code $\HGP(A,B)$ as a special case, where $R = \F_2$. Using the K\"unneth formula, the logical space can be obtained as
\begin{align} \label{eq:reader-kunneth-hgp}
    \cH_1(A \otimes B) = \cH_1(A) \otimes \cH_0(B) \,\oplus\, \cH_0(A) \otimes \cH_1(B).
\end{align}
Note that $\cH_0(B) = \F_2^{r_B}/\Im B =: \coker B$; for full row-rank $B$, $\cH_0(B) = 0$.  A nonzero K\"unneth summand requires both tensor factors in that summand to be nonzero: a kernel factor from one input and a cokernel factor from the other. With our convention, the notation $\HGP(A,B)$ means the $\F_2$ specialization of $\LP(A,B)$.  The common convention in the HGP literature instead uses the second classical check matrix as a dual map; in the present notation, that is $\HGP(A,B^\transpose)$~\cite{TillichZemor2014,BravyiHastings2013}.

We now spell out the elementary linear algebra behind the logical construction. Let $M:\F_2^n\to\F_2^r$ be a binary matrix of rank $s$.  After row operations and column permutations, write
\begin{align}
    M \sim \mqty(I_s & P \\ \bm{0} & \bm{0}).
\end{align}
Let $q = n-s$ and $t = r-s$.
The $q$ free domain coordinates give a basis of $\ker M$.
If $\mathbf e_j$ is the $j$-th standard basis vector of $\F_2^q$, then
\begin{equation} \label{eq:APP_kernel}
    b_j=\begin{pmatrix}P \mathbf e_j\cr \mathbf e_j\end{pmatrix},
    \quad j=1,\ldots,q,
\end{equation}
satisfies $M b_j = 0$, i.e., $b_j \in \ker M$.
On the other hand, the $t$ zero target rows give representatives of the cokernel $\coker M =\F_2^r/ \Im M$.
If $f_\alpha$ is the $\alpha$-th standard basis vector of $\F_2^t$, we choose
\begin{align}
    c_\alpha = \mqty( 0 \\ f_\alpha ), \quad \alpha=1,\ldots,t,
\end{align}
as a cokernel representative.  The space $\cH_1(M)$ is represented by the kernel vectors $b_j$, while $\cH_0(M)$ is represented by the cokernel vectors $c_\alpha$.

Applying this separately to $A$ and to $B$, we obtain
\begin{align}
b_i^{A}\in \ker A,\quad c_\alpha^{A}\in \coker A, \nonumber\\
b_j^{B}\in \ker B,\quad c_\beta^{B}\in
\coker B.
\end{align}
The two $Z$-logical families in Eq.~\eqref{eq:reader-kunneth-hgp} are represented by
\begin{align}
Z^{(L)}_{i\beta}
&=
\bigl(b_i^{A}\otimes c_\beta^{B},0\bigr),\quad  Z^{(R)}_{\alpha j}=\bigl(0,c_\alpha^{A}\otimes b_j^{B}\bigr).
\end{align}
The first family lives in the left physical summand $\F_2^{n_A}\otimes \F_2^{r_B}$, while the second lives in the right summand $\F_2^{r_A}\otimes \F_2^{n_B}$.

We now describe the dual data used to choose conjugate $X$-logical representatives. For a single input map $M$, the dual of $\ker M$ is represented by
\begin{align}
\coker M^\transpose
&=
\F_2^n/\Im M^\transpose
=
\F_2^n/\operatorname{row}(M),
\end{align}
with representatives
\begin{align}
    a_j&= \mqty(0 \\ \mathbf e_j),\quad j=1,\ldots,q.
\end{align}
These pair with the kernel vectors in \eqnref{eq:APP_kernel} as
\begin{align}
a_i^\transpose b_j
&=
\delta_{ij}.
\end{align}
Similarly, the dual of $\coker M$ is represented by $\ker M^\transpose$.  In the same row-reduced coordinates, choose
\begin{align}
d_\alpha = \mqty(0 \\ f_\alpha), \quad  \alpha=1,\ldots,t,
\end{align}
so that $M^\transpose d_\alpha = 0$ and $d_\alpha^\transpose c_\beta = \delta_{\alpha \beta}$.

Applying this dual construction separately to $A$ and $B$, the conjugate $X$-logical representatives are
\begin{align}
X^{(L)}_{i\beta}
&=
\bigl(
a_i^{A}\otimes d_\beta^{B},0
\bigr),
\nonumber\\
X^{(R)}_{\alpha j}
&=
\bigl(
0,d_\alpha^{A}\otimes a_j^{B}
\bigr).
\end{align}
With the above dual choices, the pairings are
\begin{align}
\langle Z^{(L)}_{i\beta},X^{(L)}_{i'\beta'}\rangle
&=
\delta_{ii'}\delta_{\beta\beta'},
\nonumber\\
\langle Z^{(R)}_{\alpha j},X^{(R)}_{\alpha'j'}\rangle
&=
\delta_{\alpha\alpha'}\delta_{jj'},
\end{align}
and the cross-summand pairings vanish.

In the $R=\F_2$ specialization, row reduction on each input map exposes kernel coordinates, cokernel coordinates, and the dual coordinates needed for conjugate pairing.  The product complex then tensors the $A$-data with the $B$-data to form a grid of logical labels.  The lifted-product case is harder because $A$ and $B$ are matrices over a group algebra rather than over a field; the spectral construction restores this same field-linear-algebra picture packet by packet.

\subsection{Failure in lifted-product codes}

The ring $R_\ell$ can contain nonzero elements that are not invertible, and such an element is not a valid pivot.  A $1\times1$ example already shows the problem.  In $R=\F_2[x]/(x^3+1)$ one has
\begin{equation}
    (x+1)(x^2+x+1)=x^3+1=0,
\end{equation}
with neither factor zero, so $x+1$ is a zero divisor and has no inverse.  The $1\times1$ matrix $M=(x+1)$ is nonzero, yet $M(x^2+x+1)=0$: over a field a nonzero $1\times1$ matrix has zero kernel, but over the ring, treating $x+1$ as a pivot and dividing by it would destroy a real kernel vector.  This is the basic reason that ordinary Gaussian elimination over the full LP ring can fail.  The right replacement is spectral: decompose the ring into field components using the Chinese remainder theorem, do Gaussian elimination inside those fields, and then lift the answers back.

\subsection{\texorpdfstring{Frobenius packets over $\F_2$}{Frobenius packets over F2}}

Why should factoring $x^\ell+1$ be the analogue of Fourier analysis?  Over the complex numbers, the cyclic shift matrix $S$ is diagonalized by the discrete Fourier transform.  Diagonalization is useful because each Fourier vector is an eigenvector:
\begin{equation}
    S v_\alpha=\alpha v_\alpha,
    \qquad \alpha^\ell=1.
\end{equation}
Then any polynomial in the shift acts on that mode by scalar evaluation:
\begin{equation}
    p(S)v_\alpha=p(\alpha)v_\alpha.
\end{equation}
So diagonalizing $S$ turns a matrix problem into separate scalar problems, one frequency at a time.

Over $\F_2$, the same idea survives, but individual roots are usually not available over $\F_2$ itself. The shift generally does not split into one-dimensional eigenspaces.  Instead, it decomposes into irreducible
invariant subspaces, one for each irreducible factor $g$ of
$x^\ell+1$.
If $\alpha$ is a root of $x^\ell+1$, then a binary polynomial $p(x)$ satisfies $p(\alpha^2)=p(\alpha)^2$, since $(u+v)^2=u^2+v^2$ under $F_2$.  If $p(\alpha)=0$, then
\begin{equation}
    p(\alpha^2)=p(\alpha^4)=p(\alpha^8)=\cdots=0.
\end{equation}
In other words, a binary polynomial cannot have $\alpha$ as a root without also having all its Frobenius conjugates as roots.  The set $\{\alpha,\alpha^2,\alpha^4,\alpha^8,\ldots\}$ until it repeats is called the Frobenius orbit of $\alpha$.  In the cyclic case, the root value $\alpha^k$ is equivalently the value of the character sending the generator of $\Z_\ell$ to $\alpha^k$; thus this root orbit is the character packet of Sec.~\ref{sec:pivots-to-packets} under the standard cyclic identification.  An irreducible factor $g(x)$ of $x^\ell+1$ is precisely the minimal polynomial whose roots are one such packet:
\begin{align}
\Omega_g &= \{\alpha \,|\, g(\alpha)=0\},
\end{align}
and this orbit is encoded by the irreducible minimal polynomial
\begin{align}
    g(x) &= \prod_{\beta\in\Omega_g}(x-\beta),
\end{align}
where $\beta$s are not expressible in $\F_2$.
The finite-field version of Fourier analysis proceeds by packets rather than individual roots.  Factoring $x^\ell+1$ over $\F_2$ finds these packets.

To illustrate, consider $x^7 + 1$ and take $\alpha$ to be a primitive seventh root of unity.  The seven frequency indices are $\{ 0,1,2,3,4,5,6\}$, where they correspond to the integer $k$ labeling the mode $\alpha^k$.
The Frobenius map $\alpha\mapsto\alpha^2$ doubles the index modulo seven. Therefore, the indices split into Frobenius orbits
\begin{equation}
    \{0\},\qquad \{1,2,4\},\qquad \{3,6,5\}.
\end{equation}
Correspondingly,
\begin{equation}
    x^7+1=(x+1)(x^3+x+1)(x^3+x^2+1)
\end{equation}
over $\F_2$.

The factor $x+1$ is the trivial or constant packet.  It has root $x=1$, and the cyclic shift eigenvector with eigenvalue $1$ is the all-ones vector
\begin{equation}
    v=\sum_{n=0}^6 x^n \sim (1,1,1,1,1,1,1).
\end{equation}
Shifting this vector does nothing.  This is the finite-field analogue of the zero-frequency component.

The two cubic factors are reciprocal.  If
\begin{equation}
    g(x)=x^3+x+1,
\end{equation}
then its reciprocal polynomial is
\begin{equation}
    g^\vee(x) := x^{\deg g}g(x^{-1})=x^3g(x^{-1})=x^3+x^2+1.
\end{equation}
On modes, reciprocity is inversion.  The inverse of the packet $\{\alpha,\alpha^2,\alpha^4\}$ is
\begin{equation}
    \{\alpha^{-1},\alpha^{-2},\alpha^{-4}\}
    =\{\alpha^6,\alpha^5,\alpha^3\},
\end{equation}
which is the other nontrivial Frobenius packet. Note that in this case, one may write a packet vector (primitive-idempotent representative) for $g(x)$ as
\begin{align}
v=1+x+x^2+x^4 \sim (1,1,1,0,1,0,0).
\end{align}
along with $xv$ and $x^2v$.
More generally, given an irreducible factor $g$ with $\deg g = d_g$, $\{ v, xv, \ldots \}$ obtained by multiplying $v$ with $\{1,x,...,x^{d_g-1} \}$ forms the basis for the subspace.

\subsection{Packet projectors and packet evaluation} \label{app:idempotent}

Suppose $x^\ell+1=\prod_g g(x)$ is the factorization into distinct irreducible polynomials and square-free. The Chinese remainder theorem (CRT) says
\begin{equation}
    R_\ell=\F_2[x]/(x^\ell+1)\cong \prod_g \Big[ \F_2[x]/(g) \Big],
\end{equation}
where $K_g = \F_2[x]/(g)$ is the packet field. For each irreducible factor $g$, there is a unique element $e_g\in R_\ell$, the primitive central idempotent attached to $g$~\cite{DummitFoote2004,MacWilliamsSloane1977}, such that
\begin{equation}
    e_g\equiv \delta_{g,h} \pmod h.
\end{equation}
It further satisfies
\begin{align}
e_g^2=e_g,
\quad e_ge_h=0\ (g\ne h),
\quad \sum_g e_g=1.
\end{align}
Multiplication by $e_g$ projects any $R_\ell$-module element onto its
$g$-packet component.  Conversely, if $u_g\in K_g^n$ is a packet-local
representative, then choosing any lift $\widetilde u\in R_\ell^n$ and
multiplying by $e_g$ gives
\begin{align}
u^{\mathrm{lift}}
&=
e_g\widetilde u
\in
R_\ell^n,
\end{align}
whose $g$ component is $u_g$ and whose other packet components vanish:
\begin{align}
u^{\rm lift}\bmod g=u_g,
\quad
u^{\rm lift}\bmod h=0
\qquad
(h\neq g).
\end{align}
This is why we call $e_g$ a packet projector. For $\ell=7$, the primitive idempotents are
\begin{align}
    e_{x+1}&=1+x+x^2+x^3+x^4+x^5+x^6,\nonumber\\
    e_{x^3+x+1}&=1+x+x^2+x^4,\nonumber\\
    e_{x^3+x^2+1}&=1+x^3+x^5+x^6.
\end{align}
After replacing $x$ by the cyclic shift matrix $S$, the idempotent becomes a concrete binary matrix projector $E_g=e_g(S)$ onto the shift-invariant packet subspace $e_gR_\ell^n\subset R_\ell^n\simeq\F_2^{n\ell}$---the finite-field analogue of filtering a signal to one Fourier band.

Packet evaluation is the corresponding coordinate description of a seed matrix.  Writing $\xi_g=x\bmod g\in K_g$, the shift variable becomes the finite-field scalar $\xi_g$ on the $g$ packet, and a seed $A=(p_{ij}(x))\in R_\ell^{r\times n}$ evaluates to
\begin{align}
A_g
&=
(p_{ij}(\xi_g))
=
(p_{ij}(x)\bmod g(x))
\in
K_g^{r\times n}.
\end{align}
This is an ordinary matrix over the field $K_g$, so Gaussian elimination is valid inside the packet: every nonzero packet scalar is a legitimate pivot.  For example, for $g(x)=x^3+x+1$ one has $\xi^3=\xi+1$ in $K_g$, and
\begin{align}
(\xi+1)(\xi^2+\xi)
=
\xi^3+\xi
=
1,
\end{align}
so $\xi+1$ is invertible in this packet even though $x+1$ is a bad global pivot over $R_\ell$, because it vanishes in the trivial packet.

At the module level, the binary block-circulant map $A_{\rm bin}=(p_{ij}(S))$ preserves every packet subspace, and the restriction of $A$ to $e_gR_\ell^n\to e_gR_\ell^r$, after the identification $e_gR_\ell\simeq K_g$, is exactly $A_g$.  Equivalently, if $d=\deg g$ and an $\F_2$-basis of $K_g$ is chosen, each field element becomes a $d\times d$ binary multiplication matrix, and a nonzero field pivot is an invertible $d\times d$ binary block: the finite-field notation is simply the clean way to perform this block Gaussian elimination.  The construction is therefore
\begin{align}
A \text{ over } R_\ell
&\longrightarrow
A_g \text{ over } K_g
\nonumber\\
&\longrightarrow
\text{row reduction over } K_g
\nonumber\\
&\longrightarrow
\text{packet vector } u_g
\nonumber\\
&\longrightarrow
e_g\widetilde u_g \text{ over } R_\ell
\nonumber\\
&\longrightarrow
\text{binary support}.
\end{align}
The pivoting happens over the packet field $K_g$; the final representative is binary only after the idempotent lift and coefficient expansion.

For a general finite Abelian group, the same picture holds with characters in place of roots: a character $\omega=(\zeta_1,\ldots,\zeta_m)$ with $\zeta_i^{\ell_i}=1$ evaluates a polynomial seed entry as $p(\omega)=\sum_a c_a\zeta_1^{a_1}\cdots\zeta_m^{a_m}$, giving $A_\Omega=(p_{ij}(\omega))$ over $K_\Omega=\F_2(\zeta_1,\ldots,\zeta_m)$.  There is a mild nonuniqueness: a packet is a Frobenius orbit of characters, not a single character, and since the seed coefficients are binary, $p(\omega^2)=p(\omega)^2$.  Choosing $\omega^2$ instead of $\omega$ therefore replaces $A_\Omega$ by its entrywise Frobenius square---an isomorphic packet computation with the same rank, kernel and cokernel dimensions, and the same binary packet subspace after lifting by the same $e_\Omega$.  It can, however, change the chosen finite-field basis and the weights of particular lifted representatives.  For a concrete basis-width certificate, one fixes the representative character, the field basis, the trace-dual basis, and the pivots, constructs the binary representatives, and verifies stabilizer commutation, the final Gram matrix, and support weights directly.

Rank drops create logical operators for the same reason they create larger kernels and cokernels. For $s_A^\Omega=\rankop A_\Omega$,
\begin{align}
    \dim\ker A_\Omega&=n_A-s_A^\Omega,\nonumber\\
    \dim\coker A_\Omega&=r_A-s_A^\Omega.
\end{align}
Lower rank means fewer independent constraints, hence more undetected degrees of freedom.  Those extra degrees of freedom become extra logical contributions through the K\"unneth products.

\subsection{Deterministic implementation conventions}
\label{app:deterministic-implementation}

Theorem~\ref{thm:finite-abelian-basis} is basis-independent.  For a reproducible certificate, however, the implementation must fix representatives for kernels, quotients, coefficient bases, trace-dual bases in reciprocal packets, idempotent lifts, and final binary row storage.  This subsection documents the convention used by the companion computations.
The constituent-field representation is the standard CRT viewpoint on quasi-cyclic codes~\cite{LingSole2001,GuneriLingOzkaya2020}; the finite-field arithmetic below uses standard polynomial-basis and modular-arithmetic algorithms~\cite{LidlNiederreiter1997,vonZurGathenGerhard2013}.

First consider a single packet field $K$ and an ordinary finite-field matrix
\begin{align}
M
&\in
K^{r\times n}.
\end{align}
The superscript $\transpose$ in this paragraph denotes ordinary transpose over $K$; the physical binary adjoint is the group-algebra adjoint $\dagger$ and is accounted for only after passing to the reciprocal packet.  We compute reduced row echelon forms over $K$ using the fixed field basis and lexicographic pivot order.  This gives deterministic column bases
\begin{align}
U_M
&\in
K^{n\times q},
\quad
D_M
\in
K^{r\times t},
\end{align}
for $\ker M$ and $\ker M^\transpose$, respectively, where $q=n-\rankop M$ and $t=r-\rankop M$.

The quotient representatives are fixed by dualizing these kernel bases.  Let $I_M\subset\{1,\ldots,n\}$ be the lexicographically first set of $q$ row indices such that the square submatrix $U_M[I_M]$ is invertible.  Define
\begin{align}
Q_M[I_M]
&=
\bigl(U_M[I_M]^\transpose\bigr)^{-1},
\quad
Q_M[I_M^c]=0.
\end{align}
Then
\begin{align}
U_M^\transpose Q_M
&=
I_q,
\end{align}
and the columns of $Q_M$ represent the quotient $K^n/\operatorname{row}(M)$ dual to $\ker M$.  Similarly, if $J_M$ is the lexicographically first set of $t$ row indices for which $D_M[J_M]$ is invertible, set
\begin{align}
C_M[J_M]
&=
\bigl(D_M[J_M]^\transpose\bigr)^{-1},
\quad
C_M[J_M^c]=0.
\end{align}
Then $D_M^\transpose C_M=I_t$, and the columns of $C_M$ represent $\coker M=K^r/\operatorname{im}M$ dual to $\ker M^\transpose$.

For a single-variable packet $g$, this convention is applied to $A_g$ and $B_g$.  We write the resulting columns as
\begin{align}
u^A_{g,i},\quad q^A_{g,i},\quad d^A_{g,\mu},\quad c^A_{g,\mu},
\end{align}
and analogously for $B$.  Here $u$ denotes a kernel vector, $q$ a quotient representative dual to a kernel vector, $d$ a left-kernel vector, and $c$ a cokernel representative.  If $\{\beta_{g,a}\}_{a=1}^{d_g}$ is the fixed $\F_2$-basis of $K_g$, a packet $Z$ basis for the left and right K\"unneth summands is represented in packet coordinates by
\begin{align}
Z^L_{g,i,\nu,a}
&=
\bigl((\beta_{g,a}u^A_{g,i})\otimes c^B_{g,\nu},0\bigr),
\nonumber\\
Z^R_{g,\mu,j,a}
&=
\bigl(0,(\beta_{g,a}c^A_{g,\mu})\otimes u^B_{g,j}\bigr).
\end{align}
These are packet-coordinate representatives; the physical binary vectors are obtained only after the primitive-idempotent lift and coefficient expansion.

The reciprocal packet fixes the conjugate $X$ labels.  Let $\tau_g:K_{g^\vee}\to K_g$ be the inversion-induced field isomorphism.  Given a basis $\{\gamma_{g^\vee,b}\}$ of $K_{g^\vee}$, form the trace Gram matrix
\begin{align}
T_{ab}
&=
\operatorname{Tr}_{K_g/\F_2}
\bigl(\beta_{g,a}\tau_g(\gamma_{g^\vee,b})\bigr).
\end{align}
The trace-dual basis in the reciprocal packet is
\begin{align}
\beta^\#_{g,a}
&=
\sum_b
(T^{-1})_{ba}\gamma_{g^\vee,b}.
\end{align}
The nondegeneracy of the finite-field trace pairing and the existence of trace-dual bases are standard finite-field facts~\cite{LidlNiederreiter1997}.
Using the quotient-dual representatives above, the corresponding trial reciprocal-packet $X$ basis is
\begin{align}
\widetilde X^L_{g^\vee,i,\nu,a}
&=
\bigl((\beta^\#_{g,a}\tau_g^{-1}(q^A_{g,i}))
\otimes \tau_g^{-1}(d^B_{g,\nu}),0\bigr),
\nonumber\\
\widetilde X^R_{g^\vee,\mu,j,a}
&=
\bigl(0,(\beta^\#_{g,a}\tau_g^{-1}(d^A_{g,\mu}))
\otimes \tau_g^{-1}(q^B_{g,j})\bigr).
\end{align}
With the dual choices above, the packet Gram block is already the identity.  For arbitrary finite-field bases or pivot choices, the implementation instead forms the binary Gram block
\begin{align}
G_{\lambda\eta}
&=
\bigl\langle Z_\lambda,\widetilde X_\eta\bigr\rangle_{\rm bin}.
\end{align}
Whenever $G$ is invertible, the trial $X$ representatives are replaced by the corresponding linear combinations using $G^{-\transpose}$, so the final representatives satisfy $ZX^\transpose=I$.

The idempotent lift is deterministic as well.  Write $f=x^\ell+1$ and $h_g=f/g$.  Choose $a_g$ by the extended Euclidean algorithm so that $a_gh_g\equiv1\pmod g$, and set
\begin{align}
e_g
&=
a_gh_g
\pmod{f}.
\end{align}
This is the usual constructive formula for the primitive idempotent; the inverse $a_g$ is computed by the extended Euclidean algorithm~\cite{DummitFoote2004,vonZurGathenGerhard2013}.
For $w_g\in K_g^m$, choose the coordinate polynomial representative $\widetilde w_g\in R_\ell^m$ with degrees below $\deg g$, and lift by
\begin{align}
w_g^{\rm lift}
&=
e_g\widetilde w_g
\in
R_\ell^m.
\end{align}
The binary expansion of a polynomial $p=\sum_{a=0}^{\ell-1}p_ax^a\in R_\ell$ is the coefficient vector
\begin{align}
\operatorname{coeff}_\ell(p)
&=
(p_0,\ldots,p_{\ell-1})
\in
\F_2^\ell.
\end{align}

Finally, every reported basis certificate is checked after binary expansion.  With row storage conventions, the code verifies
\begin{align}
H_X Z^\transpose=0,\quad H_Z X^\transpose=0,\quad Z X^\transpose=I.
\end{align}
For any standalone distance witness entering $\Dbest$, the expanded binary vector is additionally checked to commute with the stabilizers, to have the reported support weight, and to lie outside the appropriate stabilizer row space.  The theorem supplies the exact algebraic construction, while the implementation retains the concrete binary representatives and catches convention or lifting errors.  Appendix~\ref{app:seed-distance-witness-examples} illustrates this distinction for the toy duplicate-column seed.

\subsection{What changes for multicyclic Abelian lifts?}

The single-variable quasi-cyclic setting uses the cyclic group $\mathbb Z_\ell$ and the ring $\F_2[x]/(x^\ell+1)$.  A multicyclic lift uses several independent shift variables, for example
\begin{equation}
    \F_2[x,y]/(x^L+1,y^M+1),
\end{equation}
which corresponds to the group $\mathbb Z_L\times\mathbb Z_M$.  A monomial $x^ay^b$ is a displacement by $(a,b)$ on a two-dimensional periodic grid.

The spectral picture is the same.  A mode is now a point
\begin{equation}
    (\omega_1,\omega_2),
    \qquad
    \omega_1^L=1,
    \quad
    \omega_2^M=1.
\end{equation}
The Frobenius map squares every coordinate:
\begin{equation}
    (\omega_1,\omega_2)\mapsto(\omega_1^2,\omega_2^2).
\end{equation}
Packets are Frobenius orbits on this finite frequency torus, and reciprocal packets are obtained by inversion
\begin{equation}
    (\omega_1,\omega_2)\mapsto(\omega_1^{-1},\omega_2^{-1}).
\end{equation}
For example, with $L=7$ and $M=5$, the packet containing exponent pair $(1,1)$ is obtained by repeatedly doubling both coordinates modulo $(7,5)$:
\begin{align}
    &(1,1),(2,2),(4,4),(1,3),(2,1),(4,2),\nonumber\\
    &(1,4),(2,3),(4,1),(1,2),(2,4),(4,3).
\end{align}
The construction is unchanged: compute finite-field packet data, form same-packet tensor contributions, pair reciprocal packets, and lift back by primitive idempotents.  The practical new issue is that the primitive idempotents of individual packets on a multidimensional torus can be spatially dense.  This motivates partial packet localization, where several packets are grouped to trade packet resolution for smaller physical support.

\section{A small worked example}
\label{app:worked-example}

This appendix works through the toy duplicate-column seed \seedlabel{T}{15}{2}{4}, whose packet ranks and representative weights are small enough to inspect by hand.  It is not a good-distance construction example: its third and fourth columns are identical, which produces the weight-two distance witness made explicit below.  Over $R=\F_2[x]/(x^{15}+1)$, let
\begin{align}
A_T
&=
\begin{pmatrix}
1&1&1&1\\
1&x^{10}&x^7&x^7
\end{pmatrix}
\in
R^{2\times4},
\end{align}
viewed as a map $R^4\to R^2$.  The polynomial factors as
\begin{align}
x^{15}+1&=(x+1)(x^2+x+1)(x^4+x+1)\notag\\
&\times(x^4+x^3+1)(x^4+x^3+x^2+x+1),
\end{align}
giving five packets: the factors $x^4+x+1$ and $x^4+x^3+1$ are reciprocal to each other, and the other three packets are self-reciprocal.  For a packet $g$ with representative root $\theta_g$, the packet matrix is
\begin{align}
(A_T)_g
&=
\begin{pmatrix}
1&1&1&1\\
1&\theta_g^{10}&\theta_g^7&\theta_g^7
\end{pmatrix},
\end{align}
and reducing modulo the five factors gives packet ranks $1,2,2,2,2$.  The first rank is the forced monomial $x+1$ rank defect, since every monomial evaluates to $1$ at $x=1$; all other packets have full row rank.  For the self-adjoint code $C=\LP(A_T,A_T^\dagger)$, the dimension accounting is
\begin{align}
k&=\underbrace{(4-1)^2+(2-1)^2}_{x+1\text{ packet}}\notag\\
&\quad+\underbrace{2(4-2)^2}_{x^2+x+1}\notag\\
&\quad+\underbrace{4(4-2)^2+4(4-2)^2}_{\text{reciprocal degree-4 pair}}\notag\\
&\quad+\underbrace{4(4-2)^2}_{\text{self-reciprocal degree-4 packet}}=66,
\end{align}
and the binary length is $n=(4^2+2^2)15=300$.

For one explicit full-row-rank packet, take $g=x^2+x+1$ and let $\omega$ be the image of $x$ in $K_g=\F_2[x]/(g)$.  Since $10\equiv7\equiv1\pmod 3$,
\begin{align}
(A_T)_g
&=
\begin{pmatrix}
1&1&1&1\\
1&\omega&\omega&\omega
\end{pmatrix}.
\end{align}
The primitive idempotent for this packet is
\begin{equation}
e_g=x^{14}+x^{13}+x^{11}+x^{10}+x^8+x^7+x^5+x^4+x^2+x,
\end{equation}
with support weight $10$.  With pivot columns $\{1,2\}$, the two free columns give field-kernel representatives $(0,1,1,0)$ and $(0,1,0,1)$ over $K_g$; the corresponding quotient coordinates are represented by the free coordinate vectors in $K_g^4/\row((A_T)_g)$, not by the ordinary cokernel of $(A_T)_g$.  Lifting the first kernel vector gives $e_g(0,1,1,0)$, with support weight $20$.  Tensoring it with a quotient coordinate gives a $Z$ logical representative in the left K\"unneth summand, and the reciprocal-packet trace pairing gives the conjugate $X$ representatives.

The lightest representative comes instead from the self-reciprocal degree-$4$ packet $g_s(x)=x^4+x^3+x^2+x+1$.  Let $\gamma$ be the image of $x$ in $K_{g_s}$.  Since the exponents reduce modulo $5$ as $0,0,2,2$,
\begin{align}
(A_T)_{g_s}
&=
\begin{pmatrix}
1&1&1&1\\
1&1&\gamma^2&\gamma^2
\end{pmatrix}.
\end{align}
With pivot columns $\{1,3\}$, the kernel representatives may be chosen with two nonzero block coordinates, for example $u=(1,1,0,0)$ or $u=(0,0,1,1)$.  The primitive idempotent itself has weight $12$, but a trace-dual multiplier $\beta^\#$ in the reciprocal packet produces the lighter coefficient block
\begin{align}
e_{g_s}\beta^\#
&=
x^{14}+x^{10}+x^9+x^5+x^4+1
\end{align}
of support weight $6$, so the corresponding $X$ representative, with two such nonzero block coordinates, has binary weight $2\cdot6=12$.

Collecting the weights: the generic-packet trace-dual certificate is $\Wgen=24$, the rank-one $x+1$ packet has representatives of weight $2\ell=30$, and the lightest representative overall is the weight-$12$ one above.  Hence
\begin{align}
\Wcert=30,\qquad \Dbasis=12,
\end{align}
and the duplicate-column distance witness of the next subsection gives $\Dbest=2$, so the parameter form of this diagnostic code is
\begin{align}
\llbracket n,k,d \rrbracket
=
\llbracket 300,66,\le 2 \rrbracket,
\end{align}
with the hierarchy $\Dbest=2<\Dbasis=12<\Wcert=30$: the constructed addressable basis contains $66$ conjugate pairs of weight at most $30$, hence $B(C)\le30$, while the true distance is at most $2$.

\subsection{Duplicate-column distance witness}
\label{app:seed-distance-witness-examples}

The seed-level distance-witness mechanism of Sec.~\ref{subsec:seed-distance-witnesses} is explicit for this seed: since the third and fourth columns of $A_T$ coincide, $u=\mathbf e_3+\mathbf e_4$ satisfies $A_Tu=0$ with $\wt_R(u)=2$.  Taking the monomial coordinate $q=\mathbf e_3$, the induced representative $Z=(u\otimes q,0)$ is verified after binary expansion to be a nontrivial logical operator, giving
\begin{align}
d(\LP(A_T,A_T^\dagger))\leq2.
\end{align}
A $20{,}000$-iteration randomized QDist search finds the same bound independently.

\subsection{Clean comparison seed}
\label{subsec:clean-comparison-seed}

A same-size monomial comparison seed isolates the defect above:
\begin{align}
A_{\rm cmp}
=\begin{pmatrix}
1&1&1&1\\
1&x&x^3&x^7
\end{pmatrix}
\in R^{2\times4}.
\end{align}
It has the same packet-rank profile $1,2,2,2,2$ as \seedlabel{T}{15}{2}{4}, hence the same self-adjoint parameters $n=300$, $k=66$, and $\Wcert=30$, but no monomial-equivalence distance witness: $\Dbasis=18$, the generic binary RREF quotient basis gives $\Drref=6$ and $\Wrref=108$, and the same $20{,}000$-iteration randomized QDist search budget gives $d(C)\le6$.  The weight-two distance witness of \seedlabel{T}{15}{2}{4} therefore comes from the chosen exponent pattern, not from the mere use of a $2\times4$ seed at $\ell=15$.  This monomial comparison seed is separate from the sparse \seedlabel{P}{15}{2}{4} example row in Table~\ref{tab:main}.

\section{Seed data and sparse supports}
\label{app:seed-data}

For a monomial seed, we display the exponent matrix $E$ such that the seed entry in position $(i,j)$ is $x^{E_{ij}}\in R_\ell$.  The monomial seeds used in Table~\ref{tab:main} are:
\begin{widetext}
\begin{align}
E\bigl(\seedlabel{A}{45}{3}{7}\bigr)
&=
\begin{pmatrix}
29&21&31&15&37&25&27\\
13&25&19&26&11&18&29\\
31&2&27&32&41&41&18
\end{pmatrix},\qquad
E\bigl(\seedlabel{A}{75}{3}{7}\bigr)
=
\begin{pmatrix}
0&71&73&68&33&50&47\\
38&39&60&26&18&1&23\\
73&6&5&42&20&22&73
\end{pmatrix},
\\
E\bigl(\seedlabel{A}{91}{3}{7}\bigr)
&=
\begin{pmatrix}
57&75&42&80&7&67&27\\
57&73&34&12&27&50&87\\
21&53&70&18&1&3&18
\end{pmatrix},
\qquad
E\bigl(\seedlabel{M}{45}{4}{10}\bigr)
=
\begin{pmatrix}
0&0&0&0&0&0&0&0&0&0\\
0&0&11&44&1&10&14&18&29&0\\
0&18&23&4&1&0&33&7&22&43\\
0&2&20&40&1&1&40&15&40&1
\end{pmatrix},
\\
E\bigl(\seedlabel{R}{75}{3}{7}\bigr)
&=
\begin{pmatrix}
0&0&0&0&0&0&0\\
0&12&60&5&36&28&70\\
0&37&13&9&20&27&17
\end{pmatrix},
\qquad
E\bigl(\seedlabel{R}{81}{3}{7}\bigr)
=
\begin{pmatrix}
0&0&0&0&0&0&0\\
0&30&73&24&1&46&0\\
0&18&75&10&24&60&6
\end{pmatrix},
\\
E\bigl(\seedlabel{M}{45}{5}{9}\bigr)
&=
\begin{pmatrix}
0&0&0&0&0&0&0&0&0\\
0&4&7&35&29&16&38&28&14\\
0&5&43&9&35&4&33&6&17\\
0&9&22&40&40&38&9&40&31\\
0&18&19&6&9&32&26&22&31
\end{pmatrix}.
\end{align}
\end{widetext}
The reference seeds \seedlabel{A}{45}{3}{7}, \seedlabel{A}{75}{3}{7}, and \seedlabel{A}{91}{3}{7} have rank one at $x+1$ and full row rank elsewhere, giving $k=744,1224,1480$ for $\ell=45,75,91$.

The small-scale monomial example from Table~\ref{tab:main} is
\begin{align}
E\bigl(\seedlabel{R}{21}{4}{5}\bigr)
&=
\begin{pmatrix}
0&0&0&0&0\\
0&19&1&19&10\\
0&11&8&10&18\\
0&4&17&15&10
\end{pmatrix}.
\end{align}

The even-order monomial example \seedlabel{E}{14}{3}{4} is obtained from a length-$7$ residue seed by the exponent-lift move of Appendix~\ref{app:even-lifts}: each exponent $a$ in the residue seed is replaced by $a+7q$ with $q\in\{0,1\}$.  The resulting exponent matrix over $R_{14}=\F_2[x]/(x^{14}+1)$ is
\begin{align}
E\bigl(\seedlabel{E}{14}{3}{4}\bigr)
&=
\begin{pmatrix}
0&0&0&0\\
0&2&13&4\\
0&10&12&1
\end{pmatrix}.
\end{align}
Equivalently, the residue matrix over $R_7$ is
\begin{align}
\begin{pmatrix}
0&0&0&0\\
0&2&6&4\\
0&3&5&1
\end{pmatrix},
\end{align}
with row-major lift vector $(0,0,0,0,0,0,1,0,0,1,1,0)$.

For sparse seeds, we display the support matrix $S$ such that the entry in position $(i,j)$ is $\sum_{a\in S_{ij}}x^a$.  Singleton supports therefore denote monomials.  The sparse parity-systematic probes in Table~\ref{tab:main} are:
\begin{widetext}
\begin{align}
S\bigl(\seedlabel{P}{75}{3}{7}\bigr)
&=
\begin{pmatrix}
\{18\}&\{12,19\}&\{63,67\}&\{60,72\}&\{41,66\}&\{34,39\}&\{33,44\}\\
\{14,32\}&\{48\}&\{11,20\}&\{7,37\}&\{26,53\}&\{39,53\}&\{11,45\}\\
\{15,31\}&\{56,60\}&\{62\}&\{30,54\}&\{20,31\}&\{8,45\}&\{27,31\}
\end{pmatrix},
\\
S\bigl(\seedlabel{P}{91}{3}{7}\bigr)
&=
\begin{pmatrix}
\{36\}&\{9,83\}&\{11,84\}&\{17,89\}&\{75,87\}&\{27,42\}&\{63,88\}\\
\{39,82\}&\{86\}&\{23,39\}&\{32,74\}&\{45,58\}&\{16,81\}&\{26,49\}\\
\{16,31\}&\{6,27\}&\{80\}&\{1,62\}&\{50,83\}&\{2,46\}&\{8,67\}
\end{pmatrix},
\\
S\bigl(\seedlabel{H}{41}{3}{8}\bigr)
&=
\begin{pmatrix}
\{25\}&\{11,23\}&\{1,22\}&\{16,30\}&\{14,26\}&\{0,5\}&\{17,21\}&\{19,22\}\\
\{2,17\}&\{10\}&\{2,22\}&\{20,28\}&\{0,23\}&\{13,32\}&\{14,22\}&\{25,32\}\\
\{17,30\}&\{18,40\}&\{34\}&\{6,20\}&\{3,20\}&\{1,28\}&\{23,36\}&\{2,24\}
\end{pmatrix},
\\
S\bigl(\seedlabel{H}{33}{3}{9}\bigr)
&=
\begin{pmatrix}
\{30\}&\{1,20\}&\{4,32\}&\{8,30\}&\{8,29\}&\{26,32\}&\{12,15\}&\{15,21\}&\{9,13\}\\
\{23,31\}&\{23\}&\{16,27\}&\{20,25\}&\{21,29\}&\{4,24\}&\{4,30\}&\{18,24\}&\{19,32\}\\
\{0,14\}&\{1,16\}&\{17\}&\{5,14\}&\{5,14\}&\{10,12\}&\{2,3\}&\{18,22\}&\{0,22\}
\end{pmatrix}.
\end{align}

\begin{align}
S\bigl(\seedlabel{P}{15}{2}{4}\bigr)
&=
\begin{pmatrix}
\{6\}&\{4,13\}&\{0,1\}&\{3,12\}\\
\{7,14\}&\{7\}&\{0,13\}&\{3,14\}
\end{pmatrix},
\\
S\bigl(\seedlabel{P}{35}{3}{4}\bigr)
&=
\begin{pmatrix}
\{31\}&\{5,25\}&\{6,10\}&\{19,31\}\\
\{0,10\}&\{5\}&\{1,2\}&\{4,30\}\\
\{25,28\}&\{11,20\}&\{24\}&\{21,31\}
\end{pmatrix},
\\
S\bigl(\seedlabel{P}{31}{2}{5}\bigr)
&=
\begin{pmatrix}
\{8\}&\{3,22\}&\{14,17\}&\{27,28\}&\{10,28\}\\
\{20,27\}&\{0\}&\{0,2\}&\{11,23\}&\{16,29\}
\end{pmatrix}.
\end{align}

\end{widetext}


\section{Even lifts and primary packets}
\label{app:even-lifts}

The odd-order assumption is the semisimplicity assumption that makes the main theory especially transparent: the group algebra decomposes into finite fields, and packetwise ranks determine the relevant kernel and quotient dimensions.  Even-order lifts are nonsemisimple.  The algebraic ingredients below---primary decomposition, chain-ring Smith forms, and Frobenius-ring duality---are classical~\cite{vanLint1991,LingSole2003,DinhLopez2004,Wood1999}; the statements organize them for lifted-product homology.  This appendix establishes which parts of logical spectroscopy survive at theorem level without semisimplicity, and which parts are deferred to the follow-up program of Ref.~\cite{LeeCodeSearchInPrep}.

\begin{proposition}[Primary packet decomposition]
\label{prop:primary-decomposition}
Let $G$ be a finite Abelian group of arbitrary order.  Then $\RG=\F_2[G]$ decomposes as a finite product of local rings,
\begin{align}
\RG\simeq\prod_\Omega C_\Omega,
\quad
1=\sum_\Omega e_\Omega,
\quad
e_\Omega e_\Lambda=\delta_{\Omega\Lambda}e_\Omega,
\end{align}
called the \emph{primary packets}.  The product complex of Sec.~\ref{sec:finite-abelian-basis} decomposes as $Q_\bullet\simeq\prod_\Omega Q_{\bullet,\Omega}$, so $\cH_1=\bigoplus_\Omega\cH_{1,\Omega}$, and Theorem~\ref{thm:orbit-erasure-packets} holds verbatim for whole-orbit erasures, with $d_{Z,\Omega}$ and $d_{X,\Omega}$ computed as $\F_2$-dimensions of the local packet restrictions.
\end{proposition}

\begin{proof}
A finite commutative ring is a finite product of local rings~\cite{DummitFoote2004}, and the corresponding complete orthogonal idempotents give the displayed decomposition.  Free modules, tensor products, and homology commute with this finite product, and a whole-orbit erased subspace is an $\RG$-submodule, so column restriction commutes with the decomposition as well.
\end{proof}

Because each factor $C_\Omega$ is local, it has no nontrivial idempotents; hence the displayed $e_\Omega$ are the primitive central idempotents of $\RG$, while $C_\Omega$ is the associated primary packet.

Packet labels, packet attribution, and exact whole-orbit erasure spectroscopy therefore require no semisimplicity; only the finite-field rank formulas do.  In particular, Proposition~\ref{prop:primary-decomposition} applies to bivariate-bicycle codes~\cite{BravyiCross2024}, which are the scalar-seed case $\LP(a,b)$ with $a,b\in\RG$ over even-order two-variable groups.

For the rest of this appendix we take the single-variable case.  If $\ell=2^s m$ with $m$ odd and $e=2^s$, then
\begin{align}
x^\ell+1
&=
(x^m+1)^{e},
\end{align}
the repeated-root cyclic regime~\cite{vanLint1991,DinhLopez2004}, and the primary packets are the finite chain rings
\begin{align}
C_g=\F_2[x]/(g^{e}),
\qquad
g\mid x^m+1,
\end{align}
with maximal ideal generated by $\pi_g=g\bmod g^{e}$, $\pi_g^{e}=0$, and residue field $K_g=\F_2[x]/(g)$; this is the chain-ring constituent regime of quasi-cyclic codes~\cite{LingSole2003}.  Over a finite chain ring every matrix admits a Smith normal form~\cite{DinhLopez2004,Elsheikh2012}: invertible row and column operations over $C_g$ bring $A_g$ to a direct sum of elementary pieces
\begin{align}
D_a: C_g\xrightarrow{\ \pi_g^{\,a}\ }C_g\ \ (0\le a\le e),
\quad
S: C_g\to 0,
\quad
T: 0\to C_g,
\end{align}
where $S$ and $T$ are unmatched source and target coordinates and $D_e\simeq S\oplus T$ is the zero square block.  The multiset of exponents $a_i$, together with the $S$ and $T$ multiplicities, is the \emph{nilpotent Smith profile} of $A_g$; it refines the residue rank, which counts only the exponents $a_i=0$.

\begin{theorem}[Even one-variable dimension formula]
\label{thm:even-dimension}
Let $\ell=2^sm$ with $m$ odd, and let $A,B$ be seeds over $R_\ell$.  Write $A_g\simeq\bigoplus_i E_i$ and $B_g\simeq\bigoplus_j F_j$ in elementary pieces over $C_g$.  Then
\begin{align}
k=\dim_{\F_2}\cH_1\bigl(\LP(A,B)\bigr)
=\sum_{g\mid x^m+1}\deg(g)\sum_{i,j}h(E_i,F_j),
\label{eq:even-dimension-formula}
\end{align}
where
\begin{align}
h(D_a,D_b)&=2\min(a,b),
\nonumber\\
h(D_a,S)=h(D_a,T)&=a,
\nonumber\\
h(S,D_b)=h(T,D_b)&=b,
\nonumber\\
h(S,T)=h(T,S)&=e,
\quad
h(S,S)=h(T,T)=0.
\end{align}
\end{theorem}

This table is consistent with the convention $D_e\simeq S\oplus T$.  For example,
\begin{align}
h(D_e,D_b)
&=
2b
=
h(S,D_b)+h(T,D_b),
\nonumber\\
h(D_e,D_e)
&=
2e
=
h(S,T)+h(T,S).
\end{align}

\begin{proof}
By Proposition~\ref{prop:primary-decomposition} it suffices to compute the $K_g$-length of the packet homology and multiply by $\dim_{\F_2}K_g=\deg g$, since every finite $C_g$-module has a composition series with factors $K_g$.  The product complex is additive in each argument under direct sums, so it suffices to evaluate $h$ on elementary pairs; write $C=C_g$, $\pi=\pi_g$, and $m_0=\min(a,b)$.

For $D_a$ against $D_b$, the packet complex is $C\xrightarrow{\partial_2}C^2\xrightarrow{\partial_1}C$ with $\partial_2 z=(\pi^{b}z,\pi^{a}z)$ and $\partial_1(u,v)=\pi^{a}u+\pi^{b}v$.  Then $\Im\partial_1=(\pi^{a},\pi^{b})=(\pi^{m_0})$ has length $e-m_0$, so $\ker\partial_1$ has length $2e-(e-m_0)=e+m_0$; and $\ker\partial_2=\operatorname{ann}(\pi^{m_0})=(\pi^{\,e-m_0})$ has length $m_0$, so $\Im\partial_2$ has length $e-m_0$.  The homology therefore has length $2m_0$.

The remaining entries are single-length counts.  For $D_a$ against $T$, the middle term is $C$ with $\partial_1=\pi^{a}$ and $\partial_2=0$, so the homology is $\ker\pi^{a}=(\pi^{\,e-a})$ of length $a$; for $S$ against $D_b$, the middle term is $C$ with $\partial_1=0$ and $\Im\partial_2=(\pi^{b})$, giving length $b$; for $S$ against $T$, both boundary maps vanish on the middle term $C$, giving length $e$; and $S$ against $S$ or $T$ against $T$ has zero middle term.  The mixed entries $h(D_a,S)$ and $h(T,D_b)$ follow from the symmetric counts.
\end{proof}

\begin{corollary}[Reduction to the semisimple case]
\label{cor:even-odd-reduction}
For $s=0$ the packet ring is the field $K_g$, all Smith exponents vanish, and the only nonzero table entries are $h(S,T)=h(T,S)=1$.  Counting the $S$ pieces of $A_g$ as $n_A-s_A^{g}$ and the $T$ pieces as $r_A-s_A^{g}$, and likewise for $B$, recovers Eq.~\eqref{eq:finite-abelian-K}.  For $B=A^\dagger$, the involution identifies the nilpotent Smith profile of $B$ at $g^{e}$ with that of $A$ at $(g^\vee)^{e}$, generalizing the self-adjoint specialization, Eq.~\eqref{eq:self-adjoint-K}.
\end{corollary}

\vspace{3pt}\noindent\textbf{Scalar seeds and generalized bicycle codes.}
For $1\times1$ seeds $a,b\in R_\ell$, the Smith exponent of $a$ in the packet $C_g$ is the multiplicity of $g$ in $a$, so Eq.~\eqref{eq:even-dimension-formula} collapses to
\begin{align}
k=2\deg\gcd(a,b,x^\ell+1),
\end{align}
the known dimension formula for generalized-bicycle codes at arbitrary $\ell$~\cite{KovalevPryadko2013,PanteleevKalachev2021}.  Theorem~\ref{thm:even-dimension} is its rectangular Smith-profile generalization, interpolating between this scalar anchor and the semisimple rank formula of Corollary~\ref{cor:even-odd-reduction}.

The even complex can also be resolved recursively rather than packetwise.  Every even-lift complex is a square-zero thickening of the complex at half the length, in the following canonical sense.

\begin{proposition}[Halving exact sequence]
\label{prop:halving-ses}
Let $\ell$ be even, let $u=x^{\ell/2}+1\in R_\ell$, and let $A,B$ be arbitrary seeds over $R_\ell$ with reductions $\bar A,\bar B$ modulo $u$.  Then $u^2=0$ and $\operatorname{ann}(u)=(u)$, and multiplication by $u$ induces an isomorphism of complexes $Q_\bullet/uQ_\bullet\simeq uQ_\bullet$, both identified with the half-length product complex $\bar Q_\bullet$ built from $\bar A,\bar B$ over $R_{\ell/2}\simeq R_\ell/(u)$.  Hence there is a canonical short exact sequence of complexes
\begin{align}
0\longrightarrow \bar Q_\bullet \xrightarrow{\ \cdot u\ } Q_\bullet \longrightarrow \bar Q_\bullet \longrightarrow 0,
\end{align}
whose long exact homology sequence gives the exact dimension recursion
\begin{align}
\dim_{\F_2}\cH_1(\ell)
=
2\dim_{\F_2}\cH_1(\ell/2)
-\operatorname{rank}\delta_1
-\operatorname{rank}\delta_2,
\label{eq:bockstein-recursion}
\end{align}
where $\delta_i:\cH_i(\bar Q)\to\cH_{i-1}(\bar Q)$ are the connecting (Bockstein) maps associated with this short exact sequence of complexes~\cite{Weibel1994}.  The same construction applies to $u=x_i^{L_i/2}+1$ one variable at a time in a multicyclic group, so iterating over each factor of two in each variable reduces any even Abelian lift to its odd skeleton through a finite tower of such sequences.
\end{proposition}

\begin{proof}
First, $u^2=x^\ell+1=0$.  Since $m$ is odd, $x^m+1$ is squarefree over $\F_2$: its derivative is $mx^{m-1}=x^{m-1}$, which is coprime to $x^m+1$.  Thus each irreducible factor $g$ divides $x^m+1$ exactly once.  In each primary packet $C_g=K_g[\pi]/(\pi^e)$, the element $x^m+1$ has $\pi$-valuation one, so $u=(x^m+1)^{e/2}$ maps to a unit times $\pi^{e/2}$ and $\operatorname{ann}(u)=(u)$ packetwise, hence globally; in the multivariable case $u$ lies in one tensor factor and the annihilator computation is unchanged.  Multiplication by $u$ is $R_\ell$-linear, so it commutes with the boundary maps, kills $uQ_\bullet$, and induces an isomorphism $Q_\bullet/uQ_\bullet\to uQ_\bullet$ by $\operatorname{ann}(u)=(u)$.  Both complexes carry the differential reduced modulo $u$, which is the boundary map of the half-length complex of the reduced seeds under $R_\ell/(u)\simeq R_{\ell/2}$.  The recursion is the standard rank count in the long exact sequence: $\dim\cH_1(Q)$ equals $[\dim\cH_1(\bar Q)-\operatorname{rank}\delta_2]$ from the image of the subcomplex plus $[\dim\cH_1(\bar Q)-\operatorname{rank}\delta_1]$ from the kernel of $\delta_1$ in the quotient.
\end{proof}

\vspace{5pt}\noindent\textbf{Computing the Bockstein maps.}
The connecting maps are explicit.  Given a class $[\bar z]\in\cH_i(\bar Q)$, lift $\bar z$ to any $z\in Q_i$; then $\partial\bar z=0$ forces $\partial z\in uQ_{i-1}$, and $\delta_i[\bar z]=[w]$ for the unique $w\in\bar Q_{i-1}$ with $uw=\partial z$, well defined because $\operatorname{ann}(u)=(u)$.  The name reflects the analogy with the classical Bockstein of the coefficient sequence $0\to\Z_2\to\Z_4\to\Z_2\to0$: $R_\ell$ is a square-zero extension of $R_{\ell/2}$ in the same way.  In practice each $\delta_i$ is a finite binary matrix: fix a basis of $\cH_i(\bar Q)$, lift each basis cycle, apply $\partial$, divide by $u$, and expand the resulting classes in a basis of $\cH_{i-1}(\bar Q)$; $\rank\delta_i$ is then an ordinary $\F_2$ rank.

\vspace{5pt}\noindent\textbf{Bivariate-bicycle example.}
As a numerical check on the nonsemisimple statement, consider the standard bivariate-bicycle code over
\begin{align}
R_{12,6}
=
\F_2[x,y]/(x^{12}+1,y^6+1)
\end{align}
with scalar seeds
\begin{align}
a=x^3+y+y^2,
\qquad
b=y^3+x+x^2,
\end{align}
in the convention $H_X=(a\ b)$ and $H_Z=(b^\dagger\ a^\dagger)$~\cite{BravyiCross2024}.  Direct binary expansion gives $n=2\cdot12\cdot6=144$, $\rank H_X=\rank H_Z=66$, and hence $k=12$, the known $\llbracket144,12,12\rrbracket$ example.  The Bockstein tower reduces this code to the odd skeleton over $\Z_3\times\Z_3$ by halving $x$ twice and $y$ once; read bottom-up, the same tower recovers $k=12$.  Table~\ref{tab:bb-bockstein} lists the binary ranks of the connecting maps computed by the construction above.

\begin{table*}[t]
\caption{Bockstein validation for the $\llbracket144,12,12\rrbracket$ bivariate-bicycle code.  The row $(L_x,L_y)\leftarrow(L'_x,L'_y)$ means that the full complex at $(L_x,L_y)$ is the square-zero thickening of the half-length complex at $(L'_x,L'_y)$ in the displayed variable.  The recursion $k_{\rm full}=2k_{\rm half}-\rank\delta_1-\rank\delta_2$ agrees with a direct binary rank computation at every stage.}
\label{tab:bb-bockstein}
\small
\begin{tabular}{cccccccc}
\toprule
full & half & var. & $(h_2,h_1,h_0)_{\rm half}$ & $\rank\delta_1$ & $\rank\delta_2$ & $k_{\rm rec}$ & $k_{\rm direct}$\\
\midrule
$(3,6)$ & $(3,3)$ & $y$ & $(4,8,4)$ & 4 & 4 & 8 & 8\\
$(6,6)$ & $(3,6)$ & $x$ & $(4,8,4)$ & 2 & 2 & 12 & 12\\
$(12,6)$ & $(6,6)$ & $x$ & $(6,12,6)$ & 6 & 6 & 12 & 12\\
\bottomrule
\end{tabular}
\end{table*}

The same computation gives primary-packet erasure attribution.  The five primary packets are the nilpotent thickenings of the odd-skeleton Frobenius packets with representatives $(0,0)$, $(1,0)$, $(0,1)$, $(1,1)$, and $(1,2)$; their binary dimensions in $R_{12,6}$ are $(8,16,16,16,16)$.  Erasing either of the two complete physical block-orbits gives $d_Z=d_X=6$, with packet-total vector $d_Z+d_X=(0,0,0,4,8)$.  Erasing both block-orbits gives $d_Z=d_X=12$, with packet-total vector $(0,0,0,8,16)$.  Thus the exact erased-logical contribution is concentrated in the two nontrivial mixed primary packets, and the packet totals agree with the full binary erasure-rank formula.

Proposition~\ref{prop:halving-ses} complements Theorem~\ref{thm:even-dimension} in two ways.  For multivariable even lifts, where the primary packets are not chain rings and no Smith profile is available, Eq.~\eqref{eq:bockstein-recursion} still computes the logical dimension exactly by a finite tower of Bockstein-rank computations over successively smaller groups.  And the connecting maps depend on the extension data of the lift, not only on the residue seed, which is precisely the design freedom exploited by the exponent lifts below.

\begin{lemma}[Reciprocal primary pairing]
\label{lem:primary-pairing}
The group-inversion involution permutes primary packets by $g^{e}\leftrightarrow(g^\vee)^{e}$.  Under the physical binary pairing, $R_\ell^m$ decomposes into orthogonal reciprocal pairs of primary blocks, the induced pairing between $e_g R_\ell^m$ and $e_{g^\vee}R_\ell^m$ is perfect, and the CSS homology--cohomology pairing decomposes into perfect reciprocal primary blocks.  A complete conjugate basis with primary packet labels is therefore obtained, as in Theorem~\ref{thm:finite-abelian-basis}, by inverting each finite binary Gram block.
\end{lemma}

\begin{proof}
The global binary pairing $[x^0]\sum_i a_i(x)b_i(x^{-1})$ is nondegenerate, and the involution sends $e_g$ to $e_{g^\vee}$, so block orthogonality holds exactly as in Lemma~\ref{lem:reciprocal-binary-pairing}; a nondegenerate pairing restricted to an orthogonal decomposition is perfect on each reciprocal pair of blocks.  The homology--cohomology pairing between $\cL_Z$ and $\cL_X$ is perfect over $\F_2$ for any CSS code, and packet orthogonality splits it into the displayed blocks.  Structurally, $C_g$ is a finite Frobenius ring and the block pairing is a socle-coefficient functional~\cite{Wood1999}, replacing the field trace of Lemma~\ref{lem:reciprocal-trace-pairing}.
\end{proof}

The even-order basis entries for \seedlabel{E}{14}{3}{4} in Table~\ref{tab:main} use the following direct binary implementation of Lemma~\ref{lem:primary-pairing}.  First expand the CSS code and compute the binary kernel and stabilizer spaces on each side.  For each primary idempotent $e_g$, project both spaces blockwise into $e_gR_\ell$ and choose a deterministic quotient basis in the projected kernel modulo the projected stabilizers.  For \seedlabel{E}{14}{3}{4}, this gives primary packet dimensions
\begin{align}
(x+1)^2:14,\quad
(x^3+x+1)^2:6,\quad
(x^3+x^2+1)^2:6.
\end{align}
The resulting $26$ $Z$- and $26$ $X$-representatives have a nonsingular binary Gram matrix; inverting that finite matrix gives a conjugate primary-packet-labeled basis with $\Dbasis=24$ and $\Wcert=168$.  A deterministic greedy stabilizer-coset reduction, not used for the table certificate, lowers the maximum support to $112$ while preserving the Gram identity.  This computation is included to show that primary-packet labels and binary conjugate pairing already work at small even order, even though the odd-order trace-dual finite-field construction is not being claimed in this nonsemisimple case.

The label set becomes (primary packet, nilpotent layer, index): the filtration $\pi_g^{j}\cH_{1,g}$ refines each packet block, with a graded splitting fixed in practice by the Smith basis.  Two sparse design moves act on the nilpotent layers at fixed odd residue.  Since $x^{m}=1+\eta$ with $\eta=x^{m}+1$ nilpotent, the exponent lifts
\begin{align}
x^{a}\mapsto x^{a+qm}=x^{a}(1+\eta)^{q},
\qquad
q=0,\ldots,e-1,
\end{align}
change only the nilpotent expansion while preserving monomial support and the residue seed---equivalently, they change the Bockstein maps of Proposition~\ref{prop:halving-ses} at fixed residue seed; similarly, sparse nilpotent repairs $p(x)\mapsto p(x)+(x^{m}+1)h(x)$ modify higher layers only.  These moves, the associated layer profiles as search objectives, and the multivariable even case---whose primary packets are local but not chain rings, so that Proposition~\ref{prop:primary-decomposition} and exact local binary expansion apply while closed Smith-profile formulas do not---are developed in Ref.~\cite{LeeCodeSearchInPrep}.

\section{Packet-local representatives and distance}
\label{app:packet-local-distance}

The true distance is optimized over all logical representatives, while the spectral construction begins from representatives localized in single packets.  These are different optimization problems.  A global low-weight distance witness may be packet-resolved, in the sense that it decomposes into packet components, without being localized in one packet.  Physical cancellations between several packet components can make such a distance witness lighter than any representative that preserves a single packet label.

A small counterexample is $R=\F_2[x]/(x^3+1)$ with seed $A=(1\;1)$ and the self-adjoint product $\LP(A,A^\dagger)$.  The seed relation $\mathbf e_1+\mathbf e_2\in\ker A$ gives a verified weight-$2$ logical, and a direct binary check rules out weight-one logicals, so $d=2$.  However, the $x+1$ packet idempotent has weight $3$, and the $x^2+x+1$ packet has minimum nonzero block weight $2$.  Since a nonzero packet-local logical needs at least two nonzero group-algebra blocks, the best packet-local logical has weight $4$.  This gives
\begin{align}
d
&=
2
<
d_{\rm pkt\mbox{-}loc}
=
4.
\end{align}

Packet-local spectral representatives should therefore be viewed as an addressable coordinate basis; they need not contain minimum-weight logical operators.  Low-weight seed distance witnesses, minor distance witnesses, sparse syzygies, and QDist outputs are best used to tighten distance upper bounds or to build hybrid bases from a larger verified distance-witness pool.  If the goal is instead to preserve the packet labels of the constructed basis, then the optimization must remain inside each packet, or inside a deliberately chosen reciprocal-stable packet union.

There are several packet-preserving refinements.  One can search for sparse vectors in $\ker A_g$, sparse quotient representatives in $K_g^n/\row(A_g)$, sparse cokernel representatives in $\coker A_g$, and analogous reciprocal-packet data for the conjugate $X$ representatives.  One can also optimize over pivots, coefficient bases, trace-dual bases in reciprocal packets, and lifted packet elements $e_g\beta$ with $\beta\in K_g$.  These refinements are stronger than knowing only the packet rank: the rank counts packet logical coordinates, while packet-local sparsity and idempotent-support profiles ask how expensive those coordinates are after lifting.

Stabilizer-coset reduction can also be made packet-preserving, but only if the added stabilizers are themselves packet-local.  If $L$ is localized in packet $g$, meaning
\begin{align}
e_gL = L, \quad e_hL = 0
\qquad
(h\neq g),
\end{align}
then replacing $L$ by $L+s$ preserves this packet label only when
\begin{align}
s \in e_g\operatorname{Stab}.
\end{align}
A generic physical stabilizer row is usually broad-spectrum, so adding it can reduce weight while destroying the sharp packet label.  Packet-local coset reduction is therefore a more constrained problem than ordinary stabilizer-coset reduction.

For a conjugate basis, the projected replacements announced in Sec.~\ref{sec:representatives-certificates} are
\begin{align}
Z\mapsto Z+sH_Z\Pi_\Omega,
\qquad
X\mapsto X+tH_X\Pi_{\Omega^\vee},
\end{align}
where $\Pi_\Omega$ is the binary projector induced by multiplying every group-algebra block coordinate by the primitive idempotent $e_\Omega$.  The projected stabilizers are still binary stabilizers, so the logical cosets and the conjugate Gram matrix are unchanged.  Writing $\Pi_\lambda=\Pi_{\Omega_\lambda}$ and $\Pi^\vee_\lambda=\Pi_{\Omega_\lambda^\vee}$ for the projectors attached to the packet label $\Omega_\lambda$ of $Z_\lambda$, these replacements define a packet-labeled refinement of the basis-width certificate:
\begin{align}
w^Z_\lambda
&=
\min_s\wt\!\left(Z_\lambda+sH_Z\Pi_\lambda\right),
\nonumber \\
w^X_\lambda
&=
\min_t\wt\!\left(X_\lambda+tH_X\Pi^\vee_\lambda\right),
\nonumber \\
W_{\mathrm{pkt\mbox{-}red}}^{\mathrm{cert}}
&=
\max_\lambda\max\{w^Z_\lambda,w^X_\lambda\}.
\label{eq:packet-reduced-certificate}
\end{align}
The minimizations are coset-leader problems with an additional packet constraint and with physical Hamming weight as the objective; they decompose by packet label and can be run as a local improvement step after the deterministic spectral basis has been constructed.  If the resulting binary representatives are reverified, then
\begin{align}
B(C)
\leq
W_{\mathrm{pkt\mbox{-}red}}^{\mathrm{cert}}
\leq
\Wcert.
\end{align}

Finally, one need not insist on single packets.  Choose a reciprocal-stable union of packets $\mathcal P$ and define the coarser projector
\begin{align}
e_{\mathcal P}
&=
\sum_{g\in\mathcal P}
e_g.
\end{align}
Representatives built inside $e_{\mathcal P}R_\ell$ have coarser spectral labels, but they may be substantially lighter because several packet components can combine and cancel in physical coordinates.  The price is a larger conjugate-pairing block and a larger Gram inversion.  This gives an explicit tradeoff between spectral resolution and physical support.

\section{Seed-level prefilters before QDist}
\label{sec:seed-prefilters}

The spectral construction gives an addressable conjugate logical basis and a verified basis-width certificate; the true distance is a separate question.  This section assembles inexpensive seed-level diagnostics that can be applied before QDist~\cite{QDistRnd2023} or other full-code distance samplers.

The diagnostic sequence reserves full-code distance search for selected seeds.  Packet ranks, basis-width certificates, seed-kernel relations, minor distance witnesses, proper-divisor quotient lifts, short-cycle filters, and exact orbit-erasure ranks are less expensive than unrestricted distance search.  These quantities are algebraic prefilters or upper-bound mechanisms.  When verified, they can reject low-distance seeds and identify the seed-level or quotient-level relation behind some QDist distance witnesses.

A search instance fixes a lift group, seed shapes, a support budget for each group-algebra entry, and optionally a hardware displacement palette; natural search moves include changing monomial exponents, toggling binomial corrections, repairing the parity skeleton $A(1)$, and restricting displacements to allowed hardware moves.  A typical prefilter rejects candidates with the wrong rate, an undesirable packetwise rank profile, erased-logical dimension concentrated in a vulnerable packet, very small $\Dbasis$, very large $\Wcert$, or poor whole-orbit erasure diagnostics; these spectral diagnostics are then combined with ordinary LDPC graph filters and late-stage QDist, BP-OSD, or erasure-decoder tests on selected seeds.

Whenever a diagnostic entry is used as a reported value in $\Dbest$, we apply the distance-witness verification convention of Sec.~\ref{subsec:seed-distance-witnesses}.

\begin{table*}[!t]
\caption{Seed-level diagnostics before QDist.  The entries are prefilters or ranking metrics whose purpose is to reject obvious failures and identify seeds worth expensive distance and decoder studies.}
\label{tab:seed-diagnostic-summary}
\small
\renewcommand{\arraystretch}{1.16}
\setlength{\tabcolsep}{4pt}
\begin{tabular}{@{}L{0.16\textwidth}L{0.17\textwidth}L{0.31\textwidth}L{0.28\textwidth}@{}}
\toprule
Diagnostic layer & Typical quantities & What it detects & Computability and references\\
\midrule \midrule

Gauge and duplicate filters &
canonical seed form, monomial-equivalence tests &
Duplicate seeds and shifted duplicate rows or columns that can induce weight-$2$ distance witnesses after verification. &
Hash normalized exponent patterns; exponent-matrix and QC-LDPC search precedents are in Refs.~\cite{Fossorier2004,SmarandacheVontobel2012}, with related finite-length LP constraints in Ref.~\cite{Raveendran2025}.\\ \midrule

Seed-level distance-witness filters &
$\Dseed$, $\mathcal C_{r+1}$, $\Dminor$, $\Dquot$, $\Dsyz$ &
Sparse kernel relations from duplicate columns, Cramer/minor spectra, proper-divisor quotient lifts, bounded syzygies, or low-weight seed-code words. &
Efficient for bounded seed shape and distance-witness size; permanent/minor distance witnesses are standard QC-LDPC tools~\cite{SmarandacheVontobel2012}.\\ \midrule

Packet and basis diagnostics &
packet ranks, $c_g$, $\Dbasis$, $\Wcert$ &
Logical dimension, packet-attributed labels, packet-local dependency sizes, and constructed-basis width. &
Finite-field row reduction after factoring $x^\ell+1$, plus bounded packet searches; these are the spectral diagnostics supplied here.\\ \midrule

Graph and erasure diagnostics &
$N_4,N_6,N_8,\ldots$, $s_{\rm fail}$, $\rho_{s\rm orb}$ &
Short cycles, Tanner-graph defects, and whole-orbit erasures supporting logical operators. &
Exponent-matrix girth tests are classical~\cite{Fossorier2004,2021necessarysufficientgirthconditions}; orbit-erasure attribution follows from the present packet decomposition, with erasure-decoder motivation from Ref.~\cite{Yao2024}.\\ \midrule

Decoder-stage tests &
QDist, BP-OSD, erasure decoder &
Final full-code distance and decoder checks on seeds passing the algebraic filters. &
Expensive late-stage tests; BP-OSD and iterative decoder studies provide natural evaluation tools~\cite{PanteleevKalachev2021,Roffe2023,Pradhan2025,GolowichGuruswami2024}.\\

\bottomrule
\end{tabular}
\end{table*}

Figure~\ref{fig:seedwitnesses} isolates the distance-witness mechanism that is kept separate from the construction-example comparison in Fig.~\ref{fig:codesignlandscape}.

\begin{figure*}[t]
\includegraphics[width=.96\textwidth]{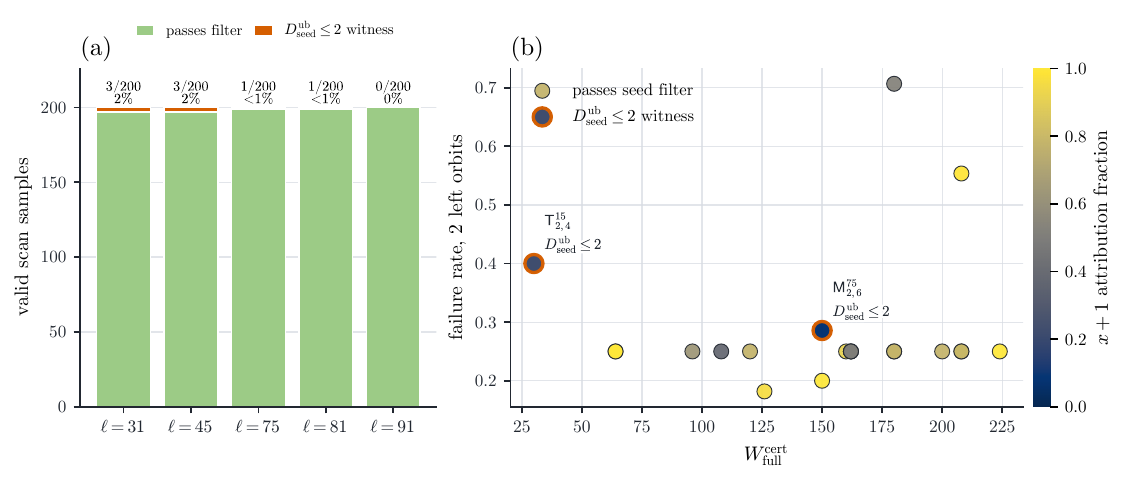}
\caption{Seed distance-witness diagnostic. (a) Frequency of low-cost monomial-equivalence distance witnesses in the reproducible random scans of $3\times7$ seeds for $\ell=31,45,75,81,91$. Each stacked bar counts valid scan samples at that lift length: the green portion passes the seed filter, while the orange cap is the subset with a verified $\Dseed\le2$ distance witness. The annotations above the bars give the distance-witness count out of the valid scan samples and the corresponding percentage. (b) Certificate-erasure coordinates for the construction examples together with diagnostic examples. Orange rings mark seeds with $\Dseed\le2$; filled points pass the seed filter. The vertical color gives the fraction of total erased-logical dimension attributed to the $x+1$ packet under two-left-block-orbit erasure. The distance-witness points are retained as diagnostics but are not treated as construction examples in Fig.~\ref{fig:codesignlandscape}.}
\label{fig:seedwitnesses}
\end{figure*}

Figure~\ref{fig:seedwitnesses} counts only the shifted-duplicate, monomial-equivalence mechanism behind very small $\Dseed$ distance witnesses.  It does not count Cramer/minor or proper-divisor quotient-lift distance witnesses; those are the separate algebraic candidates studied below and are the source of several $\Ddiag$ entries in Table~\ref{tab:main}.

Throughout this section, let
\begin{align}
R_\ell
&=
\F_2[x]/(x^\ell+1),
\end{align}
with $\ell$ odd, and let
\begin{align}
A:R_\ell^n
&\longrightarrow
R_\ell^r
\end{align}
be a single-variable quasi-cyclic seed.  For a polynomial vector $u\in R_\ell^n$, the group-basis weight $\wt_R(u)$ is as in Eq.~\eqref{eq:group-basis-weight}.

\subsection{Seed-kernel distance witnesses}

The simplest diagnostic subclass uses a sparse kernel vector of the seed map through the $Z$-type cycle $Z=(u\otimes q,0)$ of Sec.~\ref{subsec:seed-distance-witnesses}.  Let $\mathcal U_A$ be the set of nonzero $u\in\ker A$ for which some monomial coordinate $q$ gives a verified nonzero product class.  Define
\begin{align}
\Dseed(A)
&=
\min_{u\in\mathcal U_A}\wt_R(u),
\label{eq:Dseed-witness}
\end{align}
with the convention that the value is unavailable when $\mathcal U_A$ is empty.  The same test applied to $A^\dagger$ gives row-side witnesses.  Whenever either verified value is finite,
\begin{align}
d(\LP(A,A^\dagger))
\leq
\min\{
\Dseed(A),
\Dseed(A^\dagger)
\}.
\label{eq:seed-witness-bound}
\end{align}

It remains to explain the choice of the monomial coordinate $q$, and why the nontriviality condition is kept.  If $0\neq u\in\ker A$, one can choose a monomial coordinate vector $q$ with nonzero binary pairing against $u$: such a $q$ exists because $u$ has at least one nonzero monomial coefficient, and it detects a nonzero quotient component of $\coker A^\dagger$, since
\begin{align}
\langle A^\dagger s,u\rangle_{\rm bin}
&=
\langle s,Au\rangle_{\rm bin}
=
0
\end{align}
for all $s$, while $\langle q,u\rangle_{\rm bin}=1$.  However, the product representative over $R_\ell$ is nonzero only through common packet components of the two tensor legs.  A sparse seed-kernel vector is therefore included as a distance witness only after checking that the induced K\"unneth class is nonzero, equivalently after verifying the expanded binary representative as a nontrivial logical operator.  A sparse classical seed-codeword is thus a candidate low-distance witness; it becomes a quantum distance witness only after the packet or binary nontriviality check.

\subsection{Monomial-equivalent rows and columns}

For monomial seeds this candidate test becomes extremely inexpensive.  If $A_{ij}=x^{a_{ij}}$ and two columns $j,k$ satisfy
\begin{align}
a_{ij}-a_{ik}
&\equiv
\delta
\pmod\ell
\qquad
\text{for every row } i,
\end{align}
then $A_{\cdot,j}=x^\delta A_{\cdot,k}$ and
\begin{align}
u=\mathbf e_j+x^\delta \mathbf e_k
\end{align}
lies in $\ker A$ with $\wt_R(u)=2$.  If the induced product representative is verified as a nontrivial logical operator, then
\begin{align}
d\bigl(\mathrm{LP}(A,A^\dagger)\bigr)
&\le
2.
\end{align}
The row version is the same test applied to $A^\dagger$.  Appendix~\ref{app:seed-distance-witness-examples} gives the worked toy example, so the present appendix uses this filter only as one entry in the pre-QDist diagnostic table.  Related finite-length lifted-product constraints are studied in Ref.~\cite{Raveendran2025}.

\subsection{Minor and permanent distance witnesses}

The duplicate-column test is the smallest case of a more general seed-kernel candidate.  Let $A$ be an $r\times n$ seed over $R_\ell$.  For any subset $S\subset\{1,\ldots,n\}$ with $|S|=r+1$, define a vector $u_S\in R_\ell^n$ supported on $S$ by
\begin{align}
(u_S)_j
&=
\det A_{S\setminus j},
\qquad
j\in S,
\\
(u_S)_j
&=
0,
\qquad
j\notin S,
\label{eq:minor-witness}
\end{align}
where $A_{S\setminus j}$ is the $r\times r$ submatrix obtained by deleting column $j$ from $A_S$.  Since the characteristic is two, signs in the usual Laplace relation disappear.  It follows that
\begin{align}
Au_S
&=
0.
\end{align}
If $u_S$ is nonzero and sparse, it gives a seed-kernel candidate with candidate weight
\begin{align}
w_{\rm minor}(S)
&=
\sum_{j\in S}
\wt_R\!\left(\det A_{S\setminus j}\right),
\end{align}
with the convention that all-zero candidates are discarded.  It becomes a distance witness only when the induced product class is nontrivial after packet or binary verification.  Define the verified minor distance-witness value by
\begin{align}
&D_{\rm minor}^{\rm ub}(A)
=
\min\{
w_{\rm minor}(S):
|S|=r+1,\ u_S\neq0,\ \text{and} \nonumber \\
&\,\text{the induced product class is verified nontrivial} \},
\end{align}
with the convention that the value is unavailable if no minor candidate passes verification.  Whenever this value is available,
\begin{align}
d\bigl(\mathrm{LP}(A,A^\dagger)\bigr)
&\le
D_{\rm minor}^{\rm ub}(A).
\end{align}
The same construction applied to $A^\dagger$ gives row-side distance witnesses.

For sparse-polynomial seeds, these determinants are often computable directly in the polynomial ring.  For monomial seeds, the determinant is a permanent-like sum of monomials, since signs vanish over $\F_2$.  This connects the seed-level LP distance witness to the classical QC-LDPC distance-bound literature, where polynomial parity-check matrices and their minors or permanents give explicit minimum-distance upper bounds~\cite{SmarandacheVontobel2012}.

\subsection{Proper-divisor quotient-lift witnesses}
\label{app:quotient-lift-witnesses}

There is another inexpensive source of distance witnesses that is not captured by the Cramer/minor spectrum.  Let $m$ be a proper divisor of $\ell$, and reduce the seed modulo $x^m+1$:
\begin{align}
A^{(m)}
&=
A \bmod (x^m+1)
\quad
\text{over }
R_m=\F_2[x]/(x^m+1).
\end{align}
Define
\begin{align}
P_{\ell,m}
&=
1+x^m+x^{2m}+\cdots+x^{\ell-m}
\in R_\ell .
\end{align}
If $\bar u\in\ker A^{(m)}$ and $\widetilde u$ is its coefficient lift to $R_\ell^n$ with exponents in $\{0,\ldots,m-1\}$, then
\begin{align}
u
&=
P_{\ell,m}\widetilde u
\end{align}
satisfies $Au=0$ in $R_\ell^r$.  Indeed, $A\widetilde u$ is divisible by $x^m+1$, and
\begin{align}
(x^m+1)P_{\ell,m}
&=
x^\ell+1
\end{align}
vanishes in $R_\ell$.

Thus a sparse quotient-kernel vector gives a candidate product representative $Z=(u\otimes q,0)$ as in Sec.~\ref{subsec:seed-distance-witnesses}.  With the canonical lift above,
\begin{align}
\wt_R(u)
&=
\frac{\ell}{m}\wt_{R_m}(\bar u).
\end{align}
The quotient-lift diagnostic $\Dquot$ is the minimum of this quantity over quotient-kernel candidates whose expanded binary product representative passes the verification convention of Sec.~\ref{subsec:seed-distance-witnesses}.  As with $\Dseed$ and $\Dminor$, no quotient-lift candidate enters $\Dbest$ unless the expanded binary vector is checked to commute with the stabilizers, to lie outside the opposite stabilizer row space, and to have the reported support weight.

This diagnostic can be tighter than the full-lift Cramer/minor spectrum because cancellation is first performed in the smaller quotient ring $R_m$ and then repeated periodically in the full lift.  In Table~\ref{tab:main}, the verified quotient-lift values include
\begin{align}
\Dquot\!\left(\seedlabel{M}{45}{4}{10}\right)
&=
18
\quad (m=15,\ \wt_{R_m}(\bar u)=6),
\nonumber \\
\Dquot\!\left(\seedlabel{M}{45}{5}{9}\right)
&=
50
\quad (m=9,\ \wt_{R_m}(\bar u)=10),
\nonumber \\
\Dquot\!\left(\seedlabel{P}{75}{3}{7}\right)
&=
25
\quad (m=3,\ \wt_{R_m}(\bar u)=1),
\nonumber \\
\Dquot\!\left(\seedlabel{P}{91}{3}{7}\right)
&=
70
\quad (m=13,\ \wt_{R_m}(\bar u)=10).
\label{eq:quotient-lift-values}
\end{align}
The first value is superseded in the parameter column by the tighter independent $\Dwit=16$ entry, while the last three sharpen the reported diagnostic bound for their rows.  For \seedlabel{P}{91}{3}{7}, the listed value comes from a quotient-kernel vector of weight $10$ at $m=13$; a weight-$11$ vector at the same divisor would give the weaker bound $91\cdot 11/13=77$.

\subsection{Cramer/minor spectra in the highlighted example rows}
\label{app:cramer-spectra}

The minor construction also explains, before any full-code distance search, one peculiarity of the highlighted example rows in Table~\ref{tab:main}: they do not show the toy-scale determinantal distance witnesses seen in several comparison seeds.  For an $r\times n$ seed, every $(r+1)$-column subset gives a canonical Cramer/minor candidate $u_S\in\ker A$ whose entries are the $r\times r$ minors in Eq.~\eqref{eq:minor-witness}.  The Cramer/minor spectrum is the multiset
\begin{align}
\mathcal C_{r+1}(A)
&=
\left\{
w_{\rm minor}(S):
|S|=r+1,\ u_S\neq0
\right\}.
\end{align}
When this multiset is nonempty, its minimum is an inexpensive seed-level diagnostic; when a spectrum element passes the nontriviality verification above, it contributes to the verified minor distance witness $\Dminor$.

For monomial seeds this spectrum has a simple analytic ceiling.

\begin{proposition}[Monomial minor ceiling]
\label{prop:cramer-ceiling}
Let $A_{ij}=x^{a_{ij}}$ be an $r\times n$ monomial seed over $R_\ell$.  For any $r$-column set $T=\{j_1,\ldots,j_r\}$,
\begin{align}
\det A_T
&=
\sum_{\pi\in S_r}
x^{\sum_{i=1}^r a_{i,j_{\pi(i)}}}
\in R_\ell,
\label{eq:monomial-minor-permanent}
\end{align}
since signs vanish over $\F_2$; hence $\wt_R(\det A_T)\le r!$, and every nonzero Cramer/minor candidate satisfies
\begin{align}
w_{\rm minor}(S)
&\le
(r+1)r!,
\label{eq:monomial-cramer-ceiling}
\end{align}
independently of the lift length $\ell$.
\end{proposition}

\noindent As a result, a monomial $3\times n$ seed has a built-in Cramer/minor ceiling of $24$ whenever one of these candidates verifies as a nontrivial LP logical.  This explains why several $3\times7$ monomial or reference rows naturally produce verified determinantal distance witnesses with weights in the range $18$--$24$.

The ceiling is an upper-bound mechanism, not a lower-bound theorem.  It says that if the canonical determinantal syzygies survive as nontrivial logicals, then monomial $3\times n$ seeds cannot avoid a small distance witness by increasing $\ell$.  To suppress this mechanism one must either change the seed shape or change the entry support so that the relevant minors are no longer six-term monomial permanents.

The same count extends to sparse-polynomial entries.  For an $r$-column set $T=\{j_1,\ldots,j_r\}$, define the raw determinant expansion count
\begin{align}
\nu(T)
&=
\sum_{\pi\in S_r}
\prod_{i=1}^r
\wt_R(A_{i,j_{\pi(i)}}).
\label{eq:raw-determinant-count}
\end{align}
Then $\wt_R(\det A_T)\le \nu(T)$, with strict inequality when exponent collisions cancel in pairs.  If all entries have support at most $s$, then
\begin{align}
\wt_R(\det A_T)
&\le
r!s^r,
\nonumber \\
w_{\rm minor}(S)
&\le
(r+1)r!s^r.
\label{eq:sparse-cramer-ceiling}
\end{align}
Sparse binomial entries enlarge the determinantal scale by a factor up to $2^r$ compared with monomial entries, at the cost of increasing check weight.  This is a seed-level design tradeoff rather than a distance theorem: larger raw expansion counts make the standard Cramer distance witnesses less automatically small, but the resulting vectors must still be verified and may still have cancellations or be superseded by other syzygies.

A useful scale heuristic is obtained by treating the $\nu(T)$ raw monomial exponents in a determinant as independent uniform residues in $\Z_\ell$.  If $N=\nu(T)$, then a fixed residue has odd parity with probability
\begin{align}
p_{\rm odd}(N,\ell)
&=
\frac{1}{2}
\left(
1-\left(1-\frac{2}{\ell}\right)^N
\right),
\end{align}
so
\begin{align}
\mathbb E\,\wt_R(\det A_T)
&\approx
h_\ell(N)
:=
\frac{\ell}{2}
\left(
1-\left(1-\frac{2}{\ell}\right)^N
\right).
\label{eq:minor-occupancy-heuristic}
\end{align}
The function $h_\ell(N)$ grows approximately linearly for $N\ll \ell$ and saturates near $\ell/2$ for $N\gg \ell$.  The independence assumption is not a theorem: the exponent sums come from permutations of the same seed matrix and can be correlated.  We use $h_\ell(N)$ only as a scale estimate for support after parity cancellations, not as evidence for a distance claim.  Simply adding more raw terms eventually stops increasing a typical minor support; the useful search objective is to raise the lower tail of the Cramer spectrum while keeping the seed check weight and basis-width certificate acceptable.

The highlighted example rows use both effects.  For \seedlabel{M}{45}{5}{9}, the relevant $5\times5$ monomial minors have $5!$ possible permutation terms.  After reduction modulo $x^{45}+1$ and parity cancellation, the light minors in this seed retain support weights between $16$ and $24$, so the lightest six-component Cramer/minor candidate has total seed weight $114$.  A crude random-permanent heuristic gives the same scale: if the $r!$ permutation sums in Eq.~\eqref{eq:monomial-minor-permanent} are treated as independent uniform residues in $\Z_\ell$, then
\begin{align}
\mathbb E\,\wt_R(\det A_T)
&\approx
h_\ell(r!).
\end{align}
For $r=5$ and $\ell=45$ this is close to $22$ per minor, consistent with a six-component Cramer weight of order $10^2$ rather than the $3\times n$ ceiling of $24$.

For \seedlabel{P}{75}{3}{7} and \seedlabel{P}{91}{3}{7}, the sparse parity-systematic binomial entries both repair the $x+1$ rank profile and inflate the nontrivial-packet minors beyond the six-term monomial pattern.  The lightest verified Cramer/minor distance witnesses have weights $78$ and $81$, respectively.  Table~\ref{tab:main} records the smaller verified quotient-lift witnesses $\Dquot=25$ and $\Dquot=70$ for these rows, so the values $78$ and $81$ should be read only as Cramer/minor-spectrum data.  These values show that these rows suppress the standard low-weight determinantal syzygies at the seed level, rather than merely hiding them from the spectral basis; they remain upper-bound diagnostics, not distance proofs.

In this restricted sense, the Cramer spectrum predicts why \seedlabel{M}{45}{5}{9}, \seedlabel{P}{75}{3}{7}, and \seedlabel{P}{91}{3}{7} look different from the low-bound comparison rows before running an expanded-code distance search.  It predicts the absence of the standard duplicate-column and small-minor mechanisms, not the absence of all low-weight logical operators; the quotient-lift witnesses above are precisely examples of a different mechanism.  We therefore treat lower-tail Cramer spectra, verified minor distance witnesses, quotient-lift witnesses, and bounded syzygy searches as separate diagnostic components in broader code-search studies.

\subsection{Sparse syzygy and seed-code searches}

Minor distance witnesses only use $r+1$ columns.  More generally, one can search for bounded sparse polynomial syzygies.  If sparse polynomials $p_j(x)$ supported on a small column set $S$ satisfy
\begin{align}
\sum_{j\in S}
p_j(x)A_{\cdot,j}
&=
0,
\end{align}
then
\begin{align}
u
&=
\sum_{j\in S}
p_j(x)\mathbf e_j
\end{align}
lies in $\ker A$, and
\begin{align}
\wt_R(u)
&\le
\sum_{j\in S}
\wt(p_j).
\end{align}
When the induced product class is verified as nonzero, this gives
\begin{align}
d\bigl(\mathrm{LP}(A,A^\dagger)\bigr)
&\le
\sum_{j\in S}
\wt(p_j).
\end{align}
The unrestricted sparsest-syzygy problem is a classical minimum-distance problem, so the useful prefilter is the bounded search with a small support budget $w_0$.
The notation $\Dsyz$ denotes the minimum verified distance witness found by such bounded syzygy searches.

One can also run a classical low-weight codeword search on the binary expansion $A_{\rm bin}:\F_2^{n\ell}\to\F_2^{r\ell}$ and on $A_{\rm bin}^\transpose$.  Any low-weight word found in $\ker A_{\rm bin}$ or $\ker A_{\rm bin}^\transpose$ gives an LP distance witness after verification.  This search is still exponential in the target weight in the worst case, but it is performed on length $n\ell$ or $r\ell$, not on the full self-adjoint LP length $(n^2+r^2)\ell$; this is consistent with the finite-length LP distance-diagnostic viewpoint of Ref.~\cite{Raveendran2025}.

\subsection{Structural and decoder scores}

The remaining prefilters are structural scores rather than standalone distance statements.  Packet ranks determine logical dimensions, but they do not measure the minimum support of packet-local kernel vectors.  A useful refinement is the packet dependency size
\begin{align}
c_g(A)
&=
\min
\left\{
|S|:
\rank_{K_g} A_{g,S}<|S|
\right\},
\end{align}
possibly computed only up to a small support cap.  A minimal set $S$ attaining such a dependence is a minimal dependent column set (a circuit of the column matroid of $A_g$, in the matroid-theoretic rather than quantum-circuit sense).  Combined with the idempotent-support profile
\begin{align}
\mathcal W_g
&=
\left\{
\wt_R(e_g\beta):
\beta\in K_g^\times
\right\},
\end{align}
this indicates which packets are likely to produce physically light representatives after lifting.

For monomial seeds, exponent-vector collisions and short-cycle counts give fast graph-quality filters.  The duplicate-column test is the smallest collision; higher additive coincidences can be tallied as well, while Tanner-graph defects are summarized by the cycle spectrum $N_4,N_6,N_8,\ldots$.  These diagnostics are relevant to BP-type decoders and are complementary to the packet rank profile; the underlying exponent-matrix and girth tests are standard QC-LDPC tools~\cite{Fossorier2004,2021necessarysufficientgirthconditions}.

Whole-orbit erasures give an exact structured-erasure score, motivated by the role of erasure structure in QLDPC decoding~\cite{Yao2024}.  For a block subset $S$, let $\Delta(S)=d_Z(E_S)+d_X(E_S)$, where $E_S$ erases the full lift-group orbit in each selected block coordinate.  The first dangerous orbit-erasure size is
\begin{align}
s_{\rm fail}&=\min\{|S|:\Delta(S)>0\},
\end{align}
and $\rho_{s\rm orb}$ (\secref{sec:erasure}) is the fraction of size-$s$ block subsets with $\Delta(S)>0$.  Because whole-orbit erasures decompose by packets, the same calculation identifies which packets support the erased logicals.

Table~\ref{tab:erasure} collects the exact orbit-erasure scores summarized in Fig.~\ref{fig:codesignlandscape}.  The model erases one or two whole cyclic block-orbits from the left $n_{\rm seed}^2$ block summand (\secref{sec:erasure}), and the final column gives the dominant two-orbit packet by erased-logical dimension contribution.  Distance-witness rows are shown only to illustrate failure modes; they are not treated as construction examples.

\begin{table*}[t]
\caption{Selected exact orbit-erasure diagnostics used in the code-design comparison. The quantities $\rho_{1\rm orb}$ and $\rho_{2\rm orb}$ are the $s=1$ and $s=2$ instances of $\rho_{s\rm orb}$: the exact fractions of one- and two-left-block-orbit erasure patterns with nonzero erased-logical dimension. The dominant packet is the packet contributing the largest share of erased-logical dimension for two-left-block-orbit erasures; the number in parentheses is that share.}
\label{tab:erasure}
\begin{tabular}{llccccc}
\toprule
Seed & role & $\Wcert$ & $\Dbest$ & $\rho_{1\rm orb}$ & $\rho_{2\rm orb}$ & dominant packet\\
\midrule
\seedlabel{A}{45}{3}{7} & example & 120 & 16 & 0 & 0.250 & $x+1$ (0.778)\\
\seedlabel{M}{45}{4}{10} & example & 126 & 16 & 0 & 0.182 & $x+1$ (0.957)\\
\seedlabel{M}{45}{5}{9} & example & 150 & 50 & 0 & 0.200 & $x+1$ (1.000)\\
\seedlabel{R}{75}{3}{7} & example & 160 & 16 & 0 & 0.250 & $x+1$ (0.913)\\
\seedlabel{P}{75}{3}{7} & example & 180 & 25 & 0.449 & 0.707 & $x+1$ (0.529)\\
\seedlabel{A}{75}{3}{7} & example & 200 & 20 & 0 & 0.250 & $x+1$ (0.778)\\
\seedlabel{R}{81}{3}{7} & example & 162 & 16 & 0 & 0.250 & $x+1$ (0.563)\\
\seedlabel{P}{91}{3}{7} & example & 208 & 70 & 0.327 & 0.554 & $x+1$ (0.987)\\
\seedlabel{A}{91}{3}{7} & example & 224 & 22 & 0 & 0.250 & $x+1$ (1.000)\\
random $\ell=31$ scan point & filtered random & 64 & 48 & 0 & 0.250 & $x+1$ (1.000)\\
\midrule
\seedlabel{T}{15}{2}{4} & distance witness & 30 & 2 & 0 & 0.400 & non-$x+1$ (0.302)\\
\bottomrule
\end{tabular}
\end{table*}

\subsection{Computability and diagnostic summary} \label{app:score}

Most of the filters above are efficient in the intended search regime: bounded seed shape, bounded group-basis support, and bounded distance-witness size.  They should be read as bounded algebraic tests rather than polynomial-time solutions to distance.  Shifted-duplicate filters use hashed exponent differences, minor filters examine $\binom{n}{r+1}$ column sets for fixed seed shape, quotient-lift filters fold the seed to proper divisors of $\ell$, packet ranks use finite-field row reduction over packets with total degree $\sum_g\deg g=\ell$, packet-dependency searches are combinatorial only in the imposed support cap, and whole-orbit erasure enumeration is exact for the small block sizes normally used in erasure-pattern scans.

The search procedure used here is: $(i)$ seed normalization, $(ii)$ packet rank and basis-certificate computation, $(iii)$ verified seed-level distance-witness search, $(iv)$ graph and erasure diagnostics, and finally $(v)$ QDist/BP-OSD on selected seeds.  The first four stages are algebraic and combinatorial; the last stage is a full-code distance or decoder test.
The final stage is deliberately separated from the algebraic filters: BP-OSD and iterative decoder studies provide natural late-stage evaluation tools for QLDPC, HGP, and LP codes~\cite{PanteleevKalachev2021,Roffe2023,Pradhan2025,GolowichGuruswami2024}.
The diagnostic output can be summarized as
\begin{align}
\text{seed}
&\mapsto
\big(k,\Wcert,\Dbasis,\Dseed,\mathcal C_{r+1},\Dminor,\Dquot,\nonumber \\
&\hspace{2.8em}
\Dsyz,\rho_{s\rm orb},N_4,N_6,\ldots\big).
\end{align}
Here $\Dseed$ denotes the best verified direct seed-kernel distance witness, such as a shifted duplicate row or column.  The separate entries $\Dminor$, $\Dquot$, and $\Dsyz$ refer to verified minor/permanent, proper-divisor quotient-lift, and bounded-syzygy distance witnesses.  If any verified distance-witness value is smaller than $\Dbasis$, then the parameter line should use the sharper distance upper bound, while $\Dbasis$ and $\Wcert$ remain useful as properties of the constructed addressable basis.

The entries draw on standard QC-LDPC exponent-matrix, permanent/minor, girth, and decoder diagnostics, while the packet attribution and basis-width columns are supplied by the present spectral construction.  Table~\ref{tab:seed-diagnostic-summary} is therefore a compact summary for follow-up search rather than a substitute for full distance or decoder simulations.

\bibliography{ref}

\end{document}